\documentclass[11pt]{article}
\usepackage[square,authoryear]{natbib}
\usepackage{marsden_article}
\usepackage[all]{xy}
\usepackage{amscd}
\usepackage{amssymb}           
\usepackage{mathtools,mathrsfs}
\usepackage{framed}
\newtheorem{remark}[theorem]{Remark}
\usepackage{eucal}
\usepackage{url}
\usepackage{ulem}
\usepackage{framed}

\newcommand{\deltabar}{{\mkern0.75mu\mathchar '26\mkern -9.75mu\delta }}

\begin{document}

\title{General relativistic Lagrangian continuum theories\\
\medskip 
\Large Part I: reduced variational principles and junction conditions for hydrodynamics and elasticity}
\author{Fran\c{c}ois Gay-Balmaz\footnote{Division of Mathematical Sciences, Nanyang Technological University, 21 Nanyang Link, Singapore 637371; email: \url{francois.gb@ntu.edu.sg}}\;\;\footnote{CNRS \& LMD, Ecole Normale Sup\'erieure, \url{gaybalma@lmd.ens.fr}.\;\;Data sharing not applicable to this article as no datasets were generated or analysed during the current study.}}
\date{}
\maketitle

\begin{abstract}


We establish a Lagrangian variational framework for general relativistic continuum theories that permits the development of the process of Lagrangian reduction by symmetry in the relativistic context. Starting with a continuum version of the Hamilton principle for the relativistic particle, we deduce two classes of reduced variational principles that are associated to either spacetime covariance, which is an axiom of the continuum theory, or material covariance, which is related to particular properties of the system such as isotropy. The covariance hypotheses and the Lagrangian reduction process are efficiently formulated by making explicit the dependence of the theory on given material and spacetime tensor fields that are transported by the world-tube of the continuum via the push-forward and pull-back operations. It is shown that the variational formulation, when augmented with the Gibbons-Hawking-York (GHY) boundary terms, also yields the Israel-Darmois junction conditions between the solution at the interior of the relativistic continua and the solution describing the gravity field produced outside from it.
The expression of the first variation of the GHY term with respect to the hypersurface involves some extensions of previous results that we also derive in the paper. We consider in details the application of the variational framework to relativistic fluids and relativistic elasticity. For the latter case, our setting also allows to clarify the relation between formulations of relativistic elasticity based on the relativistic right Cauchy-Green tensor or on the relativistic Cauchy deformation tensor. The setting developed here will be further exploited for modelling purpose in subsequent parts of the paper.

\end{abstract}

\tableofcontents



\section{Introduction}

General relativistic continuum theories are essential in the understanding of the structure and evolution of astrophysical systems. Phenomena of high interest which request these theories include gravitational collapse, supernova explosion, galaxy formation, accretion onto a neutron star or a black hole, as well as the emission of gravitational waves. While the use of relativistic fluid models is well-established, relativistic elasticity models are also believed to be of high relevance for instance for the physics of neutron star crusts.

As we will review below, variational principles have played a main role in the derivation of relativistic continuum models, and still form today an essential modelling tool in this area on the theoretical side. In general, for mechanical systems in absence of irreversible process, should they be relativistic or not, the most elegant and physically justified variational principle is the Hamilton principle. This principle is free from any constraints on the field variations, and can be constructed directly from the knowledge of the Lagrangian density of the system. It however requires to work in the material (or Lagrangian) representation, which corresponds to the description of each individual material particle. Depending on the problem at hands, this may not be the most practical description, one reason being that the equations in material coordinates often take a complicated form compared to their Eulerian (spatial) or convective (body) versions. It is thus desirable to formulate the variational principle directly in terms of the variables associated to those representations. Several types of Eulerian variational formulations have been developed for relativistic fluids and solids, with various types of constraints and from various point of views, see below. This makes hard the study of the relation between them and their derivation from a unified point of view for both general relativistic fluids and solids.

The most efficient and justified approach to obtain in full generality the variational formulation for relativistic continua in such representations would be to systematically \textit{deduce} it from a continuum version of the Hamilton principle for the relativistic particle by making use of the symmetries of the system. These symmetries are associated to the spacetime covariance, which is a requested physical condition, and also to particular properties of the system like isotropy or homogeneity. 
Such a systematic approach to rigorously derive relevant variational formulations from the classical Hamilton principle by symmetries is well-known in Newtonian continuum mechanics and fits with the general process of \textit{Lagrangian reduction by symmetry} that we will review later. 
It is therefore highly desirable to develop an analogue process for the general relativistic case, that also yields a reliable approach for the derivation of useful extensions of those theories that will be explored in subsequent parts of this work.

\paragraph{Goal of the paper.}
The first purpose of this paper is to establish a Lagrangian variational framework for general relativistic continuum theories that systematically parallels the Lagrangian reduction approach mentioned above for the Newtonian case. In particular, our framework allows to rigorously deduce several useful variational formulations by starting from the most natural and physically justified principle, namely, a \textit{continuum version of the Hamilton principle for the relativistic particle}, as illustrated below:
\begin{align*} 
&\delta \int_{ \lambda _0}^{ \lambda _1} \!\! - \sqrt{- \mathsf{g} \big( \dot x , \dot x\big) }  \; m_0c  \; d \lambda=0\; \;\leadsto \;\; \delta \int_{ \lambda _0}^{ \lambda _1}\! \!\!\int_ \mathcal{B}\! -\frac{1}{c}  \sqrt{ -\mathsf{g}  \big(\dot \Phi , \dot \Phi  \big) }\; \big(c^2 +  \mathcal{W}\big) \; d \lambda \wedge R_0 =0.
\end{align*} 
The left hand side represents the Hamilton principle for a relativistic particle of mass $m_0$ with world-line $x( \lambda )$ in the spacetime $(\mathcal{M} ,\mathsf{g})$, where $\dot x= dx/d \lambda $. On the right hand side we write its continuum version which plays a fundamental role in this paper: the map $ \Phi ( \lambda , \mathsf{X}) \in \mathcal{M} $, with $\mathsf{X} \in \mathcal{B} $ and $\dot \Phi = d \Phi / d \lambda $, denotes the world-tube of a relativistic continuum media parameterized with the material manifold $ \mathcal{B} $, the volume form $R_0$ on $ \mathcal{B} $ is the mass form, and $\mathcal{W}$ is an internal energy expression which depends on the material and spacetime tensor fields occurring in the theory.

From this principle, and exactly as in the nonrelativistic case, symmetries of the Lagrangian density can be used to systematically deduce two main types of variational formulations yielding the relativistic version of the convective (body) and spatial (Eulerian) variational formulations. The two classes of symmetries that we consider are spacetime covariance, which is an axiom of the continuum theory, and material covariance, which is related to specific properties of the continuum. 

In order to develop a systematic relativistic analogue to the ``reduction by symmetry" approach of Newtonian continua, we make explicit the dependence of the theory on tensor fields given on the parametrization manifold and on the spacetime. This implies, for instance, the inclusion of reference tensor fields associated to the mass density and to the world-velocity. 
\textcolor{black}{In this work, we do not include dissipative mechanisms such as thermal or viscous effects.}

\medskip 

The second purpose of this paper concerns the junction or matching conditions between the solution at the interior of the relativistic continua and the solution describing the gravity field produced outside from it. The matching conditions between two spacetimes is a fundamental issue in general relativity which has been considered in classic works by \cite{Da1927}, \cite{Li1955}, \cite{OBSy1952}, \cite{Is1966}, and \cite{BoVi1981}. The Israel-Darmois junction conditions state that the interior and exterior Lorentzian metrics induce the same metric on the boundary hypersurface (assumed in this work to consist of spacelike and timelike pieces) and that, in absence of singular matter distribution, define on it the same extrinsic curvature (second fundamental form). These conditions have important applications in several gravitational topics such as wormholes (\cite{Vi1989a,Vi1989b}, \cite{LoSiVi2020}), interior structures of black-holes (\cite{BaFr1996}), gravitational collapse (\cite{FaJeLlSe1992}, \cite{FaSeTo1996}, \cite{Mu2021}), or bubble dynamics in the early Universe (\cite{BeKuTk1987}, \cite{BlGuGu1987}, \cite{DeVi2017}, \cite{De2020}), to name only a few works.

We shall show how these junction conditions, as well as their implication on the boundary conditions on the stress-energy-momentum tensor, can be directly obtained from a natural extension of the variational formulation by including the Gibbons-Hawking-York (GHY) term associated to the boundary of the relativistic continuum. This boundary term, given by the integral of the trace of the extrinsic curvature, has been introduced in \cite{Yo1972}, \cite{GiHa1977} and makes the gravitational variational principle well-defined on spacetimes with boundaries. While the variation of the GHY term with respect to the metric is well-known, its variation with respect to the hypersurface involves some extensions of previous results on the first variation of integrated mean curvature that we also derive in the paper.





\paragraph{Variational principles for relativistic continua.} Variational principles for general relativistic continua were first given by \cite{Ta1954} for the case of the perfect fluid.  Taub's variational principle is based on varying the particles' world-lines and the metric and gives both the field equations for the gravitational field created by the fluid and the equations of motion of the fluid in this gravitational field. This variational approach was then further developed and applied to several models occurring in special relativistic elasticity (\cite{Gr1971}, \cite{MaEr1972a,Ma1972a,MaEr1972b}), general relativistic elasticity (\cite{Ma1971}, \cite{Ca1973}), and relativistic fluids (\cite{CaKh1992}, \cite{CaLa1995,CaLa1998}).
This type of variational formulations forms today an essential tool for modelling in general relativistic continuum theories, as shown for instance in the recent works \cite{GaAnHa2020}, \cite{BaWa2020}, \cite{FeCa2020}, \cite{AdCo2021}, \cite{Ad2021}.

While these approaches are based on varying the world-lines of the continuum, variational principles have also been proposed based on varying the projection map which assigns to each point of the continuum in spacetime its material label (\cite{KiMa1992}, \cite{BeSc2003}).
Other types of variational principles involve arbitrary Eulerian variations, rather than arbitrary world-lines variations, but require the introduction of Clebsch-type auxiliary fields or Lagrange multiplier terms in the action functional (\cite{Sc1870},  \cite{Ma1972b}, \cite{KhLe1982}, \cite{LeKh1982}). Approaches with auxiliary fields are however harder to use for modelling purpose and lack physical justification.

The general Lagrangian variational framework that we develop in this paper encompasses the approach of \cite{Ta1954} and subsequent works mentioned above on relativistic fluid and elasticity by recasting them in a systematic and unified Lagrangian covariant reduction approach. 
It contains however essential differences from these approaches. In order to make a clear distinction between the material, spacetime, and body descriptions, and to rigorously relate the processes of covariant reduction to the symmetries of Lagrangian density, we introduce the appropriate material tensor fields on the reference manifolds, so that the corresponding Eulerian quantities are given by the push-forward with the world-tube $ \Phi $. For instance, we introduce the material mass density $R$, the material entropy density $S$, as well as a material vector $W$, so that the corresponding Eulerian mass density, entropy density, and world-velocity are written as $ \varrho = \Phi _*R$, $ \varsigma = \Phi _*S$, $ w= \Phi _* W$. Making explicit the dependence of the Lagrangian density on such tensor fields, rather than quantities such as the proper mass density or the matter current that are constructed from them and from the Lorentzian metric, is crucial for the development of the Lagrangian covariant reduction process. It allows for a general definition of material covariance in relation with the isotropy of the continuum, which permits the construction of a fully Eulerian variational formulation by covariant reduction. This variational formulation is shown to be also available for anisotropic continua at the expense of the inclusion of additional material tensor fields. Finally, the present approach also makes the variational formulation free from any Lagrange multiplier terms such as the constant magnitude constraint for the world-velocity in \cite{Ta1954}, and allows to deduce all the variational formulations from the continuum version of the Hamilton principle for the relativistic particle. In addition, since we are not only considering the variations of the metric, but also of the continuum configuration, our variational approach is directly applicable to continuum mechanics in a fixed gravitational field and to special relativity.
The approach developed here can be seen as the covariant relativistic version of the process of Lagrangian reduction by symmetry in free boundary Newtonian continuum mechanics developed in \cite{GBMaRa2012} that we will review in \S\ref{Review_Newtonian}. For Newtonian fluids on fixed domains, such techniques reduce to the Euler-Poincar\'e reduction on diffeomorphism groups \cite{HoMaRa1998} which itself has its roots in the Lie group geometric formulation of the ideal fluid due to \cite{Ar1966}.
We refer to Tables \ref{comparison} and \ref{comparison_equations} for a schematic correspondence between the relativistic and nonrelativistic Lagrangian reduction processes.



\section{Reduced variational formulations of Newtonian continuum theories}\label{Review_Newtonian}

Let us consider a non relativistic continuum body (fluid or solid) whose configuration at each time can be described by an embedding $ \varphi : \mathcal{B} \rightarrow \mathcal{S} $. Here $ \mathcal{S} = \mathbb{R} ^n $ is the ambient space in which the body evolves and $ \mathcal{B} $ is the reference configuration of the body, given by a compact $n$-dimensional submanifold of $\mathbb{R} ^n $ with smooth boundary. A motion of the body is described by a curve $ t \in [t_0,t_1]\mapsto  \varphi _t\in \operatorname{Emb}( \mathcal{B} , \mathcal{S} )$ in the manifold of all embeddings of $ \mathcal{B} $ into $ \mathcal{S} $. Given a motion $\varphi _t $, the location at time $t$ of the point of the continuum with label $\mathsf{X} \in \mathcal{B} $ is given by $x= \varphi _t(\mathsf{X})= \varphi (t, \mathsf{X}) \in \mathcal{S}$.

We review below the symmetry reduced variational formulations for free boundary continuum theories by following \cite{GBMaRa2012}. These variational formulations are systematically deduced from the Hamilton principle by exploiting the symmetries of the systems. These symmetries are best explained by explicitly emphasizing the dependence of the theory on given fixed tensor fields on the reference configuration $ \mathcal{B} $ and on the ambient space $ \mathcal{S}$, see \cite{MaHu1983,SiMaKr1988,GBMaRa2012}. We will restrict our approach to continuum theories depending on two volume forms $R, S \in \Omega ^n ( \mathcal{B} )$ representing the mass and entropy densities in the reference configuration, and on two Riemannian metrics $G \in S^2_+( \mathcal{B} )$ and $g \in S^2_+( \mathcal{S} )$ on $ \mathcal{B} $ and $ \mathcal{S} $, respectively. By adopting this setting, the relevant dynamic fields for the description of the continuum in the symmetry reduced formulations, besides the velocity, are obtained by applying the pull-back or push-forward operations associated with the configuration map $ \varphi $. For instance the Eulerian mass density $ \varrho $, the Eulerian entropy density $ \varsigma $, and the Cauchy deformation tensor $\mathsf{c}$ are obtained as
\begin{equation}\label{PF_Newtonian} 
\varrho = \varphi _* R, \qquad \varsigma = \varphi _* S, \qquad \mathsf{c}= \varphi _* G,
\end{equation} 
while the right Cauchy-Green tensor is
\begin{equation}\label{PB_Newtonian} 
C= \varphi ^* g.
\end{equation} 
We will also describe the mass and entropy density by using the functions $ \rho $ and $s$ such that 
\begin{equation}\label{rho_s}
\rho  \mu (g)= \varrho \quad\text{and}\quad s \mu (g)= \varsigma,
\end{equation} 
with $ \mu (g)$ the volume form associated to $g$ on $ \mathcal{S} $.
The specific entropy defined by
\begin{equation}\label{specific_s}
\eta = s/ \rho
\end{equation} 
will also be used.
For simplicity, we shall consider only zero traction boundary conditions and ignore surface tension effects. The variational treatment can be extended to this case, see \cite{GBMaRa2012}.

\subsection{Material description and the Euler-Lagrange equations}

Given the fields $R,S,G,g$ as above, we consider a class of continuum bodies described by Lagrangian functions $L:T \operatorname{Emb}( \mathcal{B} , \mathcal{S} ) \times  \Omega ^n ( \mathcal{B} ) \times \Omega ^n ( \mathcal{B} ) \times S^2_+( \mathcal{B} ) \times S^2_+( \mathcal{S} ) \rightarrow \mathbb{R} $ of the form
\begin{equation}\label{General_L} 
L( \varphi , \dot \varphi , R, S, G, g \circ \varphi ) = \int_ \mathcal{B} \mathscr{L} ( \dot \varphi , T \varphi  , R, S, G, g \circ \varphi) ,
\end{equation} 
with Lagrangian density
\begin{equation}\label{General_LD} 
\mathscr{L} ( \dot \varphi , T \varphi  , R, S, G, g \circ \varphi) = \frac{1}{2} g( \dot \varphi , \dot \varphi ) R - \mathscr{W}\! \left( T \varphi , R, S, G, g \circ \varphi \right)  R.
\end{equation} 
In \eqref{General_L}, $T \varphi :T \mathcal{B} \rightarrow T \mathcal{S} $ denotes the tangent map (derivative) to the configuration map $ \varphi \in \operatorname{Emb}( \mathcal{B} , \mathcal{S} )$, also known as the deformation gradient, and $V=\dot \varphi \in T_ \varphi \operatorname{Emb}( \mathcal{B} , \mathcal{S} )$ is the material velocity of the continua with $ T_ \varphi \operatorname{Emb}( \mathcal{B} , \mathcal{S} ):= \{ \textcolor{black}{  V} \in C^\infty(\mathcal{B}, T\mathcal{S}) \mid \textcolor{black}{  V(\mathsf{X})} \in T_{ \varphi (\mathsf{X})} \mathcal{S}  \}$ the tangent space to $\operatorname{Emb}( \mathcal{B} , \mathcal{S} )$ at $ \varphi $.
The first term in \eqref{General_LD} represents the kinetic energy density and the second term is (minus) the total stored energy density, with $\mathscr{W}$ an arbitrary function of the point values of $T \varphi $, $R$, $S$, $G$, $g \circ \varphi $. Expressions \eqref{General_L} and \eqref{General_LD} are referred to as the Lagrangian functions and Lagrangian densities in the material representation. \textcolor{black}{The formulation with infinite dimensional manifolds such as the manifold of embeddings will not be essential for the approach that we take in the paper. It is mentioned here only to show how this Lagrangian point of view fits within the standard Lagrangian reduction by symmetry, see Remark \ref{rmk_LRBS}.}

In this representation, the equations of motion are obtained by the Hamilton principle
\begin{equation}\label{Ham_principle_Newton} 
\left. \frac{d}{d\varepsilon}\right|_{\varepsilon=0} \int_{t_0}^{t_1}L( \varphi_ \varepsilon  , \dot \varphi _ \varepsilon , R, S, G, g \circ \varphi _ \varepsilon ) dt=0
\end{equation} 
for arbitrary variations $ \varphi _ \varepsilon $ of the curve $ \varphi (t) \in \operatorname{Emb}( \mathcal{B} , \mathcal{S} )$ with fixed endpoints \textcolor{black}{at $t=t_0,t_1$}, exactly as in classical (finite dimensional) Lagrangian mechanics. Note that $R, S, G, g$ are held fixed in this variational principle, i.e, they appear as parameters in this formulation.

\paragraph{Material frame indifference.} From the axiom of material frame indifference, the function $\mathscr{W}$ in \eqref{General_LD} must satisfy 
\begin{equation}\label{mat_frame_indiff} 
\mathscr{W}\!(T ( \psi \circ \varphi ) , R, S, G, \psi _* g \circ \psi \circ \varphi ) =\mathscr{W}\! \left( T \varphi , R, S, G, g \circ \varphi \right),
\end{equation}
for all spatial diffeomorphisms $ \psi  \in \operatorname{Diff}( \mathcal{S} )$. This implies that $\mathscr{W}$ depends on $T \varphi $ and $g$ only through the right Cauchy-Green tensor $ C= \varphi ^* g$, i.e., $\mathscr{W}$ depends only on the point values of $R,S,G,C$. As a consequence, there is a function $ \mathcal{W}$, the convective stored energy function, such that
\[
\mathscr{W}\! \left( T \varphi , R, S, G, g \circ \varphi \right)= \mathcal{W} \left( \frac{R}{ \mu (C)} , \frac{S}{R}, G,C \right).
\]
Note that we decided to choose the first two arguments of $ \mathcal{W}$ as the combinations $\frac{R}{ \mu (C)} $ and $\frac{S}{R}$ to ease the comparison with the relativistic case later, and because it is in this form that the concrete expression for $\mathcal{W}$ are usually given. Of course other forms of the function $ \mathcal{W}$, such as $\mathcal{W}(R,S,G,C)$, are possible as long as there is a bijective correspondence between the representations.

\paragraph{Material covariance.} By definition, the function $\mathscr{W}$ satisfies the property of material covariance if the equality
\begin{equation}\label{mat_cov} 
\mathscr{W}(T (\varphi \circ \psi ) , \psi ^* R, \psi ^* S, \psi ^* G, g \circ \varphi \circ \psi ) =\mathscr{W} \left( T \varphi , R, S, G, g \circ \varphi \right) \circ \psi 
\end{equation}
holds for all body diffeomorphisms $ \psi \in \operatorname{Diff}( \mathcal{B} )$. Assuming that material frame indifference in \eqref{mat_frame_indiff} is satisfied, condition \eqref{mat_cov} can be equivalently written in terms of $ \mathcal{W}$ as
\[
\mathcal{W} \left( \frac{R}{ \mu (C)}   \circ \psi , \frac{S}{R}  \circ \psi , \psi ^* G, \psi ^* C \right)  = \mathcal{W} \left(  \frac{R}{ \mu (C)}   , \frac{S}{R} , G, C \right)  \circ \psi,
\]
for all $ \psi \in \operatorname{Diff}( \mathcal{B} )$. This condition on $ \mathcal{W}$ means that the continuum is homogeneous isotropic, see \cite[\S3.5]{MaHu1983}. It allows the definition of a function $ \varpi$, the Eulerian stored energy function, such that
\[
\varpi( \rho  , \eta , \mathsf{c},g) = \mathcal{W}(\rho  \circ \varphi , \eta \circ \varphi ,  \varphi ^* \mathsf{c}, \varphi ^* g) \circ \varphi ^{-1} ,
\]
with $\mathsf{c}= \varphi _* G$ the Cauchy deformation tensor.

\subsection{
Reduced variational principle in the convective frame}

By assuming the condition \eqref{mat_frame_indiff} on $\mathscr{W}$, the Lagrangian density $\mathscr{L}$ can be written in terms of $R$,  $S$, $G$, and $C= \varphi ^* g$ as
\begin{align*} 
\mathscr{L} ( \dot \varphi , T \varphi  , R, S, G, g \circ \varphi)&= \frac{1}{2} C( \mathcal{V} , \mathcal{V} ) R - \mathcal{W} \left( \frac{R}{ \mu (C)} , \frac{S}{R}, G,C \right)R\\
&=: \mathcal{L} ( \mathcal{V} , R,S,G,C),
\end{align*} 
where we defined the convective velocity $ \mathcal{V} = T \varphi  ^{-1} \circ\dot  \varphi \in \mathfrak{X} ( \mathcal{B} )$. Therefore, $\mathscr{L}$  has an associated convective Lagrangian density $ \mathcal{L} $ as defined above.

In this case, the Hamilton principle \eqref{Ham_principle_Newton} can be shown to induce the following convective variational formulation in terms of $ \mathcal{L} $:
\[
\delta\int_{t_0}^{t_1}\! \int_ \mathcal{B}  \mathcal{L} ( \mathcal{V} , R,S,G,C) dt=0, \quad\text{for variations}\quad  \delta \mathcal{V}= \partial _t \zeta -[ \mathcal{V}, \zeta ], \quad  \delta C= \pounds _ \zeta C,
\] 
where \textcolor{black}{$ \zeta$ is an} arbitrary time dependent vector field on $ \mathcal{B} $, vanishing at $t=t_0,t_1$, $[ \mathcal{V} , \zeta ]$ is the Lie bracket of vector fields, \textcolor{black}{and $ \pounds _ \zeta C$ is the Lie derivative of the right Cauchy-Green tensor along $ \zeta $}, see Fig. \ref{figure_1} on the left. \textcolor{black}{These variations are obtained from the relations $\mathcal{V} _ \varepsilon  = T\varphi _ \varepsilon  ^{-1} \circ \dot  \varphi_ \varepsilon  $, $C_ \varepsilon  = \varphi _ \varepsilon  ^* G$ and by defining $\zeta =T \varphi  ^{-1} \circ \delta \varphi $}. A direct application of this principle yields the equations in convective representation as
\begin{equation}\label{reduced_EL_conv} 
\frac{\partial }{\partial t} \frac{\partial \mathcal{L} }{\partial \mathcal{V}} - \pounds _ \mathcal{V} \frac{\partial \mathcal{L} }{\partial \mathcal{V} } =  \frac{\partial \mathcal{L} }{\partial C}: \nabla C - 2 \operatorname{div} ^ \nabla  \left( C \cdot \frac{\partial \mathcal{L} }{\partial C}\right), 
\end{equation} 
together with the boundary condition
\begin{equation}\label{BC_EL_conv} 
\left( \mathcal{V} \otimes   \frac{\partial \mathcal{L} }{\partial \mathcal{V} } - 2 C \cdot \frac{\partial \mathcal{L} }{\partial C  }   \right)( \cdot  , N ^\flat )=0\quad \text{on} \quad \partial \mathcal{B},
\end{equation} 
see \cite{GBMaRa2012} for a detailed derivation. In \eqref{reduced_EL_conv}, $\pounds _ \mathcal{V} \frac{\partial \mathcal{L} }{\partial \mathcal{V} }$ denotes the Lie derivative of the one-form density $\frac{\partial \mathcal{L} }{\partial \mathcal{V} }$, \textcolor{black}{see Remark \ref{density_equal_nforms}}, in the direction $ \mathcal{V} $ and $\operatorname{div}^ \nabla $ is the divergence operator associated to a given torsion free covariant derivative $ \nabla $. \textcolor{black}{The dot ``$\,\cdot\,$", resp., the colon ``$\,:\,$", indicate the contraction of one covariant index and one contravariant index, resp., the contraction of all covariant and contravariant indices}. In \eqref{BC_EL_conv} the outward pointing normal vector field $N$ and the flat operator are computed relative to the metric $C$.

\color{black} 
\begin{remark}\label{density_equal_nforms}\rm In this paper we use the word density as a synonym for $(n+1)$-form as this terminology is standard. No confusion arises since all the manifolds involved are assumed to be orientable. For instance, a $(p,q)$-tensor density on $ \mathcal{M} $ is defined as a section of the vector bundle $T^p_q \mathcal{M} \otimes \wedge ^{n+1} \mathcal{M}  \rightarrow \mathcal{M}$.
\end{remark} 
\color{black} 

\subsection{Reduced variational principle in the Eulerian frame}

Assuming material covariance \eqref{mat_cov}, we obtain that the material Lagrangian density $\mathscr{L}$ is invariant under the right action of the group $ \operatorname{Diff}(\mathcal{B} )$ and thus induces the spatial Lagrangian density $\ell$ as
\[
\mathscr{L}( \dot \varphi , T \varphi , R, S, G, g \circ \varphi )= \varphi ^* [\ell(u, \varrho , \varsigma , \mathsf{c}, g)]
\]
with
\[
\ell(u, \varrho   , \varsigma , \mathsf{c},g) =   \frac{1}{2} g(u,u) \varrho - \varpi( \rho  , \eta , \mathsf{c}, g) \varrho ,
\]
where $u= \dot \varphi \circ \varphi ^{-1} $ is the Eulerian velocity and we recall the relations \eqref{PF_Newtonian}, \eqref{rho_s}, and \eqref{specific_s}.

The Hamilton principle \eqref{Ham_principle_Newton} induces the following Eulerian variational formulation
\begin{align*} 
&\!\!\left.\frac{d}{d\varepsilon}\right|_{\varepsilon=0}  \int_{t_0}^{t_1} \!\int_ {\textcolor{black}{  \mathcal{N} _ \varepsilon  }}\ell(u_ \varepsilon , \varrho _ \varepsilon   , \varsigma _ \varepsilon , \mathsf{c}_ \varepsilon ,g) dt=0 \quad \text{for variations}\\
&\delta u= \partial _t \xi +[u , \xi ], \quad   \delta \varrho = - \pounds _ \xi  \varrho, \quad  \delta \varsigma  = - \pounds _\xi  \varsigma  , \quad  \delta  \mathsf{c} = - \pounds _\xi   \mathsf{c},\phantom{\int_A^B}
\end{align*} 
where \textcolor{black}{$ \xi $ is an arbitrary time dependent vector field on $ \mathcal{N}$} vanishing at $t=t_0,t_1$, see Fig. \ref{figure_1} on the right. \textcolor{black}{The variation $ \delta \mathcal{N} $ is also slaved to $ \xi $ since it is given by the equivalence class of $ \xi |_{ \partial \mathcal{N} }$ modulo vector fields tangent to $ \partial \mathcal{N} $. These variations are obtained by the relations $\mathcal{N} _ \varepsilon  = \varphi _ \varepsilon   ( \mathcal{B} )$, $u_ \varepsilon  = \dot \varphi _ \varepsilon  \circ \varphi_ \varepsilon   ^{-1} $, $\varrho_ \varepsilon   = (\varphi _ \varepsilon  )_* R$, $\varsigma_ \varepsilon  = (\varphi_ \varepsilon  ) _* S$, $\mathsf{c}_ \varepsilon  = (\varphi_ \varepsilon  ) _* G$ and by defining $\xi = \delta \varphi \circ \varphi ^{-1} $.} A direct application of this principle yields the equations in Eulerian representation as
\begin{equation}\label{Eulerian_continuum} 
\frac{\partial }{\partial t} \frac{\partial \ell }{\partial u} +  \pounds _u \frac{\partial\ell}{\partial u} =  \varrho d \frac{\partial \ell}{\partial \varrho }+ \varsigma  d \frac{\partial \ell}{\partial \varsigma} - \frac{\partial \ell}{\partial c}: \nabla \mathsf{c} +2 \operatorname{div}^ \nabla \left(\frac{\partial \ell}{\partial \mathsf{c}} \cdot \mathsf{c} \right) 
\end{equation} 
together with the boundary conditions
\[
\left(  \left(  \ell - \varrho \frac{\partial \ell}{\partial \varrho }- \varsigma  \frac{\partial \ell}{\partial \varsigma  } \right)  \delta  - 2 \frac{\partial \ell}{\partial \mathsf{c}} \cdot \mathsf{c} \right) ( \cdot ,n ^\flat )=0 \quad \text{on} \quad   \varphi (\partial \mathcal{B} ),
\]
with $n$ the outward pointing unit normal vector field to $ \varphi ( \partial \mathcal{B} )$ with respect to $g$ and where $\flat$ is the flat operator associated to $g$. In \eqref{Eulerian_continuum} $ \nabla $ is a torsion free covariant derivative and $ \operatorname{div}^ \nabla $ the associated divergence operator. For instance, one can choose the Levi-Civita covariant derivative associated to $g$ or the one associated to $\mathsf{c}$. In the latter case, the second to last term in \eqref{Eulerian_continuum} vanishes. The continuity equations $ \partial _t \varrho + \pounds  _u \varrho =0$, $ \partial _t \varsigma  + \pounds  _u \varsigma  =0$, and $ \partial _t \mathsf{c} + \pounds  _u  \mathsf{c}  =0$ follow from \eqref{PF_Newtonian}. One also notices that \eqref{Eulerian_continuum} can be equivalently written as
\begin{equation}\label{rewriting_Euler}
\partial _t \frac{\partial \ell}{\partial u} + \operatorname{div}^ \nabla \left(  \left(  \ell - \varrho \frac{\partial \ell}{\partial \varrho }- \varsigma  \frac{\partial \ell}{\partial \varsigma  } \right)  \delta  + u \otimes \frac{\partial \ell}{\partial u} - 2 \frac{\partial \ell}{\partial \mathsf{c}} \cdot \mathsf{c} \right)=\frac{\partial \ell}{\partial g}: \nabla g,
\end{equation}
where the right hand side vanishes if $ \nabla $ is chosen as the Levi-Civita covariant derivative associated to $g$.
The (1,1) tensor field density, \textcolor{black}{see Remark \ref{density_equal_nforms}}, appearing on the left hand side takes the coordinate expression
\[
\left( \ell - \varrho \frac{\partial \ell}{\partial \varrho }- \varsigma  \frac{\partial \ell}{\partial \varsigma  } \right)  \delta^ \mu _ \nu   + u^ \mu   \frac{\partial \ell}{\partial u^ \nu } - 2 \frac{\partial \ell}{\partial \mathsf{c}_{\mu  \alpha }}  \mathsf{c}_{\nu  \alpha }.
\]
We refer to \cite{GBMaRa2012} for a detailed derivation of these equations. 
For motion in a fixed domain and in absence of $\mathsf{c}$, this formulation reduces to the Euler-Poincar\'e reduction on diffeomorphism groups, see \cite{HoMaRa1998}.

\begin{figure}[h!]
{\noindent
\footnotesize
\begin{center}
\hspace{.3cm}
\begin{xy}
\xymatrix{
& *+[F-:<3pt>]{
\begin{array}{c}
\vspace{0.1cm}\text{Hamilton principle}\\
\vspace{0.1cm}\displaystyle \delta \int_{t_0}^{t_1}\!\!\int_ \mathcal{B}  \mathscr{L}\,dt=0,\; \delta \varphi \;\text{free}\\
\vspace{0.1cm}\text{Material description}
\end{array}
}
\ar[ddl]|{\begin{array}{c}\textit{Spatial covariance} \\
\textit{$ \operatorname{Diff}( \mathcal{S} )$-symmetry}\\
\end{array}}
\ar[ddr]|{\begin{array}{c}\textit{Material covariance} \\
\textit{$ \operatorname{Diff}( \mathcal{B} )$-symmetry}\\
\end{array}} & & \\
& & & \\
*+[F-:<3pt>]{
\begin{array}{c}
\vspace{0.1cm}\text{Convective reduced}\\
\vspace{0.1cm}\text{variational formulation}\\
\vspace{0.1cm}\displaystyle\delta \int_{t_0}^{t_1}\!\! \int_ \mathcal{B}  \mathcal{L}\,dt=0,\\
\vspace{0.1cm}\delta \mathcal{V}= \partial _t \zeta -[ \mathcal{V}, \zeta ],\\
\vspace{0.1cm}\delta C= \pounds _ \zeta C
\\
\end{array}
}
& &  *+[F-:<3pt>]{
\begin{array}{c}
\vspace{0.1cm}\text{Eulerian reduced}\\
\vspace{0.1cm}\text{variational formulation}\\
\vspace{0.1cm}\displaystyle\delta \int_{t_0}^{t_1} \!\!\int_ { \varphi (\mathcal{B})}\!\!\ell \,dt=0,\\
\vspace{0.1cm}\delta u= \partial _t \xi +[u , \xi ],\\
\vspace{0.1cm}\delta \varrho = - \pounds _ \xi  \varrho,\;\delta \varsigma  = - \pounds _\xi  \varsigma,\\
\vspace{0.1cm}\delta  \mathsf{c} = - \pounds _\xi   \mathsf{c}
\\
\end{array}
}\\
}
\end{xy}
\end{center}
}
\caption{Illustration of the variational principles in the three representations of Newtonian continuum mechanics.}
\label{figure_1}
\end{figure}
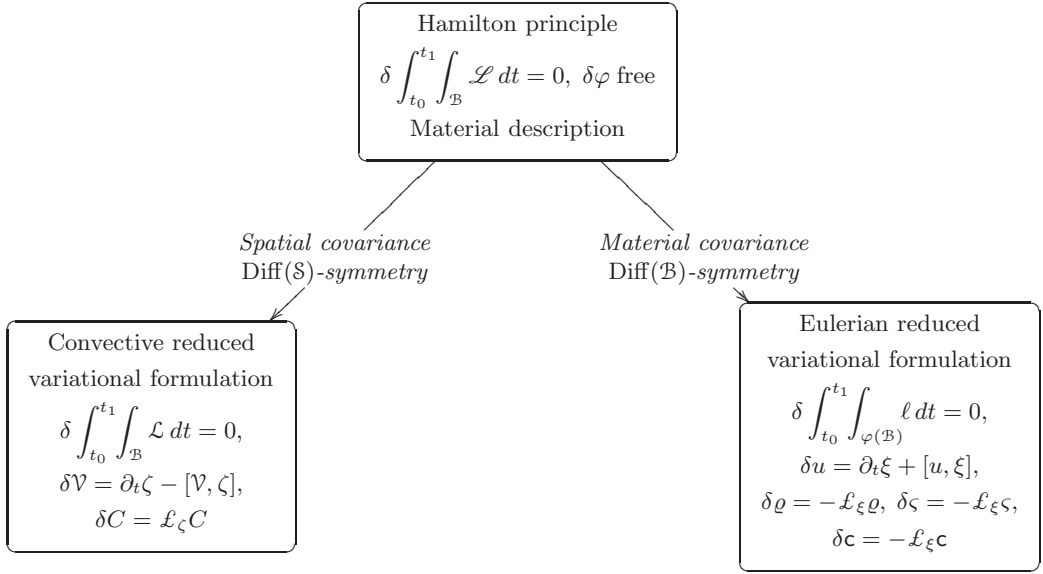
\medskip 

\color{black} 
\begin{remark}[Convective picture]\rm
Besides the well-known material and Eulerian descriptions of continuous media, we have seen that a third description, the convective picture, naturally arises in our setting which is associated to spatial covariance in the same way the Eulerian description is associated to material covariance. Its physical meaning and historical origin are both explained by the use of the so called convected coordinates. Given a coordinate system $\mathsf{X}^A$ on $ \mathcal{B} $, let $ \mathcal{X} ^A_t= \mathsf{X}^A \circ \varphi _t ^{-1}$ be the corresponding convected system of coordinates which is frozen into the medium and deforms with it.
Then, the components of the material and Eulerian velocities $V= \dot  \varphi $ and $u= \dot  \varphi \circ \varphi ^{-1} $ with respect to $ \mathcal{X} ^A_t$ equal those of the convective velocity $ \mathcal{V} = T \varphi ^{-1} \circ \dot  \varphi $ with respect to $\mathsf{X}^A$. Therefore, the convective picture is an intrinsic formulation of the description of the continuum in convected coordinates.
For rigid body motion, the convective description reduces to the body frame description. Convective (or body) variables are also an indispensable tool in the study of rods and shells, for example, regarded as one- and two-dimensional continua respectively.
\end{remark}

\color{black} 

\begin{remark}[Lagrangian reduction by symmetry]\label{rmk_LRBS}\rm In this review we have put the emphasis on the Lagrangian densities, rather than on the Lagrangian functions which are obtained by integrating the densities on the body as in \cite{GBMaRa2012}. The reason being that we shall work with Lagrangian densities in the relativistic case. Using the integrated Lagrangian makes, however, more transparent the link with the general process of \textit{Lagrangian reduction by symmetry}. This general process can be described as follows. Let $L:TQ \rightarrow \mathbb{R} $ be a Lagrangian function defined on the tangent bundle $TQ$ of the configuration manifold $Q$ of some mechanical system. Consider a free and proper action of a Lie group $G$ on $Q$ and assume that $L$ is invariant under the action of $G$ naturally induced on $TQ$ (the tangent lifted action). Let $\ell:(TQ)/G \rightarrow \mathbb{R} $ be the associated reduced Lagrangian defined on the quotient space. Then, the Lagrangian reduction approach yields a systematic way to derive the reduced Euler-Lagrange equations for $\ell$ on $(TQ)/G$, by considering the variational principle on $(TQ)/G$ induced from the Hamilton principle on $TQ$, see \cite{MaSc1993}, \cite{CeMaRa2001}.
In the case of free boundary continuum mechanics reviewed above, we have $Q= \operatorname{Emb}( \mathcal{B} , \mathcal{S} )$ and the symmetry group $G$ is a subgroup of $ \operatorname{Diff}( \mathcal{B} )$ or $ \operatorname{Diff}( \mathcal{S} )$. Continuum mechanics has also been studied from the point of view of Hamiltonian reduction by symmetry, with the aim of systematically deriving the Poisson brackets governing the dynamics in the Eulerian and convective representations. We refer to \cite{MaWe1983}, \cite{MaRaWe1984}, \cite{LeMaMoRa1986}, \cite{MaRa1989} for the Eulerian description and to \cite{HoMaRa1986}, \cite{SiMaKr1988} for the convective description.
\end{remark}


\section{Covariance properties and reduced Lagrangian densities}\label{sec_3}

In this section we describe a geometric setting for relativistic continua that allows to emphasize two notions of covariance: the usual notion of covariance with respect to spacetime diffeomorphisms as well as the notion of covariance with respect to diffeomorphisms of the reference configuration. 
These are referred to as \textit{spacetime covariance} and \textit{material covariance}. This setting is a natural extension of the geometric approach to Newtonian continuum mechanics reviewed above in \S\ref{Review_Newtonian}.

Fundamental for our description are the notions of \textit{material (or reference) tensor fields} and \textit{spacetime tensor fields} which are given tensors fields on the reference configuration and on the spacetime. Their identification is crucial for the two definitions of covariance.


Using this setting and under the covariance assumptions, it is possible to systematically define three related Lagrangian densities for a given relativistic continuum theory. Each of these three Lagrangian densities depends on \textit{dynamical variables} that are subject to variations in the critical action principle, as well as on \textit{parametric variables} that are held fixed. We use the following terminology for these densities:
\begin{itemize}
\item[\bf (1)] \textit{The material Lagrangian density}, denoted $\mathscr{L}$, which depends on the configuration map of the continua (the world-tube) and its first derivatives;
\item[\bf (2)] \textit{The spacetime (or Eulerian) Lagrangian density}, denoted $\ell$, defined from $ \mathscr{L} $ under the assumption of material covariance;
\item[\bf (3)] \textit{The convective Lagrangian density}, denoted $ \mathcal{L} $, defined from $\mathscr{L}$ under the assumption of spacetime covariance.
\end{itemize}

We shall consider the material Lagrangian density as the primary object, from which the expressions of $ \mathcal{L} $ and $\ell$ are deduced under the corresponding covariance assumption. This is justified by the fact that the critical action principle for the material Lagrangian density is given by the Hamilton principle applied to the configuration map of the continuum (the world-tube). This principle \textcolor{black}{does not} involve any constraints in the variations of the variables and is the field theoretic analogue to the Hamilton principle $ \delta \int L(q, \dot q) d t=0$ of classical mechanics. The critical action principles for the densities $\mathcal{L}$ and $\ell$ take however a more involved form, with constrained variations, that is systematically deduced from the Hamilton principe for $ \mathscr{L} $.


\color{black} 
\begin{remark}[Reduction terminology]\rm As we shall see, in the relativistic context the passing from the material to the Eulerian and convective pictures does not fit in the classical Lagrangian reduction scheme $TQ \rightarrow (TQ)/G$ of Newtonian mechanics, see Remark \ref{rmk_LRBS}, since there is not notion of absolute time at which the reduction is performed. We still use the terminology ``Lagrangian reduction" in the relativistic context since it preserves its main characteristics, namely, the use of Lie group symmetries in the material Lagrangian picture to deduce Eulerian/convected versions of Hamilton's principle, via the push-forward/pull-back tensorial operations, leading to the appearance of constrained variations. See Figures \ref{figure_1} and \ref{figure_2} for the striking analogies between the Newtonian and relativistic reductions.
\end{remark} 
\color{black}
 
\subsection{Geometric setting and Lagrangian density in material description}

Here we first review the definition of the world-tube and the associated generalized velocity and world-velocity for relativistic continua. Then, we introduce the notion of reference tensor fields and spacetime tensor fields, and we give the general geometric setting for the definition of the Lagrangian density of a relativistic continuum in the material description. In this description, the equation of evolution are obtained by the Hamilton principle applied to the world-tube.

\paragraph{World-tube and world-velocity.} Consider a $(n+1)$-dimensional spacetime $ \mathcal{M} $ endowed with a Lorentzian metric $\mathsf{g}$\footnote{Note the different notation used for the \textit{Riemannian} metric $g$ on the ambient space $ \mathcal{S} $ in \S\ref{Review_Newtonian} and the \textit{Lorentzian} metric $\mathsf{g}$ on spacetime $ \mathcal{M} $ used here.}. The relativistic motion of a continuous media in $\mathcal{M}$ is described by an embedding
\[
\Phi: \mathcal{D} =[a,b] \times \mathcal{B}  \rightarrow \mathcal{M} ,
\]
called the \textit{world-tube}, where $ \mathcal{B} $ is a $n$-dimensional compact orientable manifold with smooth boundary describing the material continuum, and where $\mathsf{g}( \partial _ \lambda \Phi , \partial _ \lambda \Phi ) <0$. For each $\mathsf{X}\in \mathcal{B} $, the curve $\lambda \in I\mapsto \Phi( \lambda ,  \mathsf{X}) \in \mathcal{M} $ is the world-line of the particle with label $ \mathsf{X}\in \mathcal{B} $. \textcolor{black}{We shall denote by $ \mathcal{N} = \Phi ( \mathcal{D} )$ the region of spacetime occupied by the continuum and by $ \partial _{\rm cont} \mathcal{N} = \Phi ([a,b] \times \partial \mathcal{B} )$ the timelike region of spacetime occupied by its boundary. }


The \textit{generalized velocity} of the media is the vector field on $\Phi( \mathcal{D} )$ defined as
\begin{equation}\label{gen_vel_1}
w=  \partial_ \lambda  \Phi  \circ \Phi ^{-1} \in \mathfrak{X} ( \Phi (\mathcal{D} )),
\end{equation} 
and its normalized version
\[
u= \frac{c}{\sqrt{-\mathsf{g}(w,w)}} w\in \mathfrak{X} ( \Phi (\mathcal{D} )), \qquad \mathsf{g}(u,u)=- c ^2 ,
\]
is the \textit{world-velocity}, \textcolor{black}{with $c$ the speed of light} . Note that we can write the generalized velocity as $w= \Phi _* \partial _ \lambda $, where $\Phi_*$ denotes the push-forward of the vector field $ \partial _ \lambda $  by $\Phi$.

For our subsequent development it is advantageous to consider a more general setting in which $\mathcal{D}$ is an arbitrary $(n+1)$-dimensional manifold with piecewise smooth boundary and $W \in \mathfrak{X} ( \mathcal{D} )$ is a given nowhere vanishing vector field on $\mathcal{D}$. This setting is helpful both for conceptual and notational reasons. In this context, the generalized velocity of the world-tube is defined as
\begin{equation}\label{gen_vel_2} 
w= \Phi _* W \in \mathfrak{X} ( \Phi (\mathcal{D} )).
\end{equation} 
The relevant case is of course $ \mathcal{D} =[a,b] \times \mathcal{B} $ and $W= \partial _ \lambda $ in which case \eqref{gen_vel_2} recovers \eqref{gen_vel_1}. \textcolor{black}{In this paper, we shall not try to attribute a physical meaning to $ \mathcal{D} $ when not given by $ \mathcal{D} = [a,b] \times \mathcal{B} $. The introduction of $ \mathcal{D} $ is above all  made for notational simplicity and to show that the abstract framework with the two covariance properties can be formulated within this general context if one is willing to do so.}

\paragraph{Reference and spacetime tensor fields.} The dynamics of relativistic continua depends on given tensor fields on $\mathcal{D}$ and $ \mathcal{M} $, called \textit{reference tensor fields} and \textit{spacetime tensor fields}, respectively. For simplicity of the exposition we assume that, besides the reference vector field $W$ on $ \mathcal{D} $, the dynamics depends on a given $(p,q)$-tensor field $K$ on $\mathcal{D} $, and a given $(r,s)$-tensor field $ \gamma $ on $ \mathcal{M} $. We use the notations $ K \in \mathcal{T} ^p_q( \mathcal{D} )$ and $ \gamma \in \mathcal{T} ^r_s( \mathcal{M} )$. It is assumed that
\begin{equation}\label{L_W_K_0} 
\pounds _WK=0,
\end{equation} 
where $\pounds _W$ denotes the Lie derivative of a tensor field in the direction $W$. Thus, $K$ is assumed to be constant on the flow $ \phi _ \tau $ of $W$, i.e. 
\[
\phi _ \tau ^* K=K.
\]

\begin{remark}[Generalization and examples]\label{remark_generalization}{\rm We present the theory for general tensor fields, but it directly also applies to tensor field densities, differential forms, and pseudo-Riemannian metrics, for instance. Also the extension to the case of the dependence on a collection $\{K_i\}$ and $\{ \gamma  _j\}$ of reference and spacetime tensor fields is straightforward. The collection of spacetime tensor fields necessarily includes the Lorentzian metric $\mathsf{g}$.
An important reference tensor field for our development is the \textit{reference volume form} $R \in \Omega ^{n+1}( \mathcal{D} )$, $ \mathcal{D} =[a,b] \times \mathcal{B}$, defined by
\[
R= d \lambda \wedge \pi _\mathcal{B}  ^* R_0,
\]
where $R_0 \in \Omega ^n( \mathcal{B} )$ is a volume form on $\mathcal{B} $, the mass form, and $ \pi _ \mathcal{B} :[a,b] \times \mathcal{B} \rightarrow \mathcal{B} $ is the projection. The assumption \eqref{L_W_K_0} is satisfied with the vector field $W= \partial _ \lambda $ since
\[
\pounds _{ \partial _ \lambda } R = \pounds _{ \partial _ \lambda }d \lambda \wedge \pi _B ^* R_0 + d \lambda \wedge \pounds _{ \partial _ \lambda }\pi _B ^* R_0=0.
\]}
\end{remark} 

Given a world-tube $\Phi:\mathcal{D}  \rightarrow \mathcal{M} $, the Eulerian expression of a given reference tensor field $K$ is defined by the push-forward of $K$ onto $ \Phi ( \mathcal{D} ) \subset \mathcal{M} $ by the world-tube:
\begin{equation}\label{Eulerian_expression} 
\kappa = \Phi _* K \in \mathcal{T} ^p_q( \Phi ( \mathcal{D} )).
\end{equation} 
From the formula \textcolor{black}{$\Phi_* (\pounds _W K)= \pounds _{ \Phi_* W} \Phi _* K$}, we get the important relation
\[
\pounds _w \kappa  =0,
\]
where $w$ is the generalized velocity. Similarly, the material expression of a given spacetime tensor field $ \gamma $ is defined by the pull-back of $ \gamma $ to $ \mathcal{D} $ by the world-tube:
\begin{equation}\label{Pull_back_gamma} 
\Gamma = \Phi ^* \gamma \in \mathcal{T} ^r_s( \mathcal{D} ).
\end{equation} 
Note that although $ \Phi $ is not a diffeomorphism $ \mathcal{D} \rightarrow  \mathcal{M} $, but only an embedding, the pull-back of an arbitrary tensor field can be defined since $ \mathcal{D} $ and $ \mathcal{M} $ have the same dimension\footnote{Explicitly, for $ \alpha ^1,..., \alpha ^r \in T^*_X \mathcal{D} $ and $ u _1,...,u_s \in T_X \mathcal{D} $, we have 
\[
\Gamma (X)( \alpha ^1,..., \alpha ^r, u_1,...,u_s)= \gamma ( \Phi (X))\left(  (T^*_X \Phi )^{-1} ( \alpha ^1 ),... (T^*_X \Phi )^{-1} ( \alpha ^r ),  T_X \Phi (u_1),..., T_X \Phi (u_s)\right) 
\]
with $T_X \Phi :T_X \mathcal{D} \rightarrow T_{ \Phi (X)} \mathcal{M} $ and $T^*_X \Phi :  T_{ \Phi (X)}^* \mathcal{M} \rightarrow  T_X ^*\mathcal{D}$ isomorphisms and $T_X \mathcal{M} =T_{ \Phi (X)}[ \Phi ( \mathcal{D} )]$ for all $X \in \mathcal{D} $.}.

\paragraph{Lagrangian density in the material description.} The material description of a relativistic continuum uses the world-tube as a primary variable. This description involves both the configuration manifold $ \mathcal{D} $ and the spacetime $ \mathcal{M} $, as opposed to the Eulerian representation and the convective representation, which only use one of them. It is in the material description that the critical action principle takes the simplest form. It is given by the Hamilton principle for the variations of the world-tube.

In general, the Lagrangian density $ \mathscr{L} $ of a relativistic continuum depends on the material point $X \in \mathcal{D} $ and on the value of the world-tube and of its first derivative at this point. It also depends parametrically on the point values of the given reference and spacetime tensor fields $W$, $K$ and $ \gamma $. It is thus a bundle map
\begin{equation}\label{material_Lagrangian} 
\mathscr{L} : J^1( \mathcal{D}  \times \mathcal{M} ) \times T \mathcal{D}  \times T^p_q\mathcal{D}  \times T^r_s \mathcal{M}  \rightarrow \wedge ^{n+1} \mathcal{D} 
\end{equation}
covering the identity on $ \mathcal{D} $. Here $J^1( \mathcal{D}  \times \mathcal{M} ) \rightarrow \mathcal{D}  \times \mathcal{M} $ denotes the first jet bundle of the trivial fiber bundle $\mathcal{D}  \times \mathcal{M}  \rightarrow \mathcal{D} $. The vector fiber of $J^1( \mathcal{D}  \times \mathcal{M} )$ at $(X,x) \in \mathcal{D} \times M$ is $L(T_X \mathcal{D} , T_x \mathcal{M} )$, the space of linear maps $T_X \mathcal{D} \rightarrow T_x \mathcal{M} $.
Also, $T^p_q \mathcal{D}  \rightarrow \mathcal{D} $ denotes the vector bundle of $(p,q)$-tensors on $\mathcal{D} $, similarly for $T^r_s \mathcal{M}  \rightarrow \mathcal{M} $. Note that $J^1( \mathcal{D}  \times \mathcal{M} ) \times T \mathcal{D}  \times T^p_q\mathcal{D}  \times T^r_s \mathcal{M}$ is a bundle over $ \mathcal{D} \times \mathcal{M} $, which can also be regarded as a bundle over $ \mathcal{D} $. It is the latter interpretation that is understood in \eqref{material_Lagrangian}.

The case when $\mathscr{L}$ depends on differential forms or pseudo-Riemannian metrics, see Remark \ref{remark_generalization}, is treated similarly, by considering appropriate subbundles of $T^p_q \mathcal{D} $ and $T^r_s \mathcal{M} $.

Sometimes it is required to appropriately restrict the bundles to open subbundles in order to have the Lagrangian density well-defined. This is typically the case to ensure the condition $\mathsf{g}( \partial _ \lambda \Phi , \partial _ \lambda \Phi ) <0$. We will not indicate such restrictions as they are obvious from the context.

In local coordinates the Lagrangian density $\mathscr{L}$ reads
\[
\mathscr{L} =\bar{ \mathscr{L} }\big(  X^a, x^\mu, v^\mu_a,W^b, K^{a_1...a_p}_{b_1...b_q}, \gamma^{\mu_1...\mu_r}_{\nu_1...\nu_s} \big) d^{n+1} X.
\]
Here $X^a$, $a=1,...,n+1$ and $ x^\mu$, $\mu=1,...,n+1$ are local coordinates on $\mathcal{D} $ and $\mathcal{M} $, and $( X^a, x^\mu, v^\mu_a)$ are the local coordinates induced on $J^1( \mathcal{D} \times \mathcal{M} )$.
We also use the notation $d^{n+1}X= dX^1 \wedge ... \wedge dX^{n+1}$.
The evaluation of the Lagrangian density on a world tube $ \Phi $, a reference vector field $W \in \mathfrak{X} ( \mathcal{D} )$, a reference tensor field $K \in \mathcal{T}^p_q(\mathcal{D} )$, and a spacetime tensor field $\gamma\in \mathcal{T}^r_s(\mathcal{M} )$, is written as
\begin{equation}\label{Lagrangian_K} 
\mathscr{L}(j^1\Phi, W, K, \gamma \circ \Phi  ),
\end{equation}
where $j^1\Phi: \mathcal{D}  \rightarrow J^1( \mathcal{D}  \times \mathcal{M} )$ denotes the first jet extension of the world-tube $\Phi$, given locally as $X^a \mapsto (X ^a , \Phi ^\mu(X^a), \Phi ^\mu_{,b}(X^a))$, \textcolor{black}{with $ \Phi ^\mu_{,b}:= \partial _{X^b}\Phi ^ \mu$}.

\textcolor{black}{When $ \mathcal{D} =[a,b] \times \mathcal{B}\ni X=( \lambda , \mathsf{X}) $, we can describe the first jet of the world-tube as the couple $( \partial _ \lambda\Phi  , T \Phi )$ with $T \Phi $ the tangent map to $\Phi ( \lambda , \cdot ):  \mathcal{B} \rightarrow \mathcal{M} $ for each fixed $ \lambda \in [a,b]$, in analogy with the non-relativistic notation $(\dot \varphi , T \varphi )$. The first jet notation is more suggestive since, as opposed to the time of the Newtonian case, the parameter $ \lambda $ is not treated differently from the labels $X$ when passing to the Eulerian and convected pictures.}


For given $W$, $K$, and $ \gamma $, the Hamilton principle for this Lagrangian density reads
\begin{equation}\label{Ham_principle} 
\left. \frac{d}{d \varepsilon }\right|_{\varepsilon=0} \int_\mathcal{D}  \mathscr{L} ( j^1\Phi_ \varepsilon , W, K , \gamma \circ \Phi _ \varepsilon )=0
\end{equation} 
for arbitrary variations $\Phi _ \varepsilon  $ of the world tube $\Phi $ \textcolor{black}{with fixed endpoints at $ \lambda = a,b$}. It is important to note that $W$, $K$, $ \gamma $ are held fixed during the process of taking variations. This principle yields the Euler-Lagrange equations and boundary conditions
\[
\partial _a \frac{\partial \bar{\mathscr{L}}}{\partial \Phi^\mu_{,a}}  - \frac{\partial \bar{\mathscr{L}}}{\partial \Phi^\mu} = \frac{\partial \bar{\mathscr{L}}}{\partial \gamma ^{\mu_1...\mu_r}_{\nu_1...\nu_s}} \partial_\mu \gamma^{\mu_1...\mu_r}_{\nu_1...\nu_s} , \qquad \frac{\partial \bar{\mathscr{L} }}{\partial \Phi^\mu_{,a}}d^nX_a=0 \quad \text{on} \quad  \textcolor{black}{[a,b] \times \partial \mathcal{B}},
\]
where $d^nX_a= \mathbf{i} _{ \partial _a}d^{n+1}X$, \textcolor{black}{with $ \mathbf{i} _U \lambda $ denoting the insertion of the vector field $U$ in the first slot of a covariant tensor field $ \lambda $, here given by the $(n+1)$-form $d^{n+1}X$}. This boundary condition arises from letting the variations $ \delta \Phi = \left. \frac{d}{d\varepsilon}\right|_{\varepsilon=0} \Phi _ \varepsilon  $ be arbitrary on $[a,b] \times  \partial \mathcal{B} $. Other boundary conditions can be obtained by imposing restriction on the boundary variations or adding a boundary term in the action functional.



\subsection{Spacetime covariance and the convective Lagrangian density}

In this section we state the definition of spacetime covariance for a Lagrangian density of the type \eqref{Lagrangian_K}, which is a general assumption on relativistic theories. We show that when a Lagrangian density is spacetime covariant, it has an associated convective Lagrangian density.

We say that the Lagrangian density \eqref{Lagrangian_K} is \textit{spacetime covariant} if it satisfies
\begin{equation}\label{spacetime_equivariance} 
\mathscr{L} ( j ^1 ( \psi \circ \Phi ), W,K, \psi _\ast \gamma \circ \psi \circ\Phi  )= \mathscr{L} ( j^1\Phi , W,K, \gamma \circ\Phi ),
\end{equation} 
for all $ \psi  \in \operatorname{Diff}( \mathcal{M} )$ and for all $ \Phi $, $W$, $K$, $ \gamma $.

A simple example of a spacetime covariant Lagrangian density is $\mathscr{L}: J^1( \mathcal{D} \times \mathcal{M} ) \times S^2_L \mathcal{M} \rightarrow \wedge ^{n+1} \mathcal{D} $ given by
\begin{equation}\label{ex_1} 
\mathscr{L} ( j ^1 \Phi , \mathsf{g} \circ \Phi )= \Phi ^* [\mu(\mathsf{g})],
\end{equation} 
where $S^2 _L\mathcal{M} \subset T^0_2 \mathcal{M}  \rightarrow \mathcal{M} $ denotes the bundle of Lorentzian metrics on $ \mathcal{M} $ \textcolor{black}{(as indicated by the index $L$)} and $ \mu (\mathsf{g})$ is the volume form associated to the Lorentzian metric $\mathsf{g}$. Spacetime covariance is checked as
\[
\mathscr{L} ( j ^1 ( \psi \circ \Phi ), \psi _*\mathsf{g} \circ \psi \circ  \Phi ) = ( \psi \circ \Phi ) ^* [ \mu ( \psi _* \mathsf{g} )]= \Phi ^* [\mu(\mathsf{g})] = \mathscr{L} ( j ^1 \Phi , \mathsf{g} \circ \Phi ).
\]
More generally, $\mathscr{L}: J^1( \mathcal{D} \times \mathcal{M} ) \times V \mathcal{D} \times S^2 _L\mathcal{M} \rightarrow \wedge ^{n+1} \mathcal{D} $ given by
\begin{equation}\label{ex_2}
\mathscr{L} (j^1 \Phi , R , \mathsf{g} \circ \Phi) = e \Big( \frac{R}{\Phi^* [\mu(\mathsf{g})]}\Big) R 
\end{equation} 
with $V \mathcal{D}  \subset \wedge ^{n+1} \mathcal{D} \rightarrow \mathcal{D} $ the bundle of \textcolor{black}{volume forms on $ \mathcal{D} $} and $ e : \mathbb{R} \rightarrow \mathbb{R} $ a strictly positive function, is spacetime covariant. 
  
\begin{proposition}[The convective Lagrangian density]\label{convected_Lagrangian}  Let $ \mathscr{L} $ be a Lagrangian density \eqref{Lagrangian_K} with the spacetime covariance property \eqref{spacetime_equivariance}. Then, there is an associated {\bfi convective Lagrangian density}, given by a bundle map $ \mathcal{L} : T \mathcal{D}  \times  T ^p _q  \mathcal{D}  \times  T^r _s \mathcal{D} \rightarrow  \wedge  ^{n+1} \mathcal{D}$ covering the identity on $ \mathcal{D} $, such that
\begin{equation}\label{scrL_to_calL} 
\mathscr{L} ( j^1\Phi ,W, K, \gamma \circ\Phi )= \mathcal{L}( W, K , \Gamma ),
\end{equation} 
where
\[
\Gamma:=\Phi ^* \gamma \in \mathcal{T} ^r_s( \mathcal{D} ).
\]
\end{proposition}
\begin{proof} We fix an embedding $\Phi_0: \mathcal{D}  \rightarrow \mathcal{N}_0 \subset \mathcal{M} $. Then it suffices to choose $\psi \in \operatorname{Diff}( \mathcal{M} )$ to be an arbitrary extension of $\Phi_0 \circ \Phi ^{-1} : \Phi ( \mathcal{D} ) \rightarrow \mathcal{N}_0$ and we have
\[
\mathscr{L} (j^1 \Phi , W, K, \gamma \circ \Phi )= \mathscr{L} (j^1 \Phi _0,W,K, (\Phi _0)_* \Gamma \circ \Phi_0).
\]
We note that $(\Phi _0)_* \Gamma \circ \Phi_0= T^r_s \Phi_0 \circ \Gamma$ depends only on the values $\Gamma(X)$ of $ \Gamma $, not on $ \Gamma $ as a field.
We define
\[
\mathcal{L}(W , K , \Gamma):= \mathscr{L} (j^1 \Phi _0,W,K, (\Phi _0)_* \Gamma \circ \Phi_0).
\]
\textcolor{black}{One can then check}, using again spacetime covariance, that $ \mathcal{L} $ does not depend on the embedding $ \Phi _0$.
\end{proof}

For the examples \eqref{ex_1} and \eqref{ex_2}, the corresponding convective Lagrangian densities are the bundle maps $ \mathcal{L} : S^2_ L \mathcal{D}  \rightarrow \wedge ^{n+1} \mathcal{D} $ and $ \mathcal{L} :  V \mathcal{D}\times S^2_ L \mathcal{D}  \rightarrow \wedge ^{n+1} \mathcal{D} $ given by 
\[
\mathcal{L} (\Gamma )= \mu ( \Gamma ) \quad\text{and}\quad \mathcal{L} ( R, \Gamma)= e \left( \frac{R}{\mu( \Gamma )} \right) R,
\]  
where $ \Gamma = \Phi ^* g$.

Note that in local coordinates the convective Lagrangian density reads
\begin{equation}\label{local_calL} 
\mathcal{L} = \bar{ \mathcal{L} }\big(X^a,W^b, K^{a_1...a_p}_{b_1...b_q}, \Gamma ^{a_1...a_r}_{b_1...b_s}\big) d^{n+1}X.
\end{equation}
In particular, $ \mathcal{L} $ can depend explicitly on $X \in \mathcal{D} $ in general. We refer to $\mathcal{L} $ as the convective (or body) Lagrangian density associated to $\mathscr{L}$, since it depends only on material (or body) tensor fields, without reference to $\mathcal{M} $.

\begin{remark}[Choice of spacetime tensors and isotropy subgroup]\label{remark_gamma}\rm
In some situations one particular choice for the spatial tensor field $ \gamma $ is enough for the description of the continuum theory under study. This is the case for instance for the Minkowski metric in the case of special relativity. In such situations, for the definition \eqref{scrL_to_calL} to hold, it is enough to assume that $\mathscr{L}$ satisfies the spacetime covariance \eqref{spacetime_equivariance} for $ \psi   $ in the isotropy subgroup of the chosen field. For $ \mathcal{M} = \mathbb{R} ^4$ and $ \gamma $ the Minkowski metric, this corresponds to covariance with respect to the Poincar\'e group.
The convective Lagrangian is then defined only for $ \Gamma $ which can be written in terms of the given $ \gamma $ as  $ \Gamma = \Phi ^* \gamma$ for some world-tube $ \Phi $.
\end{remark}

\subsection{Material covariance and the spacetime Lagrangian density}\label{subsec_mat_cov}

In this section we state the definition of material covariance for a Lagrangian density of the type \eqref{Lagrangian_K}. We show that when a Lagrangian density is \textcolor{black}{materially covariant}, it has an associated spacetime (or Eulerian) Lagrangian density. 
As opposed to spacetime covariance, material covariance is related to the properties of the material, such as isotropy. However, material covariance can also be achieved in the anisotropic case by suitably extending the collection of material tensor fields of the theory, \textcolor{black}{to include structural (or anisotropic) tensors}, see Remark \ref{MT_aniso}.

We say that the Lagrangian density \eqref{Lagrangian_K} is \textit{\textcolor{black}{materially covariant}} if it satisfies 
\begin{equation}\label{equivariance_general} 
\mathscr{L} \left(  j ^1(  \Phi \circ \varphi ), \varphi ^* W,\varphi ^* K  , \gamma \circ \Phi \circ \varphi  \right) = \varphi ^\ast  \left[ \mathscr{L} \left(  j ^1 \Phi , W, K, \gamma \circ \Phi \right) \right] ,
\end{equation} 
for all $\varphi \in \operatorname{Diff}( \mathcal{D} )$ and for all $ \Phi $, $W$, $K$, $ \gamma $. While the material covariance is stated in \eqref{equivariance_general} in terms of the world-tube and tensor fields, it is really a property of the Lagrangian density as a bundle map \eqref{Lagrangian_K}, as a direct check in local coordinates shows.

Examples of \textcolor{black}{materially covariant} Lagrangian densities are given by \eqref{ex_1} and \eqref{ex_2}.

\begin{proposition}[The spacetime Lagrangian density]\label{epsilon_N}  Let $ \mathscr{L} $ be a Lagrangian density \eqref{Lagrangian_K} with the material covariance property \eqref{equivariance_general}. Then, for any domain $\mathcal{N} \subset \mathcal{M} $ diffeomorphic to $ \mathcal{D} $, there exists an associated {\bfi spacetime Lagrangian density}, given by a bundle map $ \ell_ \mathcal{N}  :  T  \mathcal{N}  \times  T ^p _q   \mathcal{N}  \times  T^r _s  \mathcal{N}  \rightarrow \textcolor{black}{\wedge  ^{n+1}  \mathcal{N}}$ covering the identity on $ \mathcal{N} $ such that
\begin{equation}\label{scrL_to_ell} 
\mathscr{L} ( j ^1 \Phi ,W, K, \gamma  \circ \Phi     )=\Phi ^\ast[ \ell _{ \mathcal{N} } ( w, \kappa , \gamma    )], 
\end{equation} 
where
\[
 \mathcal{N} = \Phi ( \mathcal{D} ) \subset \mathcal{M} ,  \quad w= \Phi _* W \in \mathfrak{X} ( \Phi (\mathcal{D} )),\quad \text{and} \quad \kappa =\Phi_\ast  K \in \mathcal{T} ^p_q( \Phi ( \mathcal{D} )).
\]
\end{proposition}
\begin{proof} Given $ \mathcal{N} \subset \mathcal{M} $ diffeomorphic to $ \mathcal{D} $, we define $\ell_ \mathcal{N} :  T  \mathcal{N}  \times  T ^p _q   \mathcal{N}  \times  T^r _s  \mathcal{N}  \rightarrow \textcolor{black}{ \wedge  ^{n+1} } \mathcal{N}$ as
\begin{equation}\label{def_ell_N} 
\ell_{ \mathcal{N} }(w, \kappa , \gamma ):= \Phi _* [ \mathscr{L} ( j ^1 \Phi , \Phi ^* w,  \Phi ^* \kappa , \gamma  \circ \Phi     )]
\end{equation} 
where $ \Phi : \mathcal{D} \rightarrow \mathcal{M}$ is an embedding with $ \Phi ( \mathcal{D} )= \mathcal{N} $. By using the material covariance \eqref{equivariance_general} one can check that this definition does not depend on the chosen embedding with $ \Phi ( \mathcal{D} )= \mathcal{N}$. Even though the relation \eqref{def_ell_N} uses the fields $w$, $ \kappa $, $ \gamma $, this equality really defines $\ell_ \mathcal{N} $ on the point values $w(x)$, $ \kappa (x)$, $ \gamma (x)$ as a direct check in local coordinates shows, i.e., it gives $\ell_ \mathcal{N} $ as a map defined on the bundle $T \mathcal{N} \times T^p_q \mathcal{N} \times T^r_s \mathcal{N} $.
\end{proof}

\medskip 

We refer to $\ell_ \mathcal{N}$ as the spacetime Lagrangian associated to $\mathscr{L}$ on $ \mathcal{N}$, since it only depends on spacetime tensor fields, without reference to $ \mathcal{D} $. Gluing together these expressions, we get a spacetime Lagrangian density defined over the whole spacetime, denoted simply $\ell : T \mathcal{M}  \times  T ^p _q  \mathcal{M}  \times  T^r _s  \mathcal{M} \rightarrow \wedge  ^{n+1} \mathcal{M} $.

\medskip

For the examples \eqref{ex_1} and \eqref{ex_2}, the associated spacetime Lagrangians are the bundle maps $\ell: S^2_L \mathcal{M} \rightarrow \wedge ^{n+1} \mathcal{M} $ and $\ell: V \mathcal{M} \times S^2_L \mathcal{M} \rightarrow \wedge ^{n+1} \mathcal{M} $
\[
\ell(\mathsf{g}) = \mu(\mathsf{g}) \quad\text{and}\quad \ell(\varrho, \mathsf{g}) =  e\Big( \frac{\varrho }{ \mu (\mathsf{g})} \Big) \varrho.
\]
Given a world-tube $ \Phi: \mathcal{D} \rightarrow \mathcal{M}  $, the volume form $ \varrho = \Phi _* R$ is the Eulerian expression of $R$ on $ \Phi ( \mathcal{D} )$, see \eqref{Eulerian_expression}.

\begin{remark}[Choice of material tensors and isotropy subgroup]\label{remark_W_K}\rm
In many situations, see the examples later, the description of the continuum relies only on one particular choice for the fields $W$ or $K$, such as $W= \partial _ \lambda $ for $ \mathcal{D} =[a,b] \times \mathcal{B} $. In this case, for the definition \eqref{scrL_to_ell} to hold, it is enough to assume that $\mathscr{L}$ satisfies the material covariance \eqref{equivariance_general} for $ \varphi $ in the isotropy subgroup of the chosen fields. For instance, if we focus on $W= \partial _ \lambda $, the corresponding isotropy subgroup is
\begin{equation}\label{isom_sdp} 
\operatorname{Diff}_{ \partial _ \lambda }( \mathcal{D} )=\{ \varphi \in \operatorname{Diff}( \mathcal{D} ) \mid \varphi ^* \partial _ \lambda = \partial _ \lambda \} \simeq \operatorname{Diff}( \mathcal{B} ) \,\circledS\, \mathcal{F} ( \mathcal{B} , \mathbb{R} )\ni ( \psi , f)
\end{equation} 
with $ \varphi ( \lambda , \mathsf{X})= ( \lambda + f(\mathsf{X}), \psi (\mathsf{X}))$. \textcolor{black}{As written above, this group is isomorphic to the semidirect product, denoted by $\circledS$, of  $\operatorname{Diff}( \mathcal{B} )$ with the space $\mathcal{F} ( \mathcal{B} , \mathbb{R} )$ of smooth functions from $ \mathcal{B} $ to $ \mathbb{R} $,  the semidirect product group multiplication being given by $( \psi _1,f_1)( \psi _2,f_2)= ( \psi _1 \circ \psi _2, f_1 \circ \psi_2 + f_2)$.}
In this case, $\ell_ \mathcal{N} $ in \eqref{def_ell_N} is defined only for $w$ which can be written in terms of the given $W$ as  $w= \Phi _*W$ for some world-tube $ \Phi $, similarly for other fields $K$. This remark is analogous to Remark \ref{remark_gamma} about spacetime covariance.
\end{remark}



\section{Eulerian and convective covariant reductions}\label{sec_4}

In this section we derive the Eulerian form, resp. the convective form of the Hamilton principle \eqref{Ham_principle} for relativistic continua under the assumption of material covariance, resp. spacetime covariance. To obtain the Eulerian form, we use the relation \eqref{scrL_to_ell} between the material Lagrangian density and its spacetime version  as well as the definition of the Eulerian form of the given material tensor fields. Similarly, to obtain the convective form, we use the relation \eqref{scrL_to_calL} between the material Lagrangian density and its convective version, as well as the definition of the convective form of the given spacetime tensor fields. These derivations use some technical results on Lie derivatives that we give below.

\subsection{Preliminaries on Lie derivatives}\label{technical_prelim}

Recall that the local expression of the Lie derivative of a $(p,q)$-tensor field $ \kappa $ along a vector field $ \zeta $ is
\begin{equation}\label{Lie_der}
\begin{aligned} 
&(\pounds _ \zeta \kappa  )^{\alpha_1 ... \alpha_p}_{\beta _1 ... \beta _q} \\
&\quad =   \zeta ^ \gamma\partial _ \gamma \kappa  ^{\alpha_1 ... \alpha_p}_{\beta _1 ... \beta _q}- \kappa  ^{\alpha_1 ... \alpha _{r-1}\gamma  \alpha _{r+1}... \alpha_p}_{\beta _1 ... \beta _q} \partial _ \gamma \zeta ^{ \alpha _r  } + \kappa  ^{\alpha_1 ... \alpha_p}_{\beta _1 ... \beta _{r-1} \gamma \beta _{r+1} ... \beta _q} \partial _{ \beta _r} \zeta  ^ \gamma \\
&\quad = \zeta ^ \gamma\partial _ \gamma \kappa  ^{\alpha_1 ... \alpha_p}_{\beta _1 ... \beta _q}  + \widehat{ \kappa  }\;^{\alpha_1 ... \alpha_p\mu}_{\beta _1 ... \beta _q\nu} \partial _ \mu \zeta ^ \nu,
\end{aligned}
\end{equation}  
where $\widehat{ \kappa  }$ is the $(p+1, q+1)$ tensor field defined by
\begin{equation}\label{hat_kappa} 
\widehat{ \kappa  }\;^{ \alpha _1 ... \alpha _p\nu} _{ \beta _1... \beta _q\mu}= \sum_r \big( \kappa  ^{ \alpha _1 ... \alpha _p} _{ \beta _1... \beta _{r-1} \mu \beta _{r+1} ...\beta _q} \delta ^ \nu_{ \beta _r}- \kappa  ^{ \alpha _1 ... \alpha  _{r-1} \nu \alpha  _{r+1} ...\alpha _p} _{ \beta _1... \beta _q} \delta ^ {\alpha  _r} _\mu\big).
\end{equation} 
Formulas \eqref{Lie_der} can also be written in terms of a given torsion free covariant derivative $ \nabla $ on $M$ by replacing $ \partial _ \gamma $ with $ \nabla _ \gamma $. With such choice, we have the global formula
\begin{equation}\label{global_formula} 
\pounds _ \zeta  \kappa = \nabla _ \zeta \kappa  + \widehat{ \kappa  }: \nabla \zeta .
\end{equation}

For future use, we recall these formulas in the case where $ \kappa  $ is, respectively, a function $f$, a vector field $w= w^\mu \partial _\mu$, a $2$ covariant symmetric tensor field $ \mathsf{c} =  \mathsf{c}_{\mu\nu} dx^\mu \otimes dx^\nu$, and a $(n+1)$-form $\varrho   = \overline{ \varrho }\, d^{n+1}x$:
\begin{align*} 
\pounds _ \zeta f &= \partial _ \gamma f \zeta ^ \gamma\\
\pounds _ \zeta w^\mu&= \partial _ \gamma w^\mu \zeta ^ \gamma-  w ^ \gamma \partial _ \gamma \zeta ^\mu\\
\pounds _ \zeta \mathsf{c}_{\mu\nu}&= \partial _ \gamma \mathsf{c} _{\mu\nu} \zeta ^ \gamma + \mathsf{c} _{ \gamma \nu} \partial _ \mu \zeta  ^ \gamma + \mathsf{c} _{ \mu \gamma } \partial _ \nu \zeta  ^ \gamma\\
\overline{ \pounds _ \zeta \varrho  }&= \partial _ \gamma ( \bar{ \varrho  } \,\zeta ^ \gamma ).
\end{align*} 

The derivation of the reduced Euler-Lagrange equations uses the following technical result.

\begin{lemma}\label{technical_lemma} Let $ \kappa  $ be a $(p,q)$-tensor field and $\pi$ a $(q,p)$-tensor field density.
Then, we have locally
\begin{equation}\label{formula_Lemma_local}
\begin{aligned} 
(\pounds _ \zeta  \kappa ) ^{ \alpha _1 ... \alpha _p} _{ \beta _1... \beta _q}  \pi  _{ \alpha _1 ... \alpha _p} ^{ \beta _1... \beta _q}&=  \zeta ^ \mu \left( \partial _ \mu \kappa  ^{ \alpha _1 ... \alpha _p} _{ \beta _1... \beta _q} \pi _{ \alpha _1 ... \alpha _p} ^{ \beta _1... \beta _q} - \partial _ \nu( \widehat{ \kappa  }\;^{ \alpha _1 ... \alpha _p\nu} _{ \beta _1... \beta _q\mu}\pi  _{ \alpha _1 ... \alpha _p} ^{ \beta _1... \beta _q})\right) \\
& \qquad + \partial _ \nu \left( \widehat{ \kappa  }\;^{ \alpha _1 ... \alpha _p\nu} _{ \beta _1... \beta _q\mu}\pi  _{ \alpha _1 ... \alpha _p} ^{ \beta _1... \beta _q} \zeta ^ \mu \right).
\end{aligned}
\end{equation}  
The same formula holds with $\partial _ \gamma $ replaced by $ \nabla _ \gamma $, with $ \nabla $ a torsion free covariant derivative. In this case, we have the global formula
\begin{equation}\label{formula_Lemma_global} 
\textcolor{black}{ \pounds _ \zeta  \kappa  : \pi = \nabla _ \zeta \kappa  : \pi - \operatorname{div}^ \nabla (  \pi \!\therefore\! \widehat{ \kappa  }) \cdot \zeta+ \operatorname{div}( ( \pi\!\therefore\! \widehat{ \kappa  } ) \cdot \zeta  ) },
\end{equation}
where the colon ``$\,:\,$" in $\pounds _ \zeta  \kappa  : \pi $ and $\nabla _ \zeta \kappa  : \pi$ denotes the full contraction and where $ \pi \!\therefore\! \widehat{\kappa }$ is the $(1,1)$ tensor field density obtaining by contracting all the respective indices of $ \widehat{ \kappa  } $ and $ \pi $ except the last covariant and contravariant indices of $\widehat{ \kappa  }$.
\end{lemma} 
\begin{proof} This follows directly from a computation in local coordinates using \eqref{Lie_der} and \eqref{hat_kappa}.
\end{proof}

\begin{remark}[Divergences]\rm
Note that the first divergence operator is associated to the torsion free covariant derivative $ \nabla $, while the last one is canonically defined since it acts on a vector field density, hence the notations $ \operatorname{div}^ \nabla $ and $ \operatorname{div}$ in \eqref{formula_Lemma_global}. Given a $(1,1)$ tensor field density $T= T^ \mu _ \nu \partial _ \mu \otimes dx ^ \nu \otimes d^{n+1}x$ and a vector field density $X= X^ \mu \partial _ \mu \otimes d^{n+1} x$, these divergences are defined as
\begin{equation}\label{def_div} 
\begin{aligned} 
(\operatorname{div}^ \nabla \! T)(u)&= \nabla _{ \partial _ \mu } T ( u, dx ^ \mu ) = \nabla _ \mu T^ \mu _ \nu u ^ \nu d^{n+1} x ,\;\;\forall u \in T \mathcal{M}\\
\operatorname{div}X&= d( \mathbf{i} _{ \partial _ \mu } \!\left\langle dx ^ \mu , X \right\rangle  )= d( X^ \mu d^nx_ \mu )= \partial _ \mu X^ \mu d^{n+1}x,
\end{aligned} 
\end{equation} 
where $ \mathbf{i} _{ \partial _ \mu } \!\left\langle dx ^ \mu , X \right\rangle $ denotes the insertion of the vector $ \partial _ \mu $ in the $(n+1)$-form $\left\langle dx ^ \mu , X \right\rangle$, \textcolor{black}{and with $d^nx_\mu= \mathbf{i} _{ \partial _ \mu } d^{n+1}x$}. We note that for $\ell$ a $(n+1)$-form and $ \zeta $ a vector field, we have
\begin{equation}\label{example_div}
\operatorname{div}(\ell \zeta )= d( \mathbf{i} _ \zeta \ell)= \pounds _ \zeta  \ell. 
\end{equation}
\end{remark} 

\medskip 

For a function $ f $, a vector field $w= w^\mu \partial _ \mu$, a $2$ covariant symmetric tensor field $ \mathsf{c} =  \mathsf{c} _{\mu\nu} dx^\mu \otimes dx^\nu$, and a $(n+1)$-form $ \varrho  = \bar{\varrho } d^{n+1}x$, formula \eqref{formula_Lemma_local} gives 
\begin{align*} 
\pounds _ \zeta f \bar\pi &= \zeta ^ \gamma \partial _  \gamma f\bar\pi \\
\pounds _ \zeta w ^ \mu \pi _ \mu &= \zeta ^ \gamma ( \partial _ \gamma w ^ \nu \pi _ \nu + \partial _ \nu ( w ^ \nu \pi _ \gamma )) - \partial _ \nu ( w ^ \nu \pi _ \mu \zeta ^ \mu )\\
\pounds _ \zeta \mathsf{c} _{ \mu \nu } \pi ^{ \mu \nu }&= \zeta ^ \gamma \left( \partial _ \gamma \mathsf{c} _{\mu\nu} \pi ^{ \mu \nu }- 2 \partial _ \mu( \mathsf{c} _{ \gamma \nu } \pi ^{ \mu \nu }) \right)  + 2 \partial _ \gamma ( \mathsf{c} _{\mu\nu} \pi ^{ \gamma \nu } \zeta ^ \mu ) \\
\textcolor{black}{ \overline{\pounds _ \zeta \varrho } \;\pi }&= - \zeta ^ \gamma \bar{\varrho } \partial _ \gamma \pi + \partial _ \gamma \left( \bar{ \varrho  } \;\pi \zeta ^ \gamma \right).
\end{align*}

\subsection{Reduced Euler-Lagrange equations on spacetime}

We now present the variational principle induced in the Eulerian description by the Hamilton principle \eqref{Ham_principle} in the material description.

To derive the covariantly reduced Euler-Lagrange equation, we will use the formula
\begin{equation}\label{variation_moving_domain} 
\left. \frac{d}{d\varepsilon}\right|_{\varepsilon=0} \int_{ \Phi _ \varepsilon  ( \mathcal{D} )} \ell_ \varepsilon  = \int_ { \Phi (\mathcal{D} )} \left. \frac{d}{d\varepsilon}\right|_{\varepsilon=0}  \ell + \int_{  \Phi(\partial \mathcal{D} )} \mathbf{i} _{ \delta \Phi \circ \Phi ^{-1} } \ell 
\end{equation} 
for the variations of an integral on a moving domain, where $ \Phi _ \varepsilon $ is a path of world-tubes with $ \Phi_{ \varepsilon =0}= \Phi $. In the second term $ \mathbf{i} _ \zeta \ell \in \Omega ^n( \mathcal{N} )$ denotes the insertion of the vector field $ \zeta $ in the $(n+1)$-form $\ell$, \textcolor{black}{with here $ \zeta = \delta \Phi \circ \Phi ^{-1} $}.

Given a $(1,1)$ tensor field density $T$, written locally as $T=T^ \mu _ \nu \partial _ \mu \otimes dx ^ \nu \otimes d^{n+1}x$, we use the notation $ \operatorname{tr}(T)= T^\mu_\nu dx^ \nu \otimes d^nx_\mu  \in \Omega ^1( \mathcal{N} ) \otimes \Omega ^{n}( \mathcal{N} )$ for the contraction of the contravariant index and the first density index. We denote by $i_{ \partial \mathcal{N} }: \partial \mathcal{N} \rightarrow \mathcal{N} $ the inclusion.

\begin{theorem}[Covariant Eulerian reduction]\label{spacetime_reduced_EL} Let $ \mathscr{L} : J^1( \mathcal{D}  \times \mathcal{M} ) \times T \mathcal{D}  \times T^p_q\mathcal{D}  \times T^r_s \mathcal{M}  \rightarrow \wedge ^{n+1} \mathcal{D} $ be a Lagrangian density with the material covariant property \eqref{equivariance_general} and consider the associated spacetime Lagrangian density $\ell : T \mathcal{M}  \times  T ^p _q  \mathcal{M}  \times  T^r _s  \mathcal{M} \rightarrow \wedge  ^{n+1} \mathcal{M} $.

Fix the reference tensor fields $W \in \mathfrak{X} ( \mathcal{D} )$, $K \in \mathcal{T} _q^p( \mathcal{D} )$, and the spacetime tensor field $ \gamma \in \mathcal{T} ^r_s( \mathcal{M} )$.
For each world-tube $ \Phi : \mathcal{D} \rightarrow \mathcal{M} $, define $\mathcal{N}= \Phi ( \mathcal{D})$, $w= \Phi _* W$, $ \kappa = \Phi _ * K$, \textcolor{black}{and $ \partial _{\rm cont} \mathcal{N} = \Phi ([a,b] \times \partial \mathcal{B} )$ the spacetime region occupied by the boundary of the continuum}. Then, the following statements are equivalent:
\begin{itemize}
\item[\bf (i)] $\Phi$ is a critical point of the {\bfi Hamilton principle}
\[
\left. \frac{d}{d\varepsilon}\right|_{\varepsilon=0} \int_\mathcal{D}  \mathscr{L} ( j^1\Phi_ \varepsilon , W, K , \gamma \circ \Phi _ \varepsilon  )=0
\]
for arbitrary variations $\Phi _ \varepsilon  $ \textcolor{black}{with fixed endpoints at $ \lambda =a,b$}.

\item[\bf (ii)] $\Phi$ is a solution of the {\bfi Euler-Lagrange equations}
\[
\partial _a \frac{\partial \bar{\mathscr{L}}}{\partial \Phi^\mu_{,a}}  - \frac{\partial \bar{\mathscr{L}}}{\partial \Phi^\mu} = \frac{\partial \bar{\mathscr{L}}}{\partial \gamma } \partial_\mu \gamma , \qquad \frac{\partial \bar{ \mathscr{L} }}{\partial \Phi^\mu_{,a}}d^nX_a=0\quad \text{on} \quad  \textcolor{black}{[a,b] \times \partial \mathcal{B}}.
\]
\item[\bf (iii)] $w \in \mathfrak{X} (\mathcal{N} )$ and $ \kappa \in \mathcal{T} ^p_q( \mathcal{N} )$ are critical points of the {\bfi Eulerian variational principle} \color{black} 
\begin{equation}\label{Eulerian_VP}
\begin{aligned} 
&\!\!\left. \frac{d}{d\varepsilon}\right|_{\varepsilon=0}\int_{ \mathcal{N}_ \varepsilon } \ell\big( w_ \varepsilon , \kappa _ \varepsilon , \gamma \big)=0 \quad \text{for variations}\\
& \delta \mathcal{N} = \zeta |_{ \partial \mathcal{N} } \big/T \partial \mathcal{N} , \quad \delta w = -\pounds _ \zeta w, \quad   \delta \kappa = - \pounds _ \zeta \kappa,\phantom{\int_A^B}
\end{aligned} 
\end{equation}\color{black} 
\textcolor{black}{where $ \zeta$} is an arbitrary vector field on $ \mathcal{N} $ \textcolor{black}{such that $ \zeta |_{ \Phi (a, \mathcal{B} )}= \zeta |_{ \Phi (b, \mathcal{B} )}=0$}.
\item[\bf (iv)] $w \in \mathfrak{X} (\mathcal{N} )$ and $ \kappa \in \mathcal{T} ^p_q( \mathcal{N} )$ are solution of the {\bfi reduced Euler-Lagrange equations on spacetime}
\begin{equation}\label{spacetime_EL} 
\left\{
\begin{array}{l}
\displaystyle\vspace{0.2cm}\operatorname{div}^ \nabla \!\Big(  \ell \delta  + w \otimes \frac{\partial \ell}{\partial w} -  \textcolor{black}{  \frac{\partial \ell}{\partial \kappa }\!\therefore\!  \widehat{ \kappa  }}\Big)  =  \frac{\partial ^\nabla\!\ell}{\partial x}+\frac{\partial \ell}{\partial \gamma } : \nabla  \gamma\\
\displaystyle\vspace{0.2cm} i_{ \partial \mathcal{N} }^*\Big(\operatorname{tr}\Big(  \ell \delta  + w \otimes \frac{\partial \ell}{\partial w} -   \textcolor{black}{ \frac{\partial \ell}{\partial \kappa }\!\therefore\!  \widehat{ \kappa  } }\Big) \cdot \zeta\Big)=0, \;\;\forall \zeta \quad \text{on} \quad   \textcolor{black}{\partial_{\rm cont} \mathcal{N}}\\
\displaystyle\pounds _w \kappa =0,
\end{array}
\right.
\end{equation} 
written with the help of a torsion free covariant derivative $ \nabla $ \textcolor{black}{and with $ \delta $ the Kronecker delta\footnote{\textcolor{black}{No confusion should arise with the same symbol $ \delta $ also used for variations}.}}. In local coordinates, writing $\ell=\bar\ell d^{n+1}x$, the boundary condition reads
\[
\Big( \bar\ell  \delta  + w \otimes \frac{\partial \bar\ell}{\partial w}-   \textcolor{black}{\frac{\partial \bar{\ell}}{\partial \kappa }\!\therefore\! \widehat{ \kappa  } }\Big) ^\mu_\nu d^nx_\mu=0, \quad \text{on} \quad   \textcolor{black}{\partial_{\rm cont} \mathcal{N}}.
\]
\end{itemize}
\end{theorem}
\begin{proof} \textcolor{black}{First, let us note that the constrained variations in {\bf (iii)} are induced by the variations $\Phi _ \varepsilon  $ of the world-tube by using the relations $ \mathcal{N} _ \varepsilon = \Phi _ \varepsilon  ( \mathcal{D} )$, $w_ \varepsilon = (\Phi_ \varepsilon  )_*W$, $ \kappa _ \varepsilon = (\Phi_ \varepsilon  )_*K$ and defining $ \zeta= \delta \Phi \circ \Phi ^{-1}$}. For simplicity, we include only the material tensor $K$, since the treatment of the vector field $W$ is a particular case of it. Using Lemma \ref{technical_lemma} and \eqref{variation_moving_domain}, we have\color{black} 
\begin{align*} 
&\left. \frac{d}{d\varepsilon}\right|_{\varepsilon=0}\int_{ \Phi _ \varepsilon  ( \mathcal{D} )} \ell\big( (\Phi_ \varepsilon  )_*K, \gamma \big)\\
&\quad = \int_\mathcal{N}  \frac{\partial \ell}{\partial \kappa  } : \delta \kappa+ \int_{ \partial \mathcal{N} } \mathbf{i} _{ \delta \Phi \circ \Phi ^{-1} }\ell = - \int_\mathcal{N}  \frac{\partial \ell}{\partial \kappa  } : \pounds _ \zeta \kappa + \int_{ \partial \mathcal{N} } \mathbf{i} _{ \zeta }\ell\\
&\quad = - \int_\mathcal{N}  \Big(    \frac{\partial \ell}{\partial \kappa } :  \nabla _ \zeta \kappa - \operatorname{div}^ \nabla \!\Big( \frac{\partial \ell}{\partial \kappa } \!\therefore\!  \widehat{ \kappa  }\Big)  \cdot \zeta \Big)  + \int_{ \partial \mathcal{N} }\operatorname{tr}  \Big( \ell \delta -  \frac{\partial \ell}{\partial \kappa }\!\therefore\! \widehat{ \kappa  } \Big) \cdot \zeta   \\
&\quad = - \int_ \mathcal{N}  \Big(    \operatorname{div}^ \nabla \!\Big( \ell \delta  -   \frac{\partial \ell}{\partial \kappa }\!\therefore\!  \widehat{ \kappa  }\Big) \cdot \zeta - \frac{\partial ^ \nabla \!\ell}{\partial x} \cdot \zeta - \frac{\partial \ell}{\partial \gamma } : \nabla _ \zeta \gamma  \Big)  + \int_{ \partial \mathcal{N} } \operatorname{tr}  \Big( \ell \delta - \frac{\partial \ell}{\partial \kappa }\!\therefore\!  \widehat{ \kappa  } \Big) \cdot \zeta.
\end{align*} \color{black} 
In the third equality we used $\mathbf{i} _ \zeta \ell =\operatorname{tr}(\ell \delta ) \cdot  \zeta $. In the fourth equality we used the following formula
\begin{equation}\label{nabla_derivative_ell} \color{black} 
\nabla_ \zeta  [ \ell( \kappa , \gamma )]=  \frac{\partial ^\nabla\!\ell}{\partial x} \cdot \zeta +\frac{\partial \ell}{\partial \kappa } : \nabla _ \zeta \kappa + \frac{\partial \ell}{\partial \gamma } : \nabla _ \zeta \gamma 
\end{equation} 
for the derivative of the bundle map $\ell$ with respect to the given connection $ \nabla $, with $\frac{\partial ^ \nabla \ell}{\partial x}$ the derivative of $\ell$ with respect to the base point defined with the help of $ \nabla $. 

To include the reference vector field $W$ we note that for $w$ a vector field, we have
\[
\textcolor{black}{ \frac{\partial \ell}{\partial w}\!\therefore\!  \widehat{w}}= - w \otimes \frac{\partial \ell}{\partial w} .
\]
Since the vector field $ \zeta $ is arbitrary, we get the equations \eqref{spacetime_EL}. 
\end{proof}

The result stated in the above theorem still holds when material covariance is satisfied only with respect to the isotropy subgroup $ \operatorname{Diff}_{W,K}( \mathcal{D} )$ of the material tensor fields, see Remark \ref{remark_W_K}.

\begin{remark}[On the form of the boundary condition]\rm
We chose to write the boundary condition in \eqref{spacetime_EL} without introducing a Lorentzian metric, since we \textcolor{black}{did not} assume the presence of a metric in the continuum theory so far. If such a metric is present then, under the hypothesis that the boundary consist of nondegenerate hypersurfaces, the unit normal vector field can be used to rewrite the boundary condition in a more concrete way. This will be considered later, see Lemma \ref{Lemma_boundary}.
\end{remark}

\paragraph{Spacetime covariance.} We now examine the case in which the material Lagrangian density $\mathscr{L}$ in \eqref{Lagrangian_K}  is also spacetime covariant with respect to $ \operatorname{Diff}( \mathcal{M} )$, see \eqref{spacetime_equivariance}, in addition to the material covariance \eqref{equivariance_general}.

\begin{lemma}\label{spacetime_mat_ell} Let $ \mathscr{L} : J^1( \mathcal{D}  \times \mathcal{M} ) \times T \mathcal{D}  \times T^p_q\mathcal{D}  \times T^r_s \mathcal{M}  \rightarrow \wedge ^{n+1} \mathcal{D} $ be a \textcolor{black}{materially covariant} Lagrangian density and consider the associated spacetime Lagrangian $\ell : T \mathcal{M}  \times T ^p_q \mathcal{M} \times  T^r_s \mathcal{M}  \rightarrow  \wedge ^{n+1} \mathcal{M}$. Then if $ \mathscr{L} $ is also spacetime covariant, we have
\[
\psi ^\ast\left[ \ell (  \psi_\ast w, \psi_\ast \kappa , \psi _\ast \gamma ) \right] = \ell ( w, \kappa , \gamma ), \quad \forall\psi \in \operatorname{Diff}(\mathcal{M} ).
\]
Therefore, $\ell$ satisfies the following equalities:
\[
\frac{\partial ^ \nabla \!\ell}{\partial x}=0 \quad\text{and}\quad \textcolor{black}{ - w \otimes \frac{\partial \ell}{\partial w} +   \frac{\partial \ell}{\partial \kappa }\!\therefore\!  \widehat{ \kappa  } +   \frac{\partial \ell}{\partial \gamma  }\!\therefore\!   \widehat{ \gamma   } =\ell \delta }.
\]
\end{lemma} 
\begin{proof} The first assertion follows from the relation \eqref{scrL_to_ell} combined with the spacetime covariance condition \eqref{spacetime_equivariance}. Taking a path of diffeomorphisms $ \psi _ \varepsilon  \in \operatorname{Diff}( \mathcal{M} )$ passing through the identity at $\varepsilon =0$ and taking the $ \varepsilon $-derivative of the condition $\ell (  \psi_ \varepsilon ^\ast w, \psi_ \varepsilon ^\ast \kappa , \psi _ \varepsilon ^\ast \gamma )= \psi _ \varepsilon ^*[\ell(w, \kappa , \gamma )]$, we get $ \frac{\partial \ell}{\partial w} \cdot \pounds _ \zeta w + \textcolor{black}{ \frac{\partial \ell}{\partial  \kappa } : \pounds _ \zeta  \kappa + \frac{\partial \ell}{\partial  \gamma } : \pounds _ \zeta \gamma }= \operatorname{div}(\ell \zeta )$, for the vector field $ \zeta = \left. \frac{d}{d\varepsilon}\right|_{\varepsilon=0} \psi _ \varepsilon $, where $\operatorname{div}(\ell \zeta )$ is the divergence of the vector field density $\ell \zeta $, see \eqref{def_div}--\eqref{example_div}. Choosing a torsion free covariant derivative $ \nabla $ and using the formulas \eqref{global_formula}, $ \operatorname{div}(\ell \zeta ) = \nabla _ \zeta [\ell( \cdot )] + \ell  \operatorname{div}^ \nabla\!\zeta=\nabla _ \zeta [\ell( \cdot )] + \ell \delta :  \nabla \zeta  $, and \eqref{nabla_derivative_ell},  the requested identity follows since the vector field $ \zeta $ is arbitrary.
\end{proof}

In the next corollary, we obtain that the reduced spacetime Euler-Lagrange equations can be written exclusively in terms of the partial derivative with respect to $ \gamma $, when $\mathscr{L}$ satisfies both covariance properties. We recall that $ \gamma $ is not a variable in the spacetime description, but a given fixed parameter.

\begin{corollary}\label{corol_spacetime} Assume that the Lagrangian density \eqref{Lagrangian_K} is material and spacetime covariant. Then the reduced Euler-Lagrange equations \eqref{spacetime_EL} are equivalently written as
\begin{equation}\label{spacetime_EL_sc_gamma} \color{black} 
\left\{
\begin{array}{l}
\displaystyle\vspace{0.2cm}\operatorname{div}^ \nabla \!\Big(\frac{\partial \ell}{\partial \gamma  }\!\therefore\!  \widehat{ \gamma   }\Big) =\frac{\partial \ell}{\partial \gamma } : \nabla  \gamma\\
\displaystyle\vspace{0.2cm} i_{ \partial \mathcal{N} }^*\Big(\operatorname{tr}\Big(\frac{\partial \ell}{\partial \gamma  }\!\therefore\!  \widehat{ \gamma   } \Big) \cdot \zeta\Big)=0, \;\;\forall \zeta \quad \text{on} \quad\textcolor{black}{\partial_{\rm cont} \mathcal{N}} \\
\displaystyle\pounds _w \kappa =0.
\end{array}
\right.
\end{equation}
The local form of the boundary conditions is
\[
\textcolor{black}{\Big( \frac{\partial \bar{\ell}}{\partial \gamma  }\!\therefore\!  \widehat{ \gamma   }\Big)^\mu_\nu d^nx_\mu=0}.
\]
\end{corollary}

\medskip 

These results stated above still hold when material covariance is satisfied only with respect to the isotropy subgroup $ \operatorname{Diff}_{W,K}( \mathcal{D} )$ of the material tensor fields, see Remark \ref{remark_W_K}. Spacetime covariance with respect to the whole group $ \operatorname{Diff}( \mathcal{M} )$ is however needed.

\paragraph{The case of a Lorentzian metric on spacetime.} The equations obtained above hold for any spacetime tensor field $ \gamma $, not necessarily a metric. We shall consider in this paragraph the special case in which $ \gamma $ is a Lorentzian metric, denoted $ \gamma = \mathsf{g}$, hence the Lagrangian density, is a bundle map of the form
\begin{equation}\label{Lagrangian_K_g} 
\mathscr{L}: J^1( \mathcal{D} \times \mathcal{M} ) \times T \mathcal{D} \times T^p_q \mathcal{D} \times S^2_L \mathcal{M} \rightarrow \wedge ^{n+1} \mathcal{D},
\end{equation}
written as $\mathscr{L}(j^1\Phi, W,K, \mathsf{g} \circ \Phi  )$.

We shall need the following result to write the boundary condition when a Lorentzian metric is given.
Recall that a hypersurface $ \Sigma $ in $ \mathcal{M} $ is called nondegenerate if the induced metric $i_ \Sigma ^* \mathsf{g}$ is nondegenerate with $i_ \Sigma : \Sigma \rightarrow \mathcal{M} $ the inclusion. Given a normal vector field $n$ to a nondegenerate hypersurface $ \Sigma $, we use the notation
\[
\epsilon = \mathsf{g}(n,n) \in \{\pm1\}.
\]
Such an hypersurface is said to be timelike, resp., spacelike, if $i_ \Sigma ^* \mathsf{g}$ is Lorentzian ($ \epsilon =1$), resp. Riemannian ($ \epsilon =-1$).

\begin{lemma}\label{Lemma_boundary} Let $T$ be a $(1,1)$ tensor field density and assume that the piecewise smooth boundary $ \partial \mathcal{N} $ consists of nondegenerate hypersurfaces with respect to the Lorentzian metric $g$. Then we have
\begin{itemize}
\item[\bf (1)] $i_{ \partial \mathcal{N} } ^* \big[ \operatorname{tr}(T) \cdot \zeta \big]= \epsilon i_{ \partial \mathcal{N} } ^* \big[ \mathbf{i} _n(T( \zeta , n^\flat))\big]$, 
for all $ \zeta \in \mathfrak{X} ( \mathcal{N} )$, with $n$ a unit vector normal to $ \Sigma $ and $ \epsilon = \mathsf{g}(n,n) \in \{\pm1\}$.
\item[\bf (2)] $i_{ \partial \mathcal{N} } ^* \big[ \operatorname{tr}(T) \cdot \zeta \big]=0,\forall\; \zeta \in T \mathcal{N} |_{ \partial \mathcal{N} } $ $ \Longleftrightarrow $ $ T( \zeta , n^\flat )=0$ on $ \partial \mathcal{N} $, $\forall\; \zeta \in T \mathcal{N} |_{ \partial \mathcal{N} }$.
\end{itemize}
\textcolor{black}{Here $ \flat$ is the flat operator associated to $\mathsf{g}$.}
\end{lemma} 
\begin{proof} We write $T= \mathsf{T} \otimes \mu (\mathsf{g})$, where $ \mu (\mathsf{g})$ is the volume form associated to $\mathsf{g}$ and $\mathsf{T}$ is a $(1,1)$ tensor field. To prove \textbf{(1)} we consider the vector field $X=  \mathsf{T}( \zeta , \cdot )$. We assume that the boundary consists of spacelike and timelike pieces. Let $n$ be the outward pointing normal vector field to $ \partial \mathcal{N} $. We decompose $X$ as $X= \epsilon \mathsf{g}(X,n)n+ X^{\|}$ with $\epsilon =\mathsf{g}(n,n)$. We compute
\begin{align*} 
i_{ \partial \mathcal{N} } ^*\big[ \operatorname{tr}(T) \cdot \zeta \big]&= i_{ \partial \mathcal{N} } ^* \big[ \mathbf{i}_{\partial _ \mu }T(\zeta, dx^\mu) \big] = i_{ \partial \mathcal{N} } ^* \big[X(dx^\mu)  \mathbf{i}_{ \partial _ \mu } \mu (\mathsf{g})\big]\\
&=i_{ \partial \mathcal{N} } ^* \big[ \mathbf{i}_{ X } \mu (\mathsf{g})\big]= i_{ \partial \mathcal{N} } ^*\big [ \mathbf{i}_{  \epsilon g(X,n)n+ X^{\|} } \mu (\mathsf{g}) \big ]\\
&=  \epsilon  i_{ \partial \mathcal{N} } ^* \big[  \langle X,n ^\flat  \rangle\mathbf{i}_{  n} \mu (\mathsf{g}) \big] = \epsilon  i_{ \partial \mathcal{N} } ^* \big[ \mathbf{i}_{  n} (T(\zeta ,n^\flat))\big],
\end{align*}
for all $ \zeta \in \mathfrak{X} ( \mathcal{N} )$. To prove \textbf{(2)}, we note that
\[
\epsilon  i_{ \partial \mathcal{N} } ^* \big[ \mathbf{i}_{  n} (T(\zeta ,n^\flat))\big]= \epsilon  i_{ \partial \mathcal{N} } ^* \big[ \mathsf{T}(\zeta ,n^\flat)\mathbf{i}_{  n}  \mu (\mathsf{g})\big] =\epsilon \mathsf{T}(\zeta ,n^\flat)|_{ \partial \mathcal{N} } \, i_{ \partial \mathcal{N} } ^* (\mathbf{i}_{  n}  \mu (\mathsf{g})).
\]
Since $i_{ \partial \mathcal{N} } ^* (\mathbf{i}_{  n}  \mu (\mathsf{g}))$ is a volume form on the boundary $i_{ \partial \mathcal{N} } ^* \big[ \mathbf{i}_{  n} (T(\zeta ,n^\flat))\big]=0, \forall \,\zeta $ is equivalent to $\mathsf{T}(\zeta ,n^\flat)|_{ \partial \mathcal{N} }=0, \forall\, \zeta $. The latter is in turn equivalent to $ T( \zeta , n^\flat )=0$ on $ \partial \mathcal{N} $, $\forall\, \zeta \in T \mathcal{N} |_{ \partial \mathcal{N} }$.
\end{proof}

\begin{corollary}
Assume that the Lagrangian density \eqref{Lagrangian_K_g} is \textcolor{black}{materially covariant} and consider the associated spacetime Lagrangian density $\ell(w, \kappa , \mathsf{g})$. Then the reduced Euler-Lagrange equations are
\begin{equation}\label{spacetime_EL_g}  \color{black} 
\left\{
\begin{array}{l}
\displaystyle\vspace{0.2cm}\operatorname{div}^ \nabla \!\Big(  \ell \delta  + w \otimes \frac{\partial \ell}{\partial w} -   \frac{\partial \ell}{\partial \kappa }\!\therefore\! \widehat{ \kappa  }\Big)  =  \frac{\partial ^ \nabla \ell}{\partial x}\\
\displaystyle\vspace{0.2cm}  \Big(\ell  \delta  + w \otimes \frac{\partial  \ell}{\partial w}-   \frac{\partial  {\ell}}{\partial \kappa }\!\therefore\!  \widehat{ \kappa  }\Big)( \cdot  , n^\flat) =0 \quad  \text{on} \quad  \textcolor{black}{\partial_{\rm cont} \mathcal{N}}\\
\displaystyle\pounds _w \kappa =0,
\end{array}
\right.
\end{equation} 
where $ \nabla $ is the Levi-Civita covariant derivative and $n$ is a unit normal vector field to $ \partial_{\rm cont} \mathcal{N} $, all associated to $\mathsf{g}$.

If, in addition, the material Lagrangian density is also spacetime covariant, then these equations can be written as
\begin{equation}\label{spacetime_EL_sc_g} 
\operatorname{div}^ \nabla \! \frac{\partial \ell}{\partial \mathsf{g}} =0, \qquad \frac{\partial \ell}{\partial \mathsf{g}}( \cdot , n ^\flat )=0 \quad \text{on} \quad  \partial \mathcal{N}, \qquad \pounds _w \kappa =0.
\end{equation} 
\end{corollary}
\begin{proof} \textcolor{black}{We note that $\partial_{\rm cont} \mathcal{N}$ is made of timelike (hence nondegenerate) hypersurfaces}. When $ \gamma =\mathsf{g}$ is a Lorentzian metric, we can choose $ \nabla $ as the Levi-Civita covariant derivative of $\mathsf{g}$, so that $ \nabla \mathsf{g}=0$. Hence the first equation in \eqref{spacetime_EL} reduces to the first equation of \eqref{spacetime_EL_sc_g}. The boundary condition is obtained by using Lemma \ref{Lemma_boundary}.

In the spacetime equivariant case, we note that for $ \gamma =\mathsf{g}$,  we have
\[
\textcolor{black}{ \Big(   \frac{\partial \ell}{\partial \mathsf{g}  }\!\therefore\!  \widehat{ \mathsf{g}   } \Big) ^\mu_\nu}= 2\mathsf{g}_{\nu \gamma } \frac{\partial \ell}{\partial \mathsf{g}_{ \gamma \mu }} \]
from \eqref{hat_kappa}.  Hence, using also Lemma \ref{Lemma_boundary}, \eqref{spacetime_EL_sc_gamma} becomes \eqref{spacetime_EL_sc_g}.
\end{proof}



\subsection{Reduced convective Euler-Lagrange equations}

In this section we derive the convective form of the Hamilton principle \eqref{Ham_principle} for relativistic continua under the assumption of spacetime covariance. We use the relation \eqref{scrL_to_calL} between the material Lagrangian density and its convective version as well as the definition of the material form of the given spacetime tensor field.

\begin{theorem}\label{convective_reduced_EL} Let $ \mathscr{L} : J^1( \mathcal{D}  \times \mathcal{M} ) \times T \mathcal{D}  \times T^p_q\mathcal{D}  \times T^r_s \mathcal{M}  \rightarrow \wedge ^{n+1} \mathcal{D} $ be a spacetime covariant Lagrangian density and consider the associated covariantly reduced convective Lagrangian density $\mathcal{L} : T \mathcal{D}  \times  T ^p _q  \mathcal{D}  \times  T^r _s  \mathcal{D} \rightarrow \wedge  ^{n+1} \mathcal{D} $.

Fix the reference tensor fields $W \in \mathfrak{X} ( \mathcal{D} )$, $K \in \mathcal{T} _q^p( \mathcal{D} )$, and the spacetime tensor field $ \gamma \in \mathcal{T} ^r_s( \mathcal{M} )$.
For each world-tube $ \Phi : \mathcal{D} \rightarrow \mathcal{M} $, define $ \Gamma = \Phi ^* \gamma $. Then, the following statements are equivalent:
\begin{itemize}
\item[\bf (i)] $\Phi$ is a critical point of the {\bfi Hamilton principle}
\[
\left. \frac{d}{d\varepsilon}\right|_{\varepsilon=0} \int_\mathcal{D}  \mathscr{L} ( j^1\Phi_ \varepsilon , W, K , \gamma \circ \Phi _ \varepsilon  )=0
\]
for arbitrary variations $\Phi _ \varepsilon  $  \textcolor{black}{with fixed endpoints at $ \lambda =a,b$}.

\item[\bf (ii)] $\Phi$ is a solution of the {\bfi Euler-Lagrange equations}
\[
\partial _a \frac{\partial \bar{\mathscr{L}}}{\partial \Phi^\mu_{,a}}  - \frac{\partial \bar{\mathscr{L}}}{\partial \Phi^\mu} = \frac{\partial \bar{\mathscr{L}}}{\partial \gamma } \partial_\mu \gamma , \qquad \frac{\partial \bar{ \mathscr{L} }}{\partial \Phi^\mu_{,a}}d^nX_a=0 \quad\text{on}\quad \textcolor{black}{[a,b] \times \partial \mathcal{B}}.
\]
\item[\bf (iii)] $\Gamma\in \mathcal{T} ^r_s( \mathcal{D} )$ is a critical point of the {\bfi convective variational principle}
\[
\left. \frac{d}{d\varepsilon}\right|_{\varepsilon=0}\int_{ \mathcal{D} } \mathcal{L}\big( W,K, \textcolor{black}{ \Gamma }  \big)=0 \quad \textcolor{black}{ \text{for variations} \quad \delta \Gamma = \pounds _Z\Gamma},
\]
\textcolor{black}{where $ Z$} is an arbitrary vector field on $\mathcal{D}$ \textcolor{black}{vanishing at $ \lambda =a,b$}.
\item[\bf (iv)] $\Gamma\in \mathcal{T} ^r_s( \mathcal{D} )$ is a solution of the {\bfi reduced convective Euler-Lagrange equations}\color{black} 
\begin{equation}\label{convected_EL} 
\left\{
\begin{array}{l}
\displaystyle\vspace{0.2cm}\operatorname{div}^ \nabla \!\Big( \mathcal{L}  \delta  -   \frac{\partial   \mathcal{L} }{\partial \Gamma }\!\therefore\!  \widehat{ \Gamma  }\Big)  =  \frac{\partial ^ \nabla \! \mathcal{L} }{\partial X}+\frac{\partial   \mathcal{L} }{\partial W } \cdot \nabla  W+\frac{\partial   \mathcal{L} }{\partial K } \cdot \nabla  K\\
\displaystyle i_{ \partial \mathcal{D} } ^*\Big( \operatorname{tr}\Big(\frac{\partial   \mathcal{L} }{\partial \Gamma }\!\therefore\!  \widehat{ \Gamma  }\Big)\cdot Z\Big)=0,\;\; \forall Z \quad \text{on} \quad  \textcolor{black}{[a,b] \times   \partial \mathcal{B}},
\end{array}
\right.
\end{equation}
written with the help of a torsion free covariant derivative $ \nabla $ on $ \mathcal{D} $.  In local coordinates, denoting $ \mathcal{L} =\bar{ \mathcal{L} } d^{n+1}X$, the boundary condition reads
\[
\textcolor{black}{ \Big(\frac{\partial   \bar{\mathcal{L}} }{\partial \Gamma }\!\therefore\!  \widehat{ \Gamma  }\Big) ^a_b d^nX_a=0,\;\; \text{on $ \partial \mathcal{D} $}}.
\]
\end{itemize}
\end{theorem}
\begin{proof} The proof is left to the reader. It uses the results of \S\ref{technical_prelim} and is similar to that of Theorem \ref{spacetime_reduced_EL}.
\end{proof}

The result stated in the above theorem still holds when spacetime covariance is satisfied only with respect to the isotropy subgroup $ \operatorname{Diff}_\gamma ( \mathcal{M} )$ of the spacetime tensor field, see Remark \ref{remark_gamma}.

\paragraph{Material covariance.} We now examine the case in which the material Lagrangian density $\mathscr{L}$ in \eqref{Lagrangian_K} is also \textcolor{black}{materially covariant} with respect to $ \operatorname{Diff}( \mathcal{D} )$, see \eqref{equivariance_general}, in addition to the spacetime covariance \eqref{spacetime_equivariance}. The proofs are left to the reader.

\begin{lemma} Let $ \mathscr{L} : J^1( \mathcal{D}  \times \mathcal{M} ) \times T \mathcal{D}  \times T^p_q\mathcal{D}  \times T^r_s \mathcal{M}  \rightarrow \wedge ^{n+1} \mathcal{D} $ be a spacetime covariant Lagrangian density and consider the associated convective Lagrangian density $\mathcal{L} : T \mathcal{D}  \times T ^p_q \mathcal{D} \times  T^r_s \mathcal{D} \rightarrow  \wedge ^{n+1} \mathcal{D}$. Then if $ \mathscr{L} $ is also \textcolor{black}{materially covariant}, we have
\[
\varphi ^* \left[ \mathcal{L} (W,K,\Gamma) \right] =  \mathcal{L} (\varphi ^* W,\varphi ^* K,\varphi ^* \Gamma), \quad \forall \varphi  \in \operatorname{Diff}(\mathcal{D} ).
\]
Therefore, \textcolor{black}{$\mathcal{L}$ satisfies the following equalities}:
\[
\frac{\partial  ^ \nabla \! \mathcal{L} }{\partial X}=0 \quad\text{and}\quad - W \otimes \frac{\partial  \mathcal{L} }{\partial W} +   \textcolor{black}{ \frac{\partial  \mathcal{L} }{\partial K }\!\therefore\!   \widehat{ K  } +   \frac{\partial  \mathcal{L} }{\partial \Gamma  }\!\therefore\!   \widehat{ \Gamma   } =  \mathcal{L} \delta}.
\]
\end{lemma}

\begin{corollary} Assume that the Lagrangian density \eqref{Lagrangian_K} is material and spacetime covariant. Then the reduced convective Euler-Lagrange equations \eqref{convected_EL} are equivalently written as
\begin{equation}\label{convected_EL_cs_mc} \color{black} 
\left\{
\begin{array}{l}
\displaystyle\vspace{0.2cm}\operatorname{div}\Big(- W \otimes \frac{\partial  \mathcal{L} }{\partial W} +   \frac{\partial  \mathcal{L} }{\partial K }\!\therefore\!   \widehat{ K  } \Big) =\frac{\partial  \mathcal{L} }{\partial W } \cdot \nabla  W+\frac{\partial \mathcal{L} }{\partial K } \cdot \nabla  K \\
\displaystyle i_{ \partial \mathcal{D} } ^* \Big(\operatorname{tr}\Big(\mathcal{L}  \delta   + W \otimes \frac{\partial  \mathcal{L} }{\partial W} -   \frac{\partial  \mathcal{L} }{\partial K }\!\therefore\!   \widehat{ K  }\Big) \cdot Z \Big)=0, \;\;\forall Z \quad \text{on} \quad     \textcolor{black}{[a,b] \times   \partial \mathcal{B}}.
\end{array}
\right.
\end{equation}
The local form of the boundary conditions is
\[
\textcolor{black}{ \Big(\bar{\mathcal{L}}  \delta   + W \otimes \frac{\partial \bar{\mathcal{L}}}{\partial W} -   \frac{\partial  \bar{\mathcal{L}}}{\partial K }\!\therefore\!  \widehat{ K  }\Big)^a_b d^nX_a=0.}
\]
\end{corollary}

Note that, in a similar way with the case of Corollary \ref{corol_spacetime}, here $W$ and $K$ are not dynamic variables in the convective description, but given fixed parameters.
Also, these results still hold when the spacetime covariance is satisfied only with respect to the isotropy subgroup $ \operatorname{Diff}_{\gamma }( \mathcal{M} )$ of the spacetime tensor fields, see Remark \ref{remark_gamma}. Material covariance with respect to the whole group $ \operatorname{Diff}( \mathcal{D} )$ is however needed.

In the Figure below we give the analogue of Fig. \ref{figure_1} in the relativistic case.

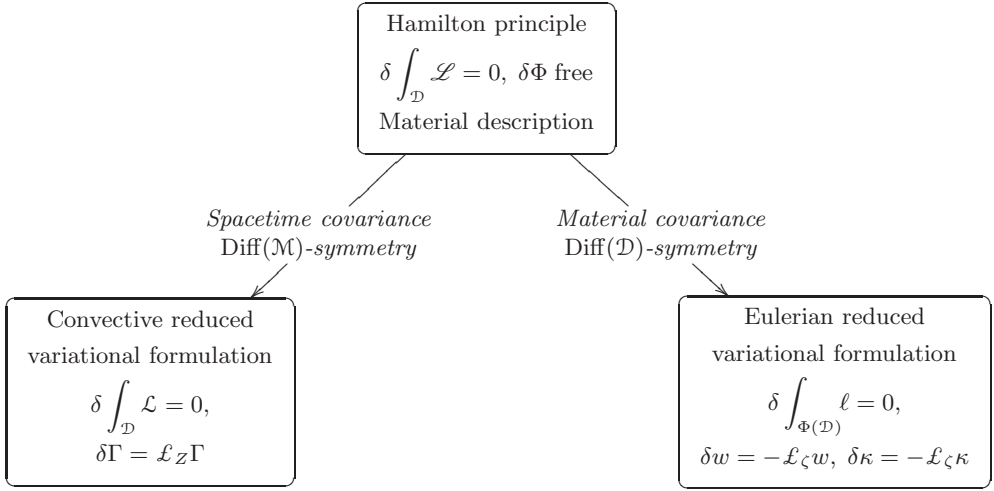
\begin{figure}[h!]
{\noindent
\footnotesize
\begin{center}
\hspace{.3cm}
\begin{xy}
\xymatrix{
& *+[F-:<3pt>]{
\begin{array}{c}
\vspace{0.1cm}\text{Hamilton principle}\\
\vspace{0.1cm}\displaystyle \delta \int_ \mathcal{D}  \mathscr{L}=0,\; \delta \Phi  \;\text{free}\\
\vspace{0.1cm}\text{Material description}
\end{array}
}
\ar[ddl]|{\begin{array}{c}\textit{Spacetime covariance} \\
\textit{$ \operatorname{Diff}( \mathcal{M}  )$-symmetry}\\
\end{array}}
\ar[ddr]|{\begin{array}{c}\textit{Material covariance} \\
\textit{$ \operatorname{Diff}( \mathcal{D} )$-symmetry}\\
\end{array}} & & \\
& & & \\
*+[F-:<3pt>]{
\begin{array}{c}
\vspace{0.1cm}\text{Convective reduced}\\
\vspace{0.1cm}\text{variational formulation}\\
\vspace{0.1cm}\displaystyle\delta \int_ \mathcal{D}  \mathcal{L}=0,\\
\vspace{0.1cm}\delta \Gamma = \pounds _ Z  \Gamma 
\\
\end{array}
}
& &  *+[F-:<3pt>]{
\begin{array}{c}
\vspace{0.1cm}\text{Eulerian reduced}\\
\vspace{0.1cm}\text{variational formulation}\\
\vspace{0.1cm}\displaystyle\delta \int_ { \Phi  (\mathcal{D})}\!\!\ell =0,\\
\vspace{0.1cm}\delta w= - \pounds _ \zeta w, \; \delta \kappa  = - \pounds _ \zeta   \kappa 
\\
\end{array}
}\\
}
\end{xy}
\end{center}
}
\caption{Illustration of the variational principles in the three representations of relativistic continuum mechanics for general material and spacetime tensor fields $K$, $W$, $ \gamma $ with associated dynamic fields $\kappa= \Phi _*K$, $w= \Phi _*W$, $ \Gamma = \Phi ^* \gamma $. Compare to Fig. \ref{figure_1}.}
\label{figure_2}
\end{figure}

\section{Coupling with the Einstein equations and junction conditions}

In this section, we shall show how the variational formulation for continua developed above can be coupled to the gravitation theory. In particular, the variational formulation that we develop in this section is able to produce
\begin{itemize}
\item[\bf (1)] the field equations for the gravitational field created by the relativistic continuum, both at the interior and outside the continuum;
\item[\bf (2)] the equations of motion of the continuum in this gravitational field;
\item[\bf (3)] the junction conditions between the solution at the interior of the relativistic continuum and the solution describing the gravity field produced outside from it.
\end{itemize}

\paragraph{Lagrangian densities.} In general, when the variational setting developed in \S\ref{sec_3}-\S\ref{sec_4} is coupled with field theories, the spacetime tensor fields $ \gamma $ that appear in the Lagrangian densities are not fixed and \textcolor{black}{their dynamics} is governed by Euler-Lagrange equations on spacetime associated to the variations $ \delta \gamma $. We shall restrict here to the coupling with gravitation, i.e., take $ \gamma =\mathsf{g}$ a Lorentzian metric. We thus consider the Lagrangian density of the continuum as a bundle map
\[
\mathscr{L} : J^1( \mathcal{D}  \times \mathcal{M} ) \times T \mathcal{D}  \times T^p_q\mathcal{D}  \times S^2_L \mathcal{M}   \rightarrow \wedge ^{n+1} \mathcal{D}
\]
written as $\mathscr{L} (j^1\Phi, W, K, \mathsf{g} \circ \Phi  )$. For simplicity of the exposition we assume that $\mathscr{L}$ is \textcolor{black}{materially covariant} with respect to the isotropy subgroup of $(W,K)$ so that there is the associated spacetime Lagrangian density 
\begin{equation}\label{spacetime_Lagran_grav} 
\ell: T \mathcal{M}  \times  T ^p _q  \mathcal{M}  \times  S^2_L \mathcal{M} \rightarrow \wedge  ^{n+1} \mathcal{M}
\end{equation} 
written as $\ell(w, \kappa , \mathsf{g})$, see \S\ref{subsec_mat_cov}.

Recall that the Einstein equations in the vacuum are the Euler-Lagrange equations for the Einstein-Hilbert Lagrangian 
\begin{equation}\label{l_geom} 
\ell_{\rm EH}:J^2S^2_L \mathcal{M}  \rightarrow \wedge ^{n+1} \mathcal{M} , \qquad \ell_{\rm EH}(j^2\mathsf{g}) = \frac{1}{2 \chi } R(\mathsf{g}) \mu (\mathsf{g}),
\end{equation} 
defined on the second order jet bundle of $S^2_L \mathcal{M} \rightarrow \mathcal{M} $. Here $\chi = 8 \pi G c^{-4}$ and $R(\mathsf{g})$ is the scalar curvature of $\mathsf{g}$. Recall that it is defined as $R=\mathsf{g}^{ \alpha \beta } Ric_{ \alpha \beta }= \mathsf{g}^{ \alpha \beta } R^ \lambda _{ \alpha \lambda \beta }$ with $R(u,v)= \nabla _u \nabla _v - \nabla _v \nabla _u - \nabla _{[u,v]}$ the Riemann curvature tensor of $\mathsf{g}$, with local coordinates $R^ \lambda {}_{\alpha \mu \beta }= \partial _ \mu \Gamma ^ \lambda _{ \alpha \beta } -  \partial _ \beta  \Gamma ^ \lambda _{ \alpha \mu  }+ \Gamma ^ \lambda _{ \sigma  \mu  }\Gamma ^ \sigma  _{ \alpha \beta }  - \Gamma ^ \lambda _{ \sigma  \beta   }\Gamma ^ \sigma  _{ \alpha \mu  }$, \textcolor{black}{for $ \Gamma^ \lambda _{ \alpha \beta }$ the Christoffel symbols}.
The spacetime covariance of $\ell_{\rm EH}$ reads
\begin{equation}\label{sc_geom}
 \ell_{\rm EH}(j^2( \varphi ^* \mathsf{g}))= \varphi ^* \big[\ell_{\rm EH}(j^2\mathsf{g})\big], \;\;\forall\; \varphi \in \operatorname{Diff}(\mathcal{M} ).
\end{equation}
We recall that the Einstein tensor field $\textcolor{black}{Ein} ^{ \alpha \beta }= Ric^{ \alpha \beta }- \frac{1}{2} R \mathsf{g}^{ \alpha \beta }$ satisfies the Bianchi identities
\[
\textcolor{black}{Ein} ^{ \alpha \beta }{}_{; \beta }=0,
\]
which can be seen as a direct consequence of \eqref{sc_geom}. \textcolor{black}{The semicolon in front of the index $ \beta $ denotes the covariant differentiation associated with $\nabla $.}

\paragraph{Gibbons-Hawking-York boundary term.} In order to couple continuum mechanics with gravitation we shall modify the Einstein-Hilbert action in order to take into account of the boundary of the spacetime $ \mathcal{N} $ occupied by the continuum.
We shall consider a well-known modification due to \cite{Yo1972}, \cite{GiHa1977} for spacetimes with boundary, by adding the integral of the trace of the extrinsic curvature (or second fundamental form) on the boundary of the continuum.

We recall the definition and our conventions on extrinsic curvature. Consider a nondegenerate hypersurface $ \Sigma \subset \mathcal{M} $, choose a normal vector field $n$, and denote $ \epsilon =\mathsf{g}(n,n) \in \{\pm 1\}$. This choice together with the choice of an orientation of $\mathcal{M} $ fixes the orientation of $ \Sigma $. The orthogonal decomposition $T \mathcal{M} |_{\Sigma }= \textcolor{black}{T \Sigma \oplus T \Sigma ^\perp}$ reads \textcolor{black}{$w=w^\parallel + w^\perp$ with $w^\parallel= w- \epsilon \mathsf{g}(w,n) n$ and $w^\perp = \epsilon  \mathsf{g}(w,n) n$}. Given $u,v \in \mathfrak{X}( \Sigma )$, and an extension $\tilde v$ of $v$, to a neighborhood of $ \Sigma $ in $M$, we have the orthogonal decomposition
\[
\nabla _u \tilde v= (\nabla _u \tilde v)^ \parallel + ( \nabla _u \tilde v)^\perp = \nabla_u^{\Sigma } v - \epsilon K(u,v)n,
\]
where $ \nabla ^{\Sigma }$ is the Levi-Civita covariant derivative for the (Riemannian or Lorantzian) metric $ h= i_ {\Sigma }^*g$ and $K(u,v)= - \mathsf{g}( \nabla _u \tilde v, n) = \mathsf{g}( u, \nabla _v n)$ is the extrinsic curvature of $\Sigma $. The trace of $K$ with respect to $h$ is denoted $k= \operatorname{Tr}K= h^{ab}K_{ab}$, where $ h_{ab}dx^adx^b$ is the local expression of $h$ on $\Sigma $. Note that $k$ is, up to a constant factor, the mean curvature of $ \Sigma$. We shall use the notations $K(\mathsf{g})$ and $k(\mathsf{g})$ to emphasize the dependency on Lorentzian metric $\mathsf{g}$. We will also consider the trace $ \operatorname{Tr}(K^2)= h^{ab}K_{bc}h^{cd}K_{da} $.
The Gibbons-Hawking-York (GHY) term associated to an oriented nondegenerate hypersurface $ \Sigma $ is
\begin{equation}\label{GHY_term} 
\frac{1}{ \chi } \int_{ \Sigma }  \epsilon k(\mathsf{g}) \mu(i^* _{ \Sigma }\mathsf{g}),
\end{equation}
with $\mu(i^* _{ \Sigma }\mathsf{g})$ the volume form associated to $i^* _{ \Sigma }\mathsf{g}$ and to the orientation of $ \Sigma $.

\paragraph{Action functional for general relativistic continuum.} We shall denote by $ \mathcal{N} ^-= \Phi ( \mathcal{D} )\subset \mathcal{M} $ the portion of spacetime occupied by the continuum, denoted $\mathcal{N} $ earlier, and we shall write $ \mathcal{N} ^+= \mathcal{M} - \operatorname{int}(\mathcal{N}^-)$. For simplicity we assume that the spacetime $ \mathcal{M} $ has no boundary and we have $ \partial \mathcal{N} ^+= \partial \mathcal{N} ^-$ assumed to be piecewise smooth. While $ \partial \mathcal{N}^+$ and $\partial \mathcal{N} ^-$ are equals as manifolds, they have opposite orientation as boundaries of $ \mathcal{N} ^+$ and $ \mathcal{N} ^-$, the latter having the orientation induced from that of $ \mathcal{M} $.

We denote by $\mathsf{g}^-$ and $\mathsf{g}^+$ the Lorentzian metrics on $ \mathcal{N}^-$ and $ \mathcal{N} ^+$, assumed to be smooth. It is assumed that $\mathsf{g}^+$ and $\mathsf{g}^-$ induce the same metric on the boundary, i.e.,
\begin{equation}\label{junction_first} 
i_{ \partial \mathcal{N}^-} ^* \mathsf{g}^-= i_{ \partial \mathcal{N} ^+} ^* \mathsf{g}^ +=:h
\end{equation} 
with $i_{ \partial \mathcal{N}^\pm}: \partial \mathcal{N}^\pm \rightarrow \mathcal{N}^\pm $ the inclusions.
We also assume that $h$ is nondegenerate, i.e., either Lorentzian or Riemannian, on each piece of the boundary. Condition \eqref{junction_first}, sometimes referred to as the preliminary junction condition, ensures that the boundary has a well-defined geometry. It is valid even in the presence of a singular matter distribution on the boundary. The remaining content of the Israel-Darmois conditions, which is sometimes referred to as the Lanczos-Israel condition, will be obtained from the variational formulation.
 
The total action functional is constructed by adding the action functional of the continuum with Lagrangian density \eqref{spacetime_Lagran_grav} on $ \mathcal{N}^- = \Phi ( \mathcal{D} )$, the Einstein-Hilbert action functionals on $ \mathcal{N} ^-$ and $ \mathcal{N} ^ +$, as well as the corresponding GHY boundary terms, which yields
\begin{equation}\label{total_action}
\begin{aligned} 
&\int_{ \mathcal{N}^- }\! \ell( w , \kappa  , \mathsf{g}^- ) + \frac{1}{2 \chi }\int_{ \mathcal{N}^-}  \!R(\mathsf{g}^-) \mu (\mathsf{g}^-)+ \frac{1}{2 \chi }\int_{\mathcal{N} ^+} \! R(\mathsf{g}^+) \mu (\mathsf{g}^+)\\
&\qquad \qquad  \qquad \qquad \qquad \qquad + \frac{1}{ \chi }\int_{ \partial \mathcal{N}^-} \!\epsilon\, k(\mathsf{g}^-)   \mu ^-(h)+ \frac{1}{ \chi }\int_{ \partial  \mathcal{N} ^+} \!\epsilon\, k(\mathsf{g}^+)   \mu ^+(h).
\end{aligned}
\end{equation}
In the last two terms of \eqref{total_action}, $ \epsilon $ is $1$, resp., $-1$, on the timelike, resp., spacelike piece of the boundary, and $ \mu ^\pm(h)$ is the volume form on $ \partial \mathcal{N} ^\pm$ associated to $h$ and to the boundary orientation of $ \partial \mathcal{N} ^\pm$.

The case when the spacetime $ \mathcal{M} $ itself has a boundary can be easily considered by adding the appropriate GHY term associated to $ \partial  \mathcal{M} $.

Below we shall derive the equations as well as the boundary and junction conditions by extending the covariant reduced variational principle developed above to the action functional \eqref{total_action}. The presence of the GHY is essential for this derivation, as we shall see below.  While the variation of the GHY term with respect to the metric is well-known, we shall derive the variation with respect to the hypersurface, thereby extending previous results.

\subsection{First variation of the Gibbons-Hawking-York term}

In this section we consider a smooth nondegenerate hypersurface $ \Sigma$ in the spacetime $( \mathcal{M} , \mathsf{g})$ and denote $h= i_ \Sigma ^* \mathsf{g}$ the induced metric on $ \Sigma $. We assume that $ \Sigma $ has a piecewise smooth boundary $ \partial \Sigma $, consisting of nondegenerate pieces, with outward pointing unit normal vector field $ \nu $. We write $ \sigma = h( \nu , \nu ) \in \{\pm 1\}$ along $ \partial \Sigma $ and denote $\textcolor{black}{\vartheta} = i_{ \partial \Sigma } ^* h$ the induced metric on $ \partial \Sigma $.

The orientation of $ \mathcal{M} $ and the choice of a normal vector field $n$ on $ \Sigma $ determines the orientations of $ \Sigma $ and of $ \partial \Sigma $. We denote by $ \mu (h)$ and $ \mu ( \textcolor{black}{\vartheta})$ the volume forms associated to the metrics and these orientations.

\paragraph{First variation with respect to the metric.} We recall the variation of the GHY boundary term with respect to the Lorentzian metric in the Lemma below. \textcolor{black}{We refer to \cite{LeMyPoSo2016}, especially (2.19)--(2.26), for the main formulas giving this result, stated here in an intrinsic form}.

\begin{lemma}\label{variation_GBY_g} The first variation of the GHY term with respect to the spacetime metric $\mathsf{g}$ is
\[
2\! \left. \frac{d}{d\varepsilon}\right|_{\varepsilon=0}  \int_{\Sigma } k(\mathsf{g}_ \varepsilon ) \mu (h_ \varepsilon ) =\int_{\Sigma } \big(  (K - k h)\!: \!\delta h ^{-1}  - \mathsf{g}(\deltabar V,n) \big)  \mu (h) -  \int_{ \partial \Sigma }\! \sigma \, h(\deltabar W, \nu ) \mu ( \textcolor{black}{\vartheta} ),
\]
where the vector fields $\deltabar V$ and $\deltabar W$ on $ \mathcal{M} $ and $ \Sigma $ are defined from the variation $ \delta \mathsf{g}$ as
\[
\deltabar V^ \mu = \mathsf{g}^{ \alpha \beta } \delta \Gamma ^ \mu _{ \alpha \beta } \textcolor{black}{ - \mathsf{g}^{ \alpha \mu } \delta \Gamma ^ \beta _{ \alpha \beta  }} \quad\text{and}\quad \deltabar W = \left( i_ \Sigma  ^* ( \mathbf{i} _ n \delta \mathsf{g}) \right) ^{\sharp_h} .
\]
\end{lemma} 
We note that in the last term $ \sigma $ can take the values $1$ or $-1$ on various parts of $ \partial \Sigma $.

\paragraph{First variation with respect to the hypersurface.} We derive in this paragraph the variation of the GHY boundary terms with respect to the hypersurface. The resulting first variation formula extends earlier results obtained for the case of mean curvature of surfaces in the Euclidean space, see \cite{DoNo2011} and reference therein, and in a space of constant sectional curvature, see \cite{GrToTr2019}. Our formula shows the occurrence of the Ricci tensor of the ambient Lorentzian metric. As opposed to most of the earlier approaches, our proof is elementary, based on a direct computation. 

\begin{theorem}\label{var_GHY} The first variation of the GHY term with respect to the hypersurface is
\begin{align*} 
\left. \frac{d}{d\varepsilon}\right|_{\varepsilon=0} \int_{ \varphi _ \varepsilon ( \Sigma )} k (\mathsf{g}) \mu (h) &= \int_ \Sigma \delta \varphi ^\perp\big( k ^2 - \operatorname{Tr}(K ^2 ) -  Ric\big(n,n)\big) \mu (h)\\
& \qquad \qquad + \int_{ \partial \Sigma } \sigma \,h\big(  k \delta \varphi ^{\|} - \epsilon \operatorname{grad}^ \Sigma (\delta \varphi ^\perp), \nu \big) \mu ( \textcolor{black}{\vartheta} ),
\end{align*} 
where $ \varphi _ \varepsilon $ is a path of diffeomorphisms of $ \mathcal{M} $ with $ \varphi _{0}=id$ and we use the orthogonal decomposition
\[
\delta \varphi |_ \Sigma = \delta \varphi ^\perp n + \delta \varphi ^{\|} .
\]
\end{theorem}
The proof proceeds in several lemmas. Some of the computations will be done in local charts $\psi _ \Sigma : U \subset \Sigma \rightarrow \psi _ \Sigma(U) \subset   \mathbb{R} ^n$ of $ \Sigma $ and $ \psi _ \mathcal{M} : V \subset  \mathcal{M}  \rightarrow \psi _ \mathcal{M} (V) \subset  \mathbb{R} ^{n+1}$ of $ \mathcal{M} $. By appropriate restriction, we get the local representation of $ \varphi _ \varepsilon $ as
\[
\varphi _ \varepsilon ^{\rm loc}=\psi _ \mathcal{M}  \circ \varphi  _ \varepsilon \circ \psi _ \Sigma  ^{-1} : \psi_ \Sigma (U) \subset \mathbb{R} ^n \rightarrow \psi _ \mathcal{M} (V) \subset \mathbb{R} ^{n+1}.
\]
For brevity, we shall suppress the notation $( \cdot )^{\rm loc}$ everywhere. Without loss of generality, we can assume that locally
\[
\varphi _ \varepsilon (y)= \varphi (y) + \varepsilon  \big( f(y)n( \varphi (y))+ T _y\varphi (X(y))\big),
\]
where $f= \delta \varphi ^\perp \in C^\infty( \Sigma )$ and $X= \delta \varphi ^{\|}  \in \mathfrak{X} ( \Sigma )$.

\begin{lemma}\label{lemma_1} We have the following formulas:
\begin{itemize}
\item[\bf (1)] $\displaystyle\left. \frac{D}{D\varepsilon}\right|_{\varepsilon=0} T \varphi _ \varepsilon (v)= (\mathbf{d} f  \cdot v) n + f \nabla_v n + \nabla _v^ \Sigma X - \epsilon K(v,X)n$, for all $v \in T \Sigma $;
\item[\bf (2)] $\displaystyle\left. \frac{D}{D\varepsilon}\right|_{\varepsilon=0} n_ \varepsilon \circ  \varphi _ \varepsilon =  - \epsilon \operatorname{grad}^ \Sigma f + (\mathbf{i} _XK)^{\sharp_h}$ on $ \Sigma $;
\end{itemize}
where $ \nabla $ and $ \nabla ^ \Sigma $ are Levi-Civita covariant derivative associated to $\mathsf{g}$ on $ \mathcal{M} $ and to $h$ on $ \Sigma $, where $\frac{D}{D\varepsilon} \alpha ( \epsilon) $ denotes the covariant derivative of a curve $ \alpha ( \varepsilon ) \in T \mathcal{M} $ with respect to $ \nabla $, \textcolor{black}{and $ \mathbf{d} $ denotes the exterior derivative.}
\end{lemma} 

\begin{proof} For {\bf(1)}, we compute
\begin{align*} 
&\left. \frac{D}{D\varepsilon}\right|_{\varepsilon=0} T _y\varphi _ \varepsilon (v)= \left. \frac{D}{D\varepsilon}\right|_{\varepsilon=0} \left.\frac{d}{dt} \right|_{t=0} \varphi _ \varepsilon (y+ t v)= \left. \frac{D}{Dt}\right|_{t=0} \left.\frac{d}{d\varepsilon } \right|_{\varepsilon =0} \varphi _ \varepsilon (y+ t v)\\
&= \left. \frac{D}{Dt}\right|_{t=0} \big( f(y+ tv)n( \varphi (y+ t v))+ T_{y+ t v} \varphi (X(y+ t v))\big)\\
&=( \mathbf{d} f \cdot v) n( \varphi (y))+ f(y) \nabla _{ T_y \varphi ( v)} n + \nabla _{ T_y \varphi ( v)} T \varphi (X)(y)\\
&= (\mathbf{d} f \cdot v) n( \varphi (y))+ f(y) \nabla _{ T_y \varphi ( v)} n + T _y\varphi (\nabla ^ \Sigma _{v} X) - \epsilon K(v, X) n( \varphi (y)).
\end{align*} 
To prove {\bf(2)} we first note that from the equality $\mathsf{g}( \varphi _ \varepsilon (y))\left( n_ \varepsilon ( \varphi _ \varepsilon (y)), T_y \varphi _ \varepsilon ( v) \right) =0$, $\forall v \in T \Sigma$, we get
\begin{align*}
&\mathsf{g}( \varphi (y))\Big( \left. \frac{D}{D\varepsilon}\right|_{\varepsilon=0}  \!\!n_ \varepsilon ( \varphi _ \varepsilon (y))  , T_y \varphi  (v) \Big) + \mathsf{g}( \varphi  (y))\Big( n ( \varphi  (y)), \left. \frac{D}{D\varepsilon}\right|_{\varepsilon=0} \!\!\textcolor{black}{T _y\varphi _ \varepsilon} (v) \Big) =0, \;\; \forall v.
\end{align*}
By using {\bf(1)} in the second term, we get
\begin{equation}\label{intermediate_step} 
\mathsf{g}( \varphi (y))\Big( \left. \frac{D}{D\varepsilon}\right|_{\varepsilon=0} \!\! n_ \varepsilon ( \varphi _ \varepsilon (y)) , T _y\varphi  ( v) \Big)=-  \epsilon \mathbf{d} f \cdot  v  +   K(v, X), \;\; \forall v.
\end{equation} 
From the equality $\mathsf{g}( \varphi _ \varepsilon (y))\left( n_ \varepsilon ( \varphi _ \varepsilon (y)),n_ \varepsilon ( \varphi _ \varepsilon (y))\right) = \epsilon $, we get
\[
\mathsf{g}( \varphi  (y))\Big( \left. \frac{D}{D\varepsilon}\right|_{\varepsilon=0}\!\! n_ \varepsilon ( \varphi _ \varepsilon (y)),n ( \varphi  (y))\Big)=0
\]
hence we can write $ \left. \frac{D}{D\varepsilon}\right|_{\varepsilon=0} n_ \varepsilon \textcolor{black}{ ( \varphi _ \varepsilon (y))}= T_y \varphi (W(y))$ for some vector field $W$ on $ \Sigma $. Inserting this equality in \eqref{intermediate_step} and using $\mathsf{g}( \varphi (y))\left(T_y \varphi (W(y)), T _y\varphi  (v) \right)=h(y)(W(y), v)$, we get the equality claimed in {\bf(2)}. \end{proof}

\begin{lemma}\label{var_h_mu} The variation of the induced metric $h= i_{ \Sigma } ^* \mathsf{g}$ and associated volume $ \mu (h)$ with respect to the hypersurface are
\begin{itemize}
\item[\bf (1)] $\delta h = 2 f K + \pounds _ X h$;
\item[\bf (2)] $\delta  (\mu (h))= \left( f k + \operatorname{div}^ \Sigma  X \right)  \mu (h)$.
\end{itemize}
\end{lemma} 
\begin{proof} To prove {\bf(1)}, we note that for all $u,v \in T \Sigma $, we have
\begin{align*} 
\delta h(y)(u,v)&= \left. \frac{d}{d\varepsilon}\right|_{\varepsilon=0} \mathsf{g}( \varphi _{\varepsilon}(y))\big( T_y \varphi _{\varepsilon}(u), T_y \varphi_{\varepsilon} ( v)\big)\\
&=\mathsf{g}( \varphi (y))\Big(  \left. \frac{D}{D\varepsilon}\right|_{\varepsilon=0}T_y \varphi_{ \varepsilon } ( u), T_y \varphi ( v)\Big) \\
& \qquad \qquad \qquad   + \mathsf{g}( \varphi (y))\Big(   T_y \varphi (u), \left. \frac{D}{D\varepsilon}\right|_{\varepsilon=0}T_y \varphi_{ \varepsilon } ( v)\Big).
\end{align*} 
Using Lemma \ref{lemma_1} {\bf(1)}, we get
\begin{align*} 
\delta h(y)(u,v)&=2 f(y)K(y)( u, v) + h(y)\big(  \nabla ^ \Sigma _ u X, v  \big) + h(y)\big( u ,  \nabla ^ \Sigma _ v X \big)\\
&=2 f(y)K(y)( u, v)  + \pounds _X h( u, v).
\end{align*}
To prove {\bf(2)}, we first use {\bf(1)} and $h^{ab}\pounds _X h_{ab}= h^{ab}h_{cb}X^c_{;a}+ h^{ab}h_{ac}X^c_{;b}= 2 \operatorname{div}^ \Sigma X$, \textcolor{black}{with the semicolon indicating the covariant differentiation associated with $ \nabla ^ \Sigma $}, to compute
\begin{align*}
\delta \operatorname{det}h&= (\operatorname{det}h) h^{ab} \delta h_{ab}=  (\operatorname{det}h) h^{ab}( 2 f K_{ab} + (\pounds _X h)_{ab})=2(\operatorname{det}h) ( f k  +  \operatorname{div}^ \Sigma X).
\end{align*}
Then, for $\operatorname{det}(h)>0$, we have
\begin{align*}
\delta \sqrt{ \operatorname{det}(h)}&=\frac{1}{2\sqrt{ \operatorname{det}(h)} }2(\operatorname{det}h) ( f k  +  \operatorname{div}^ \Sigma X)=\sqrt{ \operatorname{det}(h)} ( f k  +  \operatorname{div}^ \Sigma X),
\end{align*} 
from which {\bf (2)} follows, similarly for $ \operatorname{det}(h)<0$. 
\end{proof}
 
\begin{lemma}\label{delta_Kk}  The variation, with respect to the hypersurface, of the extrinsic curvature $K$ and its trace $k$ is
\begin{itemize}
\item[\bf (1)] $\delta K = f K ^2   - f R\textcolor{black}{ (n ,\cdot \,, n, \cdot \, ) }  -  \epsilon   \nabla ^\Sigma \mathbf{d} f   + \pounds _ X K $;
\item[\bf (2)] $\delta k = -f \operatorname{Tr}_h(K ^2 ) - f Ric(n,n)  - \epsilon \Delta ^ \Sigma f  + \mathbf{d} k \cdot X$;
\end{itemize}
where $K ^2 _{ab}=K_{ae}h^{ef}K_{fb}$.
\end{lemma} 
\begin{proof} {\bf(1)} From the definition of the extrinsic curvature
\begin{align*} 
K(y)(u,v)&= \mathsf{g}( \varphi (y)) \left( T_y \varphi ( u), \nabla _{ T _y\varphi ( v)} n \right) =  \mathsf{g}( \varphi (y)) \Big( T_y \varphi (u), \left.\frac{D}{Dt} \right|_{t=0}  n( \varphi (y+ t v)) \Big),
\end{align*} 
we get
\begin{equation}\label{two_terms} 
\begin{aligned}
\delta K(y)(u,v)&= \left. \frac{d}{d\varepsilon}\right|_{\varepsilon=0}  \mathsf{g}( \varphi _ \varepsilon (y)) \Big( T_y \varphi_ \varepsilon  ( u), \left.\frac{D}{Dt} \right|_{t=0}  n_ \varepsilon ( \varphi_ \varepsilon  (y+ t v)) \Big)\\
&= \mathsf{g}( \varphi (y)) \Big( \left. \frac{D}{D\varepsilon}\right|_{\varepsilon=0} T_y \varphi _ \varepsilon ( u), \left.\frac{D}{Dt} \right|_{t=0}  n ( \varphi (y+ t v)) \Big)\\
& \qquad +  \mathsf{g}( \varphi (y)) \Big( T_y \varphi  (u), \left. \frac{D}{D\varepsilon}\right|_{\varepsilon=0} \left.\frac{D}{Dt} \right|_{t=0}  n_ \varepsilon ( \varphi _ \varepsilon (y+ t v)) \Big).
\end{aligned}
\end{equation}

\noindent (1.A) We start by treating the first term in \eqref{two_terms}. Since the vector field in the second slot is parallel to $ \Sigma $, by Lemma \ref{lemma_1} {\bf(1)} the first term in \eqref{two_terms} can be written as
\begin{equation}\label{A} 
\!\!\! \mathsf{g}( \varphi (y)) \left(   f(y) \nabla _{ T_y \varphi ( u)} n + T _y\varphi (\nabla ^ \Sigma _u X) , \nabla _{ T _y\varphi (v)} n \right)=f(y) K^2(u,v) + K( \nabla _u^ \Sigma X, v).
\end{equation}
The equality above is obtained from $ \mathsf{g}( \varphi (y)) \left(   \nabla _{ T_y \varphi ( \partial _a)} n  , \nabla _{ T _y\varphi ( \partial _b)} n \right)=K_{ad}h^{df}K_{fb}$, which follows from the fact that the vector $\nabla _{ T _y\varphi ( \partial _b)} n$ is parallel to $ \Sigma $ and can be written as $\nabla _{ T _y\varphi ( \partial _a)} n ^ \alpha = \varphi ^ \alpha _{,b} h^{cb}K_{ca}$ by definition of $K$.

\noindent (1.B) To treat the second term in  \eqref{two_terms}, we use
\[
\frac{D}{D\varepsilon} \frac{D}{Dt}  \alpha ( \varepsilon , t) - \frac{D}{Dt}  \frac{D}{D \varepsilon }\alpha ( \varepsilon , t) = R \Big( \frac{d}{d \varepsilon  }  \gamma  ( \varepsilon , t),  \frac{d}{d t} \gamma  ( \varepsilon , t)\Big)  \alpha  (  \varepsilon  , t) 
\]
for $ \alpha  (\varepsilon , t) \in T _{ \gamma ( \varepsilon , t)}\mathcal{M} $ and we compute
\begin{align*} 
&\left. \frac{D}{D\varepsilon}\right|_{\varepsilon=0} \left.\frac{D}{Dt} \right|_{t=0}  n_ \varepsilon ( \varphi _ \varepsilon (y+ t v)) \\
&= \left.\frac{D}{Dt} \right|_{t=0}  \left. \frac{D}{D\varepsilon}\right|_{\varepsilon=0}  n_ \varepsilon( \varphi _ \varepsilon (y+ t v)) + R \Big( \left. \frac{d}{d\varepsilon}\right|_{\varepsilon=0} \varphi _ \varepsilon (y),  \left. \frac{d}{d t}\right|_{t=0}  \varphi  (y+ t v)\Big) n( \varphi (y))\\
&= \left.\frac{D}{Dt} \right|_{t=0}  T_y \varphi \Big( \big( - \epsilon \operatorname{grad}^ \Sigma f + (\mathbf{i} _XK)^{\sharp_h} \big)  (y+ t v) \Big) \\
& \qquad \qquad \qquad + R \big( f(y)n( \varphi (y))+ T \varphi (X(y)), T_y \varphi  (v)\big) n( \varphi (y)).
\end{align*}
From this equality, we can rewrite the second term in \eqref{two_terms} as  
\begin{align*} 
& \mathsf{g}( \varphi (y)) \Big( T_y \varphi  ( u),T_y \varphi \Big(  \left.\frac{D^ \Sigma }{Dt} \right|_{t=0}  \left( - \epsilon \operatorname{grad}^ \Sigma f + (\mathbf{i} _XK)^{\sharp_h} \right)  (y+ t v) \Big) \Big) \\
& \qquad + R( \varphi (y)) \big( T_y \varphi  ( u),n( \varphi (y)),f(y)n( \varphi (y))+ T \varphi (X(y)), T_y \varphi  ( v) \big) \\
&=h(y) \Big( u,  \nabla ^ \Sigma _ v \big( - \epsilon \operatorname{grad}^ \Sigma f + (\mathbf{i} _XK)^{\sharp_h} \big)  (y) \Big) \\
& \qquad + f(y)R( \varphi (y)) \big( T_y \varphi  ( u),n( \varphi (y)),n( \varphi (y)), T_y \varphi  (v)\big)\\
& \qquad  + R( \varphi (y)) \big( T_y \varphi  ( u),n( \varphi (y)), T \varphi (X(y)), T_y \varphi  (v) \big),
\end{align*} 
where $R(u,v,w,z)= \mathsf{g}(u,R(w,z)v)$. We note that the first term can be written as $- \epsilon \left\langle \nabla _v^ \Sigma \mathbf{d} f, u \right\rangle + \nabla _v^ \Sigma  K(X, u) + K( \nabla _v^ \Sigma X, u) $, while from the Gauss-Codazzi equation, the last term can be written as $\nabla^ \Sigma  _X K( u,v) - \nabla^ \Sigma  _v K( u, X)$. From this, the second term of \eqref{two_terms} can be written as
\begin{equation}\label{B}
\begin{aligned} 
&- \epsilon \left\langle \nabla _v^ \Sigma \mathbf{d} f, u \right\rangle - f(y)R( \varphi (y)) \big( n( \varphi (y)), T_y \varphi  ( u), n( \varphi (y)), T_y \varphi  (v)\big)  \\
& \qquad \qquad \qquad   + K( \nabla _v^ \Sigma X,u) + \nabla^ \Sigma  _X K(u, v) .
\end{aligned}
\end{equation}  
By summing the results obtained in \eqref{A} and \eqref{B}, we get the statement claimed in {\bf(1)}.

{\bf(2)} We have $\delta k = \delta ( h^{ab}K_{ab}) = \delta h^{ab} K_{ab}+ h^{ab} \delta K_{ab}$. By Lemma \ref{var_h_mu} {\bf (1)} the first term is
\begin{align*} 
\delta h^{ab} K_{ab}&=- h^{ad} \delta h_{dc} h^{cb}K_{ab} =  - h^{ad} 2 f K_{dc} h^{cb} K_{ab}  -h^{ad}  (\pounds _X h_{dc}) h^{cb} K_{ab}\\
&=-2 f \operatorname{Tr}_h(K ^2 ) + \pounds _Xh ^{-1}\!:\!K,
\end{align*} 
where we used $ \pounds _X h ^{-1} \!: \!h + h ^{-1} \!:\!\pounds _Xh=0$. By Lemma \ref{delta_Kk} {\bf (1)} the second term is
\[
h^{ab}\delta K_{ab} = f(y) \operatorname{Tr}_h(K ^2 )- f(y) Ric \big(n( \varphi (y)), n( \varphi (y))\big) - \epsilon \Delta ^ \Sigma f + h ^{-1} : \pounds _XK,
\]
where we used
\[
h^{ab}R( \varphi (y)) \Big( n( \varphi (y)), T_y \varphi  ( \partial _a), n( \varphi (y)), T_y \varphi  (\partial _b)\Big)= Ric \big(n( \varphi (y)), n( \varphi (y))\big).
\]
The final result follows from $h ^{-1} \!: \!\pounds _XK + \pounds _Xh ^{-1}\! :\! \pounds _XK=\pounds _X( h ^{-1}\! :\!K)= \mathbf{d} k \cdot  X$.
\end{proof}

\noindent \textbf{Proof of Theorem \ref{var_GHY}.} From Lemma \ref{delta_Kk} {\bf (2)} and Lemma \ref{var_h_mu} {\bf (2)}, we have
\begin{align*} 
\delta ( k \mu (h))&= \delta k \mu (h) + k \delta ( \mu (h))\\
&= -f \operatorname{Tr}_h(K ^2 )\mu (h) - f Ric\big(n( \varphi (y), n( \varphi (y))\big) \mu (h) \\
& \quad - \epsilon \Delta ^ \Sigma f\mu (h) + \mathbf{d} k \cdot X \mu (h) + k (f k + \operatorname{div}^ \Sigma X) \mu (h)\\
&= f\Big( k ^2 - \operatorname{Tr}_h(K ^2 ) -  Ric\big(n( \varphi (y), n( \varphi (y))\Big) \mu (h) \\
& \quad +\operatorname{div}^ \Sigma ( - \epsilon \operatorname{grad}^ \Sigma f + k X) \mu (h).
\end{align*} 
The result then follows from the divergence theorem
\[
\int_ \Sigma \operatorname{div}^ \Sigma X \mu (h) =  \int_{ \partial \Sigma } \!\sigma \,  h(X, \nu ) \mu ( \textcolor{black}{\vartheta} ),
\]
for $ X \in \mathfrak{X} ( \Sigma )$, where we used that $ \partial \Sigma $ consists of nondegenerate pieces with outward pointing unit vector field $ \nu $, and where $ \sigma =h( \nu , \nu ) \in \{\pm 1\}$ on $ \partial \Sigma$. $\qquad\blacksquare$

\subsection{Variational formulation for relativistic continuum coupled to gravity}

From the developments of the previous section, we can now state and prove our main result concerning the variational formulation for relativistic continuum coupled to gravitation, based on the action functional given in \eqref{total_action}.
We shall consider $ \mathcal{D} = [a,b] \times \mathcal{B} $ where $ \mathcal{B} $ has a smooth boundary and we assume that $ \Phi (a, \mathcal{B} )$ and $ \Phi (b, \mathcal{B} )$ are spacelike hypersurfaces, while $ \Phi ( [a,b] , \partial \mathcal{B} )$ is timelike since $ \Phi $ is a world-tube. In particular, the boundary does not contain null hypersurfaces and no contribution arises from the codimension-two surfaces at which the boundary pieces are joined together in our setting, see Remark \ref{non_smooth_and_null} for these delicate issues.

As explained earlier, we assume that the Lorentizian metrics $ \mathsf{g}^-$ and $ \mathsf{g}^+$ on which the action functional is evaluated satisfy the preliminary junction condition \eqref{junction_first} and we derive the Lanczos-Israel junction condition as a critical point condition. It is also possible to obtain \eqref{junction_first} from the variational formulation without assuming it a priori, by including the corresponding Lagrange multiplier term.

\begin{theorem}\label{main} Let $ \mathscr{L} : J^1( \mathcal{D}  \times \mathcal{M} ) \times T \mathcal{D}  \times T^p_q\mathcal{D}  \times S^2_L \mathcal{M}   \rightarrow \wedge ^{n+1} \mathcal{D} $ be the Lagrangian density of the continuum assumed to be material and spacetime covariant, and consider the associated spacetime Lagrangian density $\ell: T \mathcal{M}  \times  T ^p _q  \mathcal{M}  \times  S^2_L \mathcal{M} \rightarrow \wedge  ^{n+1} \mathcal{M} $.

Fix the reference tensor fields $W \in \mathfrak{X} ( \mathcal{D} )$ and $K \in \mathcal{T} _q^p( \mathcal{D} )$.
For each world-tube $ \Phi : \mathcal{D} \rightarrow \mathcal{M} $, define $\mathcal{N}^-= \Phi ( \mathcal{D})$, $ \mathcal{N} ^+= \mathcal{M} - \operatorname{int}( \mathcal{N} ^-)$, $w= \Phi _* W$, and $ \kappa = \Phi _ * K$. Consider smooth Lorentzian metrics $ \mathsf{g}^\pm \in \mathcal{S} ^2_L( \mathcal{N} ^\pm)$ such that $i_{\partial  \mathcal{N} ^-}^* \mathsf{g}^-= i_{ \partial \mathcal{N} ^+}^* \mathsf{g}^+$.

Then the following statements are equivalent:
\begin{itemize}
\item[\bf (i)]  $w \in \mathfrak{X} (\mathcal{N}^- )$, $ \kappa \in \mathcal{T} ^p_q( \mathcal{N}^- )$, $\textcolor{black}{\mathsf{g}^\pm }\in \mathcal{S} ^2_L( \mathcal{N} ^\pm)$ are critical points of the {\bfi Eulerian variational principle}
\begin{equation}\label{total_action_variation}
\begin{aligned} 
&\hspace{-0.7cm}\left. \frac{d}{d\varepsilon}\right|_{\varepsilon=0} \left[\int_{ \mathcal{N}^- _\varepsilon} \ell( w_\varepsilon , \kappa _\varepsilon ,  \mathsf{g}^- _\varepsilon) + \frac{1}{2 \chi }\int_{ \mathcal{N}^-_\varepsilon}  R( \mathsf{g}^-_\varepsilon) \mu ( \mathsf{g}^-_\varepsilon)+ \frac{1}{2 \chi }\int_{\mathcal{N} ^+_\varepsilon}  R( \mathsf{g}^+_\varepsilon) \mu ( \mathsf{g}^+_\varepsilon)\right. \\
&\qquad \qquad  \qquad   \left. + \frac{1}{ \chi }\int_{ \partial \mathcal{N}^-_\varepsilon} \!\epsilon \,k( \mathsf{g}^-_\varepsilon)   \mu ^-(h_\varepsilon)+ \frac{1}{ \chi }\int_{ \partial  \mathcal{N} ^+_\varepsilon} \!\epsilon\, k( \mathsf{g}^+_\varepsilon)   \mu ^+(h_\varepsilon) \right] =0,
\end{aligned}
\end{equation}
\textcolor{black}{for variations $ \delta \mathcal{N} = \zeta |_{ \partial \mathcal{N} } \big/T \partial \mathcal{N}$, $ \delta w = -\pounds _ \zeta w$, $ \delta \kappa = - \pounds _ \zeta \kappa $, where $ \zeta $ is an arbitrary vector field on $ \mathcal{N} $ such that $ \zeta |_{ \Phi (a, \mathcal{B} )}= \zeta |_{ \Phi (b, \mathcal{B} )}=0$, and for variations $ \delta \mathsf{g}^\pm$ connected to $ \zeta $ on $ \partial \mathcal{N} $ via the infinitesimal version of $i_{\partial  \mathcal{N} ^-_ \varepsilon } ^*\mathsf{g}^-_ \varepsilon = i_{ \partial \mathcal{N} ^+_ \varepsilon } ^*\mathsf{g}^+_ \varepsilon $, see the proof for the concrete formula, with $ \delta  \mathsf{g}|_{ \Phi (a, \mathcal{B} )}=\delta  \mathsf{g}|_{ \Phi (b, \mathcal{B} )}=0$.}
\item[\bf (ii)] $w \in \mathfrak{X} (\mathcal{N}^- )$, $ \kappa \in \mathcal{T} ^p_q( \mathcal{N}^- )$, $\mathsf{g}^\pm \in \mathcal{S} ^2_L( \mathcal{N} ^\pm)$ satisfy
\begin{equation}
\left\{
\begin{array}{l}
\displaystyle\operatorname{div}^ \nabla \! \left( \frac{\partial \ell}{\partial \mathsf{g}^-} \right) =0, \qquad \pounds _w \kappa =0\quad  \text{on} \quad    \mathcal{N}^-\\
\displaystyle\vspace{0.2cm} \textcolor{black}{Ein} (\mathsf{g}^-) \mu (\mathsf{g}^-) = 2 \chi  \frac{\partial \ell}{\partial \mathsf{g}^-}\quad  \text{on} \quad     \mathcal{N}^-\\
\displaystyle\vspace{0.2cm} \textcolor{black}{Ein} (\mathsf{g}^+) =0 \quad  \text{on} \quad     \mathcal{N}^+\\
\displaystyle [K]=0 \quad  \text{on} \quad  \textcolor{black}{\partial_{\rm cont} \mathcal{N}},
\end{array}
\right.
\end{equation}
where $2 \mathsf{g}^- \!\cdot\! \frac{\partial \ell}{\partial \mathsf{g}^-}=\ell \delta  + w \otimes \frac{\partial \ell}{\partial w} -   \frac{\partial \ell}{\partial \kappa }: \widehat{ \kappa  }$, \textcolor{black}{$ \nabla $ is the Levi-Civita covariant derivative associated to $ \mathsf{g}^-$}, and $[K]$ denotes the jump of the extrinsic curvature along $ \partial \mathcal{N} $.
\end{itemize}
\end{theorem}
\begin{proof}  \textcolor{black}{We consider arbitrary variations $\Phi _ \epsilon $ and $ \mathsf{g}_ \varepsilon ^\pm$ with fixed endpoint at $ \lambda  = a,b$, and which satisfy $i_{\partial  \mathcal{N} ^-_ \varepsilon } ^*\mathsf{g}^-_ \varepsilon = i_{ \partial \mathcal{N} ^+_ \varepsilon } ^*\mathsf{g}^+_ \varepsilon $, where $ \mathcal{N} _ \varepsilon ^-= \Phi _ \varepsilon  ( \mathcal{D} )$. As earlier, the constrained variations in {\bf (i)} are induced by the variations $\Phi _ \varepsilon  $ of the world-tube by using the relations $ \mathcal{N} _ \varepsilon = \Phi _ \varepsilon  ( \mathcal{D} )$, $w_ \varepsilon = (\Phi_ \varepsilon  )_*W$, $ \kappa _ \varepsilon = (\Phi_ \varepsilon  )_*K$, and defining $ \zeta= \delta \Phi \circ \Phi ^{-1}$}.\\
(1) \textit{Metric variation}: We first fix the world-tube and take the variation with respect to the Lorentzian metrics $ \mathsf{g}^\pm$. The variation of the Einstein-Hilbert terms is well-known and given by
\begin{align*} 
&\left. \frac{d}{d\varepsilon}\right|_{\varepsilon=0}  \left[ \frac{1}{2 \chi }\int_{ \mathcal{N}^-}  R( \mathsf{g}^-_\varepsilon) \mu ( \mathsf{g}^-_\varepsilon)+ \frac{1}{2 \chi }\int_{\mathcal{N} ^+}  R( \mathsf{g}^+_\varepsilon) \mu ( \mathsf{g}^+_\varepsilon) \right] \\
& = - \frac{1}{ 2\chi}\int_{ \mathcal{N} ^-} \textcolor{black}{Ein} ( \mathsf{g}^-): \delta  \mathsf{g}^-\, \mu ( \mathsf{g}^-)  - \frac{1}{ 2\chi}\int_{ \mathcal{N} ^+} \textcolor{black}{Ein} ( \mathsf{g}^+): \delta  \mathsf{g}^+ \,\mu ( \mathsf{g}^+)\\
&\quad \quad + \frac{1}{2\chi} \int_{ \partial \mathcal{N} ^-} \!\epsilon \, \mathsf{g}^-( \deltabar V^-, n^-) \mu ^-(h)+ \frac{1}{2\chi} \int_{ \partial \mathcal{N} ^+} \!\epsilon\,  \mathsf{g}^+( \deltabar V^+, n^+) \mu ^+(h).
\end{align*} 
On the other hand, denoting $ \partial _p \mathcal{N} ^\pm$, $p=1,2,...$, the smooth pieces of the boundary $\partial \mathcal{N} ^\pm$ and using Lemma \ref{variation_GBY_g} with $ \Sigma = \partial _p \mathcal{N} ^\pm$, the variation of the GHY terms is
\begin{align*} 
&\left. \frac{d}{d\varepsilon}\right|_{\varepsilon=0}  \left[ \frac{1}{ \chi }\int_{ \partial \mathcal{N}^-}  k( \mathsf{g}^-_ \varepsilon )   \mu ^-(h_ \varepsilon )+ \frac{1}{ \chi }\int_{ \partial  \mathcal{N} ^+} k( \mathsf{g}^+_ \varepsilon )   \mu ^+(h_ \varepsilon ) \right] \\
&  =\frac{1}{2\chi}\int_{ \partial \mathcal{N} ^-} \left( (K( \mathsf{g}^-)- k( \mathsf{g}^-)h): \delta h ^{-1} -  \mathsf{g}^-(\deltabar V^-, n^-) \right) \mu ^-(h)\\
&   \quad \quad +\frac{1}{2\chi}\int_{ \partial \mathcal{N} ^+} \left( (K( \mathsf{g}^+)- k( \mathsf{g}^+)h): \delta h ^{-1} -  \mathsf{g}^+(\deltabar V^+, n^+) \right) \mu ^+(h)\\
&   \quad \quad  -\frac{1}{2\chi}\sum_p \int_{ \partial ( \partial _p\mathcal{N} ^-)} \!\sigma _p \,h(\deltabar W^-, \nu _p) \mu _p^-( \textcolor{black}{\vartheta} ) - \frac{1}{2\chi}\sum_p \int_{ \partial ( \partial _p\mathcal{N} ^+)} \!\sigma _p\, h(\deltabar W^+, \nu _p) \mu _p^+( \textcolor{black}{\vartheta} ),
\end{align*}
where $ \nu _p$ is the outward pointing unit normal vector field to $ \partial _p \mathcal{N} ^\pm$ and $ \sigma _p= h( \nu _p, \nu _p)\in \{\pm\}$ on $ \partial (\partial _p \mathcal{N} ^\pm) $.

Note that the terms involving $ \deltabar V$ cancel and hence, the critical conditions of \eqref{total_action_variation} with respect to $ \delta  \mathsf{g}^-$, $ \delta  \mathsf{g}^+$, $ \delta h$ are, respectively:
\begin{align} 
\frac{\partial \ell}{\partial  \mathsf{g}^-}&= \frac{1}{2\chi}  \textcolor{black}{Ein} ( \mathsf{g}^-) \mu ( \mathsf{g}^-) \quad \text{on $ \mathcal{N} ^-$}\label{cond_1}\\
0&= \frac{1}{2\chi}  \textcolor{black}{Ein} ( \mathsf{g}^+) \mu ( \mathsf{g}^+) \quad \text{on $ \mathcal{N} ^+$}\label{cond_2}\\
0&= K( \mathsf{g}^-)- k( \mathsf{g}^-)h+ K( \mathsf{g}^+)- k( \mathsf{g}^+)h  \quad \text{on $\textcolor{black}{\partial _{\rm cont}\mathcal{N}}$}\label{cond_3}. \textcolor{white}{\frac{1}{2} }
\end{align}
The last condition is obtained by noting that $\mu ^-(h)=- \mu ^+( h)$ and $ \partial \mathcal{N} ^-$ and $ \partial \mathcal{N} ^+$ have opposite orientation.
Since $K(\mathsf{g}^+)$ and $k(\mathsf{g}^+)$ are computed with respect to the orientation of $ \partial \mathcal{N} ^+$, similarly for $K( \mathsf{g}^-)$, $k( \mathsf{g}^-)$ and $ \partial \mathcal{N} ^-$, choosing a common orientation on $ \partial \mathcal{N} $, this condition reads $[K-kh]=0$ or, equivalently, $[K]=0$.
\textcolor{black}{This condition only appears on $ \partial _{\rm cont} \mathcal{N} $ because $ \delta  \mathsf{g}|_{ \Phi (a, \mathcal{B} )}=\delta  \mathsf{g}|_{ \Phi (b, \mathcal{B} )}=0$. For the same reason, the integrals over the codimension-two surfaces $ \partial ( \partial _p \mathcal{N} ^\pm)$ above vanish.}

\noindent (2) \textit{World-tube variation}: The world-tube cannot be a priori varied independently of the metric. Let us define $ \zeta = \delta \Phi \circ \Phi ^{-1} $ and write the orthogonal decomposition $ \zeta = \zeta ^\perp_\pm n^\pm + \zeta ^{\|}$ along $\textcolor{black}{\partial \mathcal{N} ^\pm} $ with $ \zeta ^\perp_\pm = \epsilon  \mathsf{g}( \zeta , n^\pm)$. 
Taking the variation of the condition $i_{\partial  \mathcal{N} ^-}^* \mathsf{g}^-= i_{ \partial \mathcal{N} ^+}^* \mathsf{g}^+$ we get $ 2  \zeta ^\perp_- K( \mathsf{g}^-) + \pounds _ { \zeta ^{\|}} i_{ \partial \mathcal{N} ^-} ^*  \mathsf{g}^-+ i_{\partial  \mathcal{N} ^-}^*\delta  \mathsf{g}^-=  2 \zeta ^\perp_+ K( \mathsf{g}^+) + \pounds _ { \zeta ^{\|}} i_{ \partial \mathcal{N} ^+} ^*  \mathsf{g}^++ i_{ \partial \mathcal{N} ^+}^*\delta \mathsf{g}^+$ from Lemma \ref{var_h_mu}. From the junction condition already derived and from $i_{\partial  \mathcal{N} ^-}^* \mathsf{g}^-= i_{ \partial \mathcal{N} ^+}^* \mathsf{g}^+$, we obtain that at the critical condition already derived, $ \delta \Phi $ can be varied independently of $ \mathsf{g}^\pm$.

We now take the variation of \eqref{total_action_variation} with respect to the world-tube. As seen in the proof of Theorem \ref{spacetime_reduced_EL}, the variation of the term associated to the continuum is
\begin{align*} 
\left. \frac{d}{d\varepsilon}\right|_{\varepsilon=0} \int_{ \mathcal{N}^- _\varepsilon} \ell( w_\varepsilon , \kappa _\varepsilon ,  \mathsf{g}^-)&= - \int_ {\mathcal{N}  ^-}\Big(    \operatorname{div}^ \nabla \!\Big( \ell \delta  + w \otimes \frac{\partial \ell}{\partial w} -  \textcolor{black}{  \frac{\partial \ell}{\partial \kappa }\!\therefore\!  \widehat{ \kappa  }}\Big) \cdot \zeta - \frac{\partial ^ \nabla \!\ell}{\partial x} \cdot \zeta  \Big)\\
& \qquad \qquad \qquad + \int_{ \partial \mathcal{N} ^-} \operatorname{tr}  \Big( \ell \delta + w \otimes \frac{\partial \ell}{\partial w}- \textcolor{black}{ \frac{\partial \ell}{\partial \kappa }\!\therefore\!   \widehat{ \kappa  }} \Big) \cdot \zeta,
\end{align*}
with $ \nabla $ the Levi-Civita covariant derivative associated to $ \mathsf{g}^-$ and $ \zeta = \delta \Phi \circ \Phi ^{-1} $.
 
The variation of the Einstein-Hilbert terms with respect to the world-tube is
\begin{align*} 
&\left. \frac{d}{d\varepsilon}\right|_{\varepsilon=0} \left[ \frac{1}{2 \chi }\int_{ \mathcal{N}^-_\varepsilon }  R( \mathsf{g}^-) \mu ( \mathsf{g}^-)+ \frac{1}{2 \chi }\int_{\mathcal{N} ^+_ \varepsilon }  R( \mathsf{g}^+) \mu ( \mathsf{g}^+) \right] \\
& = \frac{1}{ 2\chi}\int_{ \partial  \mathcal{N} ^-} R(  \mathsf{g}^-) \zeta ^\perp_- \mu ^-(h) + \frac{1}{ 2\chi}\int_{ \partial  \mathcal{N} ^+} R(  \mathsf{g}^+) \zeta ^\perp_+ \mu ^+(h) .
\end{align*} 
Finally, using Theorem \ref{var_GHY} with $ \Sigma = \partial _p \mathcal{N} ^\pm$, the variation of the GHY terms yields
\begin{align*} 
&\left. \frac{d}{d\varepsilon}\right|_{\varepsilon=0}  \left[ \frac{1}{ \chi }\int_{ \partial \mathcal{N}^-_\varepsilon}\! \epsilon\, k( \mathsf{g}^-)   \mu ^-(h)+ \frac{1}{ \chi }\int_{ \partial  \mathcal{N} ^+_\varepsilon}\! \epsilon\, k( \mathsf{g}^+)   \mu ^+(h)\right] \\
& =\frac{1}{ \chi }\int_ { \partial \mathcal{N} ^-} \epsilon \,\zeta  ^\perp_-\big[ k( \mathsf{g}^-) ^2 - \operatorname{Tr}(K( \mathsf{g}^-) ^2 ) -  Ric^-\big(n^-,n^-)\big] \mu ^-(h) \\
& \quad \quad  + \frac{1}{ \chi }\sum_p\int_{ \partial ( \partial _p\mathcal{N} ^-)} \epsilon \,\sigma_p h\big(  k( \mathsf{g}^-) \zeta ^{\|} - \epsilon \operatorname{grad}^ { \partial \mathcal{N} }  (\zeta  ^\perp_-), \nu_p \big) \mu _p^-( \textcolor{black}{\vartheta} )\\
& \quad \quad+\frac{1}{ \chi }\int_ { \partial \mathcal{N} ^+} \epsilon \,\zeta  ^\perp_+\big[ k( \mathsf{g}^+) ^2 - \operatorname{Tr}(K( \mathsf{g}^+) ^2 ) -  Ric^+\big(n^+,n^+)\big] \mu ^+(h) \\
& \quad \quad+ \frac{1}{ \chi }\sum_p\int_{ \partial ( \partial _p\mathcal{N} ^+)} \epsilon \,\sigma_p h\big(  k( \mathsf{g}^+) \zeta ^{\|} - \epsilon \operatorname{grad}^ { \partial \mathcal{N} }  (\zeta  ^\perp_+), \nu_p \big) \mu _p^+( \textcolor{black}{\vartheta} ),
\end{align*}
where we used the same notations $\partial _p \mathcal{N} ^\pm$, $ \nu _p$, $ \sigma _p$, as earlier. \textcolor{black}{The integrals over the codimension-two surfaces vanish since $ \zeta |_{ \Phi (a, \mathcal{B} )}= \zeta |_{ \Phi (b, \mathcal{B} )}=0$.}

By grouping these results, we get the following conditions
\[
\operatorname{div}^ \nabla \!\Big( \ell \delta  + w \otimes \frac{\partial \ell}{\partial w} -   \textcolor{black}{ \frac{\partial \ell}{\partial \kappa }\!\therefore\!  \widehat{ \kappa  }}\Big)  = \frac{\partial ^ \nabla \!\ell}{\partial x} \quad \text{on} \quad \mathcal{N} ^-
\]
and, using Lemma \ref{Lemma_boundary} \textbf{(1)}, 
\begin{equation}\label{condition_boundary}
\!\!\begin{aligned}
&\epsilon \,\mathbf{i} _{\textcolor{black}{ n^-}} \Big(\ell \delta  + w \otimes \frac{\partial \ell}{\partial w} -  \textcolor{black}{  \frac{\partial \ell}{\partial \kappa }\!\therefore\!  \widehat{ \kappa  }}\Big)( \zeta ^\flat , \textcolor{black}{ n^-}) \\
&  + \frac{1}{ 2\chi} R(  \mathsf{g}^-) \zeta ^\perp_- \mu ^-(h) - \frac{1}{ 2\chi} R(  \mathsf{g}^+) \zeta ^\perp_+ \mu ^+(h)\\
&  + \frac{1}{ \chi} \epsilon \,\zeta  ^\perp_-\big[ k( \mathsf{g}^-) ^2 - \operatorname{Tr}(K( \mathsf{g}^-) ^2 ) -  Ric^-\big(n^-,n^-)\big] \mu ^-(h)\\
&  - \frac{1}{ \chi} \epsilon \,\zeta  ^\perp_+\big[ k( \mathsf{g}^+) ^2 - \operatorname{Tr}(K( \mathsf{g}^+) ^2 ) -  Ric^+\big(n^+,n^+)\big] \mu ^+(h)=0,\;\forall\; \zeta \;\; \text{on}\;\; \textcolor{black}{\partial_{\rm cont} \mathcal{N}},
\end{aligned}
\end{equation}
\textcolor{black}{where we used $ \zeta |_{ \Phi (a, \mathcal{B} )}= \zeta |_{ \Phi (b, \mathcal{B} )}=0$}.
By using $ \zeta ^\perp_+=- \zeta ^\perp_-$, $ \mu ^+(h)= -\mu ^-(h)$, and the junction condition $[K]=0$ already obtained above, see \eqref{cond_3}, condition \eqref{condition_boundary} is equivalently written as 
\begin{equation}\label{intermediate_step1}
\begin{aligned}
&\epsilon \,\mathbf{i} _{n^-} \Big(\ell \delta  + w \otimes \frac{\partial \ell}{\partial w} -  \textcolor{black}{\frac{\partial \ell}{\partial \kappa }\!\therefore\!  \widehat{ \kappa }}\Big)( \zeta ^\flat , n^-) \\
& \quad  + \frac{1}{ \chi } \left(\textcolor{black}{Ein} ( \mathsf{g}^+)(n^+,n^+) -  \textcolor{black}{Ein} ( \mathsf{g}^-)(n^-,n^-) \right) \epsilon \,\zeta _-^\perp \mu  ^-(h)=0 ,\;\forall\; \zeta \;\; \text{on}\;\; \textcolor{black}{ \partial_{\rm cont}\mathcal{N}}.
\end{aligned}
\end{equation}
From the spacetime covariance of $\ell$ we have
\[
\ell \delta  + w \otimes \frac{\partial \ell}{\partial w} -   \textcolor{black}{\frac{\partial \ell}{\partial \kappa }\!\therefore\!  \widehat{ \kappa } }= 2  \mathsf{g}^- \!\cdot\! \frac{\partial \ell}{\partial  \mathsf{g}^-} \qquad\text{and}\qquad \frac{\partial ^\nabla\ell}{\partial x} =0,
\]
see Lemma \ref{spacetime_mat_ell} with $ \gamma = \mathsf{g}^-$. Hence, using \eqref{cond_1} and \eqref{cond_2}, the condition \eqref{intermediate_step1} reads 
\[
\textcolor{black}{Ein} ( \mathsf{g}^-)( \zeta  , n^-)  \mu ^-(h) -  \textcolor{black}{Ein} ( \mathsf{g}^-)(n^-,n^-) \zeta _-^\perp \mu  ^-(h)=0 ,\;\forall\; \zeta \;\; \text{on}\;\; \textcolor{black}{\partial _{\rm cont}\mathcal{N}}.
\]
This is equivalent to
\[
i_{ \partial \mathcal{N} } ^* \left( \mathbf{i} _{n^-}\textcolor{black}{Ein} ( \mathsf{g}^-) \right) =0
\]
which can also be written as $i_{ \partial \mathcal{N} } ^* \left( \mathbf{i} _{n^-}[\textcolor{black}{Ein} ] \right) =0$ since $\textcolor{black}{Ein} ( \mathsf{g}^+)=0$. It turns out that this condition is a consequence of the junction condition already derived, as it is easily seen by using the Gauss-Codazzi equation, see also Remark \ref{OBrien_Synge}.
\end{proof}

\begin{remark}[On the O'Brien-Synge junction conditions]\label{OBrien_Synge}\rm From the Gauss-Codazzi equations on a nondegenerate hypersurface $ \Sigma $, one obtains the equalities $\textcolor{black}{Ein} (n, n) = - \frac{1}{2} ( \epsilon R^ \Sigma + \operatorname{Tr}(K^2) - k^2)$ and $\textcolor{black}{Ein} ( \zeta , n) =\textcolor{black}{  \nabla^ \Sigma  _{ \partial _a}}K ( \zeta , \partial _b) h^{ab}- \textcolor{black}{  \nabla^ \Sigma  _ \zeta} k$ for all $ \zeta \in \mathfrak{X} ( \Sigma )$. Hence the Darmois-Israel junction conditions $[h]=[K]=0$ imply the conditions $[\textcolor{black}{Ein} (n,n)]=[\textcolor{black}{Ein} ( \zeta , n)]=0$ for all $ \zeta \in \mathfrak{X} ( \Sigma )$, which are a generalized form of the O'Brien-Synge conditions, see \cite{Is1966} and \cite{OBSy1952}.
\end{remark}

\begin{remark}[Null hypersurfaces and nonsmooth intersection]\label{non_smooth_and_null}\rm In our setting we have assumed that the boundaries include only nondegenerate hypersurfaces and that no contribution arises from the codimension-two surfaces at which the boundary pieces are joined together.
On null hypersurfaces the GHY term is ill-defined and an appropriate replacement for the boundary contribution has to be considered in the Einstein-Hilbert action, \cite{Ne2012}, \cite{PaChMaPa2016}. Also, when the intersection between two pieces of the boundary is not smooth, the extrinsic curvature is singular and the boundary action acquires additional contributions from the intersection, \cite{HaSo1981}, \cite{Ha1993}. We refer to \cite{LeMyPoSo2016} for a complete treatment of these situations.
\end{remark}

\color{black} 
\begin{remark}[$ \partial \mathcal{N} $ and $ \partial _{\rm cont} \mathcal{N}$]\rm It should be noted that while the boundary conditions only appear on the spacetime region $ \partial _{\rm cont} \mathcal{N} $ occupied by the boundary of the continuum, which follows from the endpoint condition $\zeta |_{ \Phi (a, \mathcal{B} )}= \zeta |_{ \Phi (b, \mathcal{B} )}=0$ and $ \delta \mathsf{g} |_{ \Phi (a, \mathcal{B} )}=  \delta \mathsf{g}  |_{ \Phi (b, \mathcal{B} )}=0$, the GHY boundary terms in the action functional are needed on the whole boundary $\partial \mathcal{N} ^\pm$ (as well as on $\partial \mathcal{M} $ if not empty).
\end{remark} 

\color{black} 

\section{Examples}

We show in this section how the reduced variational framework proposed above applies to relativistic fluids and elasticity. 
It should be noted that this setting allows to obtain the complete system of field equations by a variational principle directly deduced from the Hamilton principle, by using the covariance assumptions and without the need to include any constraints or unphysical variables.
In the case of relativistic elasticity, our setting also allows to clarify the relation between formulations based on the relativistic right Cauchy-Green tensor or on the relativistic Cauchy deformation tensor.
For the examples below, we take $ \mathcal{D} = [a,b] \times \mathcal{B} $ and $W= \partial _ \lambda $.


\subsection{General relativistic fluids}

\paragraph{Material tensor fields.}
For the description of a relativistic fluid, besides the vector field $\partial _ \lambda \in \mathfrak{X} ( \mathcal{D} )$, two other material tensor fields are needed, given by volume forms $R, S \in \Omega ^{n+1}( \mathcal{D} )$ with
\begin{equation}\label{condition_R_S} 
\pounds _{ \partial _ \lambda } R = \pounds _{\textcolor{black}{\partial_ \lambda }} S =0.
\end{equation}
These two forms correspond to the reference mass density and entropy density and are chosen as
\begin{equation}\label{RR_0SS_0} 
R= d \lambda  \wedge \pi _ \mathcal{B} ^* R_0 \quad\text{and}\quad S= d \lambda  \wedge \pi _ \mathcal{B} ^* S_0,
\end{equation}
where $R_0, S_0 \in \Omega ^n( \mathcal{B} )$ are volume forms on $ \mathcal{B} $, called the mass form and entropy form, \textcolor{black}{and $ \pi _ \mathcal{B} :[a,b] \times \mathcal{B} \rightarrow \mathcal{B} $ is the projection.}
As shown in Remark \ref{remark_generalization}, properties \eqref{condition_R_S} hold for $R,S$ given in \eqref{RR_0SS_0}.

The corresponding Eulerian quantities are the generalized velocity $w$, generalized mass density $ \varrho $, and generalized entropy density $ \varsigma $ given by:
\[
w= \Phi _* W, \qquad \varrho = \Phi _* R, \qquad \varsigma= \Phi_*S.
\]
From these quantities, the world-velocity, the proper mass and entropy densities, and the proper specific entropy are defined as
\begin{equation}\label{relation_fluid_1}
u = \frac{cw}{\sqrt{-
\mathsf{g}(w,w)}} , \quad \rho  = \frac{\sqrt{-\mathsf{g}(w,w)}\;\varrho}{c \;\mu (\mathsf{g})} , \quad s  = \frac{\sqrt{-\mathsf{g}(w,w)}\; \varsigma}{c \;\mu (\mathsf{g})} , \quad \eta = \frac{s}{\rho    }=  \frac{\varsigma }{\varrho }.
\end{equation}
The distinction between the generalized and proper densities plays a central role in our approach and it arises since we  are using world-tubes that do not necessarily satisfy the normalisation condition $\sqrt{- \mathsf{g}( \dot \Phi , \dot \Phi  )}=c$. As we already commented earlier, this allows to formulate the variational principle in the material picture as a standard Hamilton principle, without any constraints.

\paragraph{Lagrangian densities for relativistic fluids.}
In the setting described above, assuming spacetime and material covariance, the corresponding Lagrangian densities are functions of the form
\[
\mathscr{L}( j^1 \Phi , \partial _ \lambda , R,S, \mathsf{g} \circ \Phi ) , \qquad \mathcal{L} (\partial  _ \lambda , R,S, \Gamma ), \qquad \ell(w, \varrho , \varsigma , \mathsf{g}),
\]
\textcolor{black}{where $ \Gamma = \Phi ^* \mathsf{g}$ is the pull-back of the spacetime Lorentzian metric.}
For a fluid with specific internal energy given by a function $e( \rho  , \eta )$, the Lagrangian density is usually given in the Eulerian form, and is defined as (minus) the sum of the proper rest-mass energy density and proper internal energy density, namely
\begin{equation}\label{ell_fluid} 
\begin{aligned}
\ell(w, \varrho , \varsigma , \mathsf{g})&= - \rho  ( c^2 + e( \rho  , \eta )) \mu (\mathsf{g})\\
&= -  \frac{1}{c} \sqrt{- \mathsf{g}(w,w)} \Big( c ^2 + e\Big(  \frac{1}{c}  \sqrt{- \mathsf{g}(w,w)} \frac{\varrho}{ \mu (\mathsf{g})}  ,\frac{\varsigma}{\varrho}  \Big) \Big)\varrho,
\end{aligned}
\end{equation} 
where in the second equality we used the expression of $ \rho  $ and $ \eta $ in terms of $w$, $ \varrho $, $ \varsigma $ given in \eqref{relation_fluid_1}.
From this, the material Lagrangian density is found from the relation
\begin{equation}\label{fluid_L_ell}
\mathscr{L}(j^1 \Phi , \partial _ \lambda  , R, S, \mathsf{g} \circ \Phi )= \Phi ^* [\ell(w, \varrho , \varsigma , \mathsf{g})],
\end{equation}
with $w= \Phi _* \partial _ \lambda$, $\varrho = \Phi _*R$, $\varsigma = \Phi _* S$ which gives
\begin{equation}\label{Lagrangian_density_fluid}
\!\!\!\!\!\mathscr{L}(j^1 \Phi , \partial _ \lambda  , R, S, \mathsf{g} \circ \Phi )= -  \frac{1}{c} \sqrt{-  \mathsf{g}( \dot \Phi   ,  \dot \Phi  )} \bigg( c ^2  +  e  \bigg( \frac{1}{c} \sqrt{-  \mathsf{g}( \dot \Phi   ,  \dot \Phi  ) } \frac{R }{\Phi ^* [\mu (\mathsf{g})]} ,\frac{ S }{R} \bigg) \bigg) R.
\end{equation} 
This Lagrangian is automatically \textcolor{black}{materially covariant} with respect to the isotropy subgroup of $ \partial _ \lambda $, see \eqref{isom_sdp}. 
The extension of $ \mathscr{L}$ to arbitrary non vanishing vector fields $W$ with $ \Gamma (W,W)<0$ is found by writing $w= \Phi _*W$ in the relation \eqref{fluid_L_ell} above, from which we deduce that $\mathscr{L}(j^1 \Phi , W  , R, S, \mathsf{g} \circ \Phi )$ is given by the same expression \eqref{Lagrangian_density_fluid} with $\dot \Phi $ replaced by $ T \Phi \! \cdot \!W$.
Now, from its definition, this expression is automatically \textcolor{black}{materially covariant} under the whole group $ \operatorname{Diff}( \mathcal{D} )$. Spacetime covariance is easily checked and leads to the convective Lagrangian density
\[
\mathcal{L} (W, R,S, \Gamma )= -  \frac{1}{c}\sqrt{-   \Gamma ( W, W )}\bigg( c ^2  +  e  \bigg( \frac{1}{c} \sqrt{-   \Gamma (  W,W  )}\frac{R }{\mu ( \Gamma )} ,\frac{ S }{R} \bigg) \bigg) R.
\]

\begin{remark}[Material Lagrangian, relativistic dust, and the relativistic particle]\rm The material Lagrangian for relativistic fluid given in \eqref{Lagrangian_density_fluid} \textcolor{black}{has not} been apparently considered earlier. This expression elegantly relates the Lagrangian for fluids to that of the relativistic particle. For instance, for the relativistic dust ($e=0$), using \eqref{RR_0SS_0}, it takes the form
\[
\mathscr{L}(j^1 \Phi , \partial _ \lambda  , R, \mathsf{g} \circ \Phi )= -   c\sqrt{-  \mathsf{g}( \dot \Phi   ,  \dot \Phi  )} \; d \lambda \wedge R_0.
\]
This Lagrangian is nothing else than a continuum collection of the Lagrangian density for the relativistic particle
\[
\mathscr{L}( x, \dot x) = - mc\sqrt{-  \mathsf{g}( \dot x   ,  \dot x  )} \; d \lambda,
\]
to which it reduces when $ \mathcal{B} $ is just a point $ \mathcal{B} = \{\mathsf{X}\}$.
\end{remark}

\paragraph{Spacetime reduced Euler-Lagrange equations for fluids.} The results stated in Theorem \ref{spacetime_reduced_EL} and \ref{convective_reduced_EL} are directly applicable to the relativistic fluids. In particular, we get the following Eulerian variational principle and general forms of the spacetime reduced Euler-Lagrange equations by direct application of \eqref{Eulerian_VP} and \eqref{spacetime_EL_g}.

\begin{proposition}\label{prop_Eulerian_VP_fluid} The Eulerian variational formulation for relativistic fluid takes the form\color{black} 
\begin{equation}\label{Eulerian_VP_fluid}
\begin{aligned} 
&\!\!\left. \frac{d}{d\varepsilon}\right|_{\varepsilon=0}\int_{ \mathcal{N}_ \varepsilon } \ell\big( w_ \varepsilon , \varrho  _ \varepsilon , \varsigma _ \varepsilon , \mathsf{g} \big)=0 \quad \text{for variations}\\
& \delta \mathcal{N} = \zeta |_{ \partial \mathcal{N} } \big/T \partial \mathcal{N} , \quad \delta w = -\pounds _ \zeta w, \qquad  \delta \varrho  = - \pounds _ \zeta \varrho , \qquad  \delta \varsigma   = - \pounds _ \zeta \varsigma,\phantom{\int_A^B}
\end{aligned} 
\end{equation}\color{black} 
\textcolor{black}{where $ \zeta$} is an arbitrary vector field on $ \mathcal{N} $ \textcolor{black}{such that $ \zeta |_{ \Phi (a, \mathcal{B} )}= \zeta |_{ \Phi (b, \mathcal{B} )}=0$}.

The critical conditions associated to \eqref{Eulerian_VP_fluid} are
\begin{equation}\label{spacetime_EL_fluid} 
\left\{
\begin{array}{l}
\displaystyle\vspace{0.2cm}\operatorname{div}^ \nabla \!\left( \left(  \ell - \varrho \frac{\partial \ell}{\partial \varrho }- \varsigma  \frac{\partial \ell}{\partial \varsigma  } \right)  \delta  + w \otimes \frac{\partial \ell}{\partial w}\right)=0\\
\displaystyle \left(  \left(  \ell - \varrho \frac{\partial \ell}{\partial \varrho }- \varsigma  \frac{\partial \ell}{\partial \varsigma  } \right)  \delta  + w \otimes \frac{\partial \ell}{\partial w}\right) ( \cdot ,n ^\flat )=0\;\;\text{on $\textcolor{black}{\partial_{\rm cont} \mathcal{N}}$}
\end{array}
\right.
\end{equation}
and the variables $w$, $ \varrho $, $ \varsigma $ satisfy
\begin{equation}\label{advections_fluid} 
\pounds _w \varrho =0 \quad\text{and}\quad \pounds _w \varsigma  =0.
\end{equation} 
\end{proposition}

\medskip 

We note that equations \eqref{advections_fluid} are equivalently written in terms of the world-velocity, the proper mass and entropy density as
\[
\pounds _u (\rho  \mu(g))  =0 \quad\text{and}\quad \pounds _u (s  \mu(g))   =0,
\]
which are obtained by using the equality $ \pounds _{\frac{1}{f}w}( \varrho f)= \pounds _ w \varrho $, \textcolor{black}{for $f$ a non-vanishing function}.

\paragraph{Relativistic fluid equations.} From the expression \eqref{ell_fluid} of the Lagrangian density for fluids, one directly computes the partial derivatives
\begin{align*}
\frac{\partial \ell}{\partial w }&=\frac{1}{c ^2 }   u ^\flat \Big( c^2 + e + \rho  \frac{\partial e}{\partial \rho  } \Big)  \varrho, \qquad \frac{\partial \ell}{\partial \varsigma  }= - \frac{1}{c}\sqrt{-\mathsf{g}(w,w)}  \frac{\partial e}{\partial \eta } \\
\frac{\partial \ell}{\partial \varrho }&=  - \frac{1}{c} \sqrt{-\mathsf{g}(w,w)}  \Big( c^2 + e + \frac{\partial e}{\partial \rho  } \rho  - \frac{\partial e}{\partial \eta }\eta  \Big)
\end{align*}  
which yields the \textcolor{black}{stress-energy-momentum tensor density} for fluids 
\begin{equation}\label{SEM_fluid} 
\mathfrak{T}_{\rm fluid}=\left( \ell- \frac{\partial \ell}{\partial \varrho } \varrho- \frac{\partial \ell}{\partial \varsigma  } \varsigma  \right) \delta  + w \otimes \frac{\partial \ell}{\partial w} = \Big( p \,\delta +  \frac{1}{c ^2 } u \otimes    u ^\flat ( \epsilon _{\rm tot}+ p )\Big) \mu (\mathsf{g}),
\end{equation} 
with $ \epsilon _{\rm tot}=  \rho  ( c ^2 + e( \rho, \eta   )) $ the total energy density and $p= \rho  ^2 \partial e/ \partial  \rho  $ the pressure.

The first equation in \eqref{spacetime_EL_fluid} yields the relativistic Euler equation and energy equation
\begin{equation}\label{relativistic_Euler} 
\frac{1}{c^2} (  \epsilon _{\rm tot} + p )  \nabla _ u u =- \mathsf{P} \nabla p\qquad\text{and}\qquad  \operatorname{div}(\epsilon _{\rm tot}u) + p \operatorname{div}u  =0 
\end{equation} 
while the second one gives the vacuum boundary conditions
\begin{equation}\label{vaccum_Euler}
p|_{\textcolor{black}{ \partial_{\rm cont} \mathcal{N} }}=0,
\end{equation}
which follows from $\mathfrak{T}_{\rm fluid}( \cdot , n^\flat)= p \,n^\flat \textcolor{black}{\mu ( \mathsf{g})}$ on $\textcolor{black}{ \partial_{\rm cont} \mathcal{N} }$.
Note that we also have $\mathsf{g}(u,n)=0$ on $\textcolor{black}{ \partial_{\rm cont} \mathcal{N} }$, by definition of $u$ in terms of the world-tube. In \eqref{relativistic_Euler} we introduced the orthogonal projector $\mathsf{P}$ onto $(\operatorname{span}u)^\perp$, see \eqref{expression_p} later. \textcolor{black}{We refer to \cite{MiThWh1973} for a derivation of \eqref{SEM_fluid} and \eqref{relativistic_Euler} via the standard methods used in relativity theory}.

\paragraph{Coupling with the Einstein equations and junction conditions.} The extension of the variational formulation in Proposition \ref{prop_Eulerian_VP_fluid} to the coupling with general relativity can be obtained by particularizing Theorem \ref{main} to the fluid Lagrangian density \eqref{ell_fluid}. The resulting variational principle then yields the Einstein equations on $ \mathcal{N} ^\pm$, the relativistic fluid equations \eqref{relativistic_Euler} on $ \mathcal{N} ^-$, as well as the Israel-Darmois junction conditions on $\textcolor{black}{ \partial_{\rm cont} \mathcal{N} }$. Recall from Remark \ref{OBrien_Synge} that the latter conditions imply the O'Brien-Synge boundary conditions. Using the Einstein equations on $ \mathcal{N} ^\pm$ one directly gets from the O'Brien-Synge conditions and from the expression \eqref{SEM_fluid}, the boundary conditions \eqref{vaccum_Euler}. To summarize, for the general relativistic fluid we have
\[
[h]=[K]=0 \;\Longrightarrow \; p|_{ \textcolor{black}{ \partial_{\rm cont} \mathcal{N} }}=0.
\]

\subsection{General relativistic elasticity}

As we have seen for the Newtonian case in \S\ref{Review_Newtonian}, the description of relativistic elasticity requires the introduction of an additional material tensor field, namely, a Riemannian metric $G_0$ on $ \mathcal{B} $. From it, we define the 2-covariant symmetric positive tensor $G= \pi _ \mathcal{B} ^* G_0$ on $ \mathcal{D} =[a,b] \times \mathcal{B} $. It satisfies
\[
\pounds _{ \partial _ \lambda } G= \mathbf{i} _{ \partial _ \lambda }G=0.
\]
\textcolor{black}{Note in particular that $G$ is not a metric on $ \mathcal{D} $.}

Given a Lorentzian metric $\mathsf{g}$ on $ \mathcal{M} $, a world-tube $ \Phi : \mathcal{D} \rightarrow \mathcal{M} $ and the associated world-velocity $u$, an important tensor on spacetime for the description of relativistic elasticity is the projection tensor $\mathsf{p}$ defined by
\[
\mathsf{p}(v_x,w_x)= \mathsf{g}(\mathsf{P}(v_x) , \mathsf{P}(w_x)), \quad \text{for all $v_x,w_x \in T_x \mathcal{M} $},
\]
where $v_x \mapsto \mathsf{P}(v_x)$ denotes the orthogonal projector onto $(\operatorname{span}u)^\perp$, see \cite{CaQu1972}. This tensor, \textcolor{black}{also called the Landau-Lifshitz radar metric, see \cite{Br2021}}, plays somehow the same role as the Riemannian metric $g$ on $ \mathcal{S} $ in the non relativistic case, see \S\ref{Review_Newtonian}, while being now an unknown in the relativistic case. This point will be made precise below. Explicitly, $\mathsf{p}$ and $\mathsf{P}$ are given by
\begin{equation}\label{expression_p} 
\begin{aligned} 
\mathsf{p} &=\textcolor{black}{ \mathsf{g}} + \frac{1}{c ^2 } u ^\flat \otimes u ^\flat =\textcolor{black}{ \mathsf{g}}- \frac{1}{g(w,w)}  w^\flat \otimes w ^\flat\\
\mathsf{P}&= \delta + \frac{1}{c ^2 } u  \otimes u ^\flat = \delta - \frac{1}{g(w,w) } w  \otimes w ^\flat .
\end{aligned}
\end{equation}

\paragraph{Relativistic deformation tensors.} Here we recall the definition of some relativistic deformation tensors, see, e.g., \cite{GrEr1966,ErMa1990,Ma1978}, then we clarify their intrinsic nature and their relation with the tensors $G$ and $\mathsf{p}$ in simple geometric terms. In particular, the approach presented here also allows to clarify the relation between formulations of relativistic elasticity based on the relativistic right Cauchy-Green tensor and the metric $G_0$ on one hand, or based on the relativistic Cauchy deformation tensor and the projection tensor on the other hand. A main step for this clarification is the result stated in Lemma \ref{tensors_elasticity}.

For $ \mathcal{D} = [ a,b] \times \mathcal{B} \ni X=( \lambda , \mathsf{X})$, we denote by $\mathsf{X}^A$, $A=1,...,n$ the coordinates on $ \mathcal{B} $ and by $X^a$, $a=0,...,n$ the coordinates on $ \mathcal{D} $ with $X^a= \lambda $ for $ a=0$ and $X^a=\mathsf{X}^a$, for $a=1,...,n$.
Given a world-tube $ \Phi $ and a Lorentzian metric $\mathsf{g}$, the relativistic deformation gradient $F^ \alpha _A$ and the relativistic right Cauchy-Green tensor $C_{AB}$ are defined as
\begin{equation}\label{def_F_C} 
F^ \mu  _A =\mathsf{P}^ \mu _ \lambda {\Phi _{,A} }^ \lambda     \quad\text{and}\quad C_{AB}= \mathsf{g}_{ \mu  \nu  } F^ \mu  _A F^ \nu  _B.
\end{equation} 
Note that, contrarily to what the notation may suggest, $C_{AB}$ is not a tensor field on $ \mathcal{B} $ since it depends on $\lambda $. The inverse relativistic deformation gradient $({}^{\text{-1}\!\!}F)^A_ \mu$ and relativistic Cauchy deformation tensor $\mathsf{c}_{ \mu \nu }$ are defined by
\begin{equation}\label{def_Fi_c} 
({}^{\text{-1}\!\!}F)^A_ \mu = (\Phi ^{-1}) ^A_{, \mu } \quad\text{and}\quad \mathsf{c}_{ \mu \nu }= G_{AB} ({}^{\text{-1}\!\!}F)^A_ \mu({}^{\text{-1}\!\!}F)^B_ \nu.
\end{equation}

\begin{lemma}\label{tensors_elasticity} The relativistic right Cauchy-Green tensor $C$ and the relativistic Cauchy deformation tensor $\mathsf{c}$ are related to the tensors $ \mathsf{p}$ and $G= \pi _ \mathcal{B}  ^* G_0$ via the world-tube $ \Phi $ as
\begin{equation}\label{relation_C_c_p_G} 
C= \Phi ^* \mathsf{p}  \qquad\text{and}\qquad \mathsf{c}= \Phi _* G.
\end{equation} 
In particular $C$ and $\mathsf{c}$ are tensor fields on $ \mathcal{D} $ and $ \mathcal{M} $, respectively. They satisfy
\[
\mathbf{i} _{ \partial _ \lambda }C=0 \qquad\text{and}\qquad \mathbf{i} _u \mathsf{c}=0.
\]
For all $( \lambda , \mathsf{X}) \in \mathcal{D} $ and all $ x \in \Phi ( \mathcal{D} )$,
\[
C( \lambda , \mathsf{X}):T_\mathsf{X} \mathcal{B} \times T_\mathsf{X} \mathcal{B} \rightarrow \mathbb{R} \quad\text{and}\quad 
\mathsf{c}(x): \operatorname{span}(u(x))^\perp \times  \operatorname{span}(u(x))^\perp \rightarrow \mathbb{R}
\]
are positive definite.
\end{lemma} 
\begin{proof} Consider the tensor field on $ \mathcal{D} $ defined by $C= \Phi ^* \mathsf{p}$. In coordinates, we have
\[
C_{ab}= \mathsf{p}_{ \mu \nu } \Phi ^ \mu _{,a} \Phi ^ \nu _{,b}=\mathsf{g}_{ \mu \nu } \mathsf{P}^ \mu _ \lambda \Phi  _{,a} ^ \lambda   \mathsf{P}^ \nu _ \kappa \Phi _{,b} ^ \kappa .
\]
On one hand we have $C_{0b}=0$ for all $b=0,...,n$ since $\mathsf{P}^ \mu _ \lambda \Phi ^ \lambda _ {,0}= \mathsf{P}^ \mu _ \lambda w^ \lambda  \circ \Phi =0$, on the other hand $C_{ab}$ coincides with the local expression given in \eqref{def_F_C} for $a,b=1,...,n$. Hence the tensor field $C$ defined by $C= \Phi ^* \mathsf{p}$ is canonically identified with the relativistic right Cauchy-Green tensor as defined in \eqref{def_F_C}, and satisfies $ \mathbf{i} _{ \partial _ \lambda }C=0$.

To prove the second equality we note that $ \mathsf{c}= \Phi _* G= \Phi _* \pi _ \mathcal{B} ^* G_0= ( \pi _ \mathcal{B} \circ \Phi ^{-1} ) ^* G_0$ and it is easily seen that the local expression of $( \pi _ \mathcal{B} \circ \Phi ^{-1} ) ^* G_0$ coincides with that given in \eqref{def_Fi_c}. We also have $ \mathbf{i} _ u \mathsf{c}=  \mathbf{i} _{ \Phi _* \partial _ \lambda } \Phi _* G= \Phi _*( \mathbf{i} _{ \partial _ \lambda } G)=\Phi _*( \mathbf{i} _{ \partial _ \lambda } \pi _ \mathcal{B} ^*  G_0)=0$.
\end{proof}

\paragraph{Lagrangian densities for relativistic elasticity.} In the setting described above, assuming spacetime covariance, the corresponding Lagrangian densities are functions of the form
\[
\mathscr{L}( j^1 \Phi , \partial _ \lambda , R,G, \mathsf{g} \circ \Phi ) , \qquad \mathcal{L} (\partial  _ \lambda , R,G, \Gamma ).
\]
Under material covariance we can further define the Lagrangian density
\[
\ell(w, \varrho ,  \mathsf{c} , \mathsf{g}).
\]
We consider a general expression $\mathcal{W}(G_0, C)$ for the specific stored energy function for (possibly anisotropic) elasticity, written in terms of the given Riemannian metric $G_0$ on $ \mathcal{B} $ and the relativistic right Cauchy-Green tensor $C$. As a map it reads $ \mathcal{W}: S^2_+ \mathcal{B} \times S^2_+ \mathcal{B} \rightarrow \mathbb{R} $, where $S^2_+ \mathcal{B} \rightarrow \mathcal{B} $ denotes the bundle of 2-covariant symmetric positive tensors, i.e., locally, we have $ \mathcal{W}\big(\mathsf{X}^A,G_0(\mathsf{X})_{AB}, C( \lambda ,\mathsf{X})_{AB}\big)$. Note that an explicit dependence of $ \mathcal{W}$ on $\mathsf{X} \in \mathcal{B} $ is possible.
Inspired by the expression \eqref{Lagrangian_density_fluid} of the material Lagrangian density for fluids, we construct the following Lagrangian density for elasticity with stored energy $ \mathcal{W}$:
\begin{equation}\label{Lagrangian_density_elasticity}
\mathscr{L}(j^1 \Phi , \partial _ \lambda  , R, G, \mathsf{g} \circ \Phi )= -  \frac{1}{c}\sqrt{-  \textcolor{black}{\mathsf{g}} ( \dot \Phi   ,  \dot \Phi  )}\big( c ^2  +  \mathcal{W}(G_0, C) \big) R.
\end{equation} 
It is important to note here that $C= \Phi ^* \mathsf{p}$ on the right hand side, is considered as a function of the material tensor $ \partial _ \lambda $, the \textcolor{black}{spacetime} tensor $\mathsf{g}$ and the first jet extension $j ^1 \Phi $, namely, we have
\[
C=C(j^1 \Phi , \partial _ \lambda , \mathsf{g} \circ \Phi )= \Phi ^* \mathsf{p} = \Phi ^* \left( \mathsf{g} - \frac{1}{\mathsf{g}(w,w)}  w^\flat \otimes w ^\flat \right) , \qquad w= \Phi _* \partial _ \lambda .
\]
\textcolor{black}{The Lagrangian density \eqref{Lagrangian_density_elasticity} for relativistic elasticity has been considered in \cite{Br2021} in the case $ \mathcal{W} $ is a function of the Lagrangian strain tensor, i.e., $ \mathcal{W} (G_0, C)= \overline{ \mathcal{W} }\big((C-G_0)/2\big)$, for $R= d \lambda  \wedge \pi _ \mathcal{B} ^*  R_0$ with $R_0= \mu (G_0)$ the volume form associated to $ G_0$ on $ \mathcal{B} $, see also \cite{DW1962}.}

The Lagrangian density \eqref{Lagrangian_density_elasticity} is spacetime covariant, which follows from the fact that $\sqrt{-  \mathsf{g}( \dot \Phi   ,  \dot \Phi  )}$ is spacetime covariant and from the identity
\[
C( j ^1 ( \psi  \circ \Phi ), \partial _ \lambda , \psi _* \mathsf{g} \circ \psi \circ \Phi )=C(j^1 \Phi , \partial _ \lambda , \mathsf{g} \circ \Phi ),
\]
for all $ \psi \in \operatorname{Diff}( \mathcal{M} )$. The associated reduced convective Lagrangian density is
\[
\mathcal{L} ( \partial _ \lambda , R, G, \Gamma ) = -  \frac{1}{c} \sqrt{-   \Gamma ( \partial _\lambda ,\partial _\lambda  )}\big( c ^2  +  \mathcal{W}(G_0, C) \big) R,
\]
where now $C$ appears as a function of $ \partial _ \lambda $ and $ \Gamma $ 
\[
C=C( \partial _ \lambda ,\Gamma ) = \Gamma - \frac{1}{ \Gamma ( \partial _ \lambda , \partial _ \lambda )} \partial _ \lambda ^{\;\flat_ \Gamma } \otimes \partial _ \lambda ^{\;\flat_ \Gamma }.
\]
Based on the expression above, one directly verifies that $ \mathscr{L}$ is \textcolor{black}{materially covariant} with respect to $ \operatorname{Diff}_{\partial _ \lambda }( \mathcal{D} )$, i.e. satisfies $ \mathscr{L}( j ^1 ( \Phi \circ \varphi ) , \partial _ \lambda , \varphi ^* R,\varphi ^* G, \mathsf{g} \circ \Phi \circ \varphi  )= \varphi ^* [ \mathscr{L}( j ^1 \Phi , \partial _ \lambda , G, \mathsf{g} \circ \Phi )]$ for all $ \varphi \in \operatorname{Diff}_{ \partial _ \lambda }( \mathcal{D})$, if and only if $ \mathcal{W}$ satisfies
\begin{equation}\label{isot_hom} 
\mathcal{W}( \psi ^* G_0, \psi ^* C)= \mathcal{W}( G_0, C) \circ \psi , \quad  \text{for all $ \psi \in \operatorname{Diff}( \mathcal{B} )$}, 
\end{equation} 
see Remark \ref{remark_W_K}, especially \eqref{isom_sdp}. This condition on $ \mathcal{W}$ means that the elastic material is homogeneous isotropic, see \cite[\S3.5]{MaHu1983}.

Under material covariance, it is possible to define the corresponding spacetime Lagrangian density, see \S\ref{subsec_mat_cov}, which is found as
\begin{equation}\label{ell_elasticity} 
\ell(w, \varrho , \mathsf{c},\mathsf{g} )= - \rho  \left( c ^2 + \varpi (\mathsf{c}, \mathsf{p} ) \right) \mu (\mathsf{g} ),
\end{equation} 
where the Eulerian version $\varpi ( \mathsf{c}, \mathsf{p} )$ of the specific stored energy is defined by
\begin{equation}\label{def_varpi} 
\varpi ( \mathsf{c}, \mathsf{p} )= \mathcal{W}( \Phi ^* \mathsf{c}, \Phi ^* \mathsf{p} ) \circ \Phi ^{-1} 
\end{equation}
for a world-tube $ \Phi : \mathcal{D} \rightarrow \mathcal{N} $ with $ \Phi _* \partial _ \lambda = w$. Note that $\varpi$ is defined only on symmetric covariant tensors $\mathsf{c}$, $\mathsf{p}$ such that their pull-back with respect to some word-tube are positive definite on $T \mathcal{B} $ and degenerate on $ \partial _ \lambda $, see \eqref{relation_C_c_p_G}. One uses \eqref{isot_hom} to check that this definition does not depend on the world-tube with $ \Phi _* \partial _ \lambda =w$.

\begin{remark}[On the need of $\mathsf{c}$ and $\mathsf{p}$]\rm One can be surprised that only the projection tensor $\mathsf{p}$, instead of the Lorentzian metric $\mathsf{g}$, is enough to intrinsically determine the invariants of $\mathsf{c}$ via $\varpi( \mathsf{c},\mathsf{p})$. This is due to the fact that both $\mathsf{c}$ and $ \mathsf{p} $ are nondegenerate on the same subspace, namely $(\operatorname{span}u) ^\perp$. Note that from \eqref{def_varpi} it follows that $\varpi$ is well-defined only on such pair $( \mathsf{c}, \mathsf{p} )$ in general. Our approach thus includes the forms of energy considered in \cite{GrEr1966} and \cite{KiMa1992}. Our approach also naturally shows that both tensors $\mathsf{c}$ and $\mathsf{p}$ are needed to express the stored energy of a relativistic material in the spacetime description, thereby confirming a remark of G\'erard Maugin, see \cite[Appendix I]{Ma1978b} about the unsoundness of some proposed formulations of relativistic elasticity based exclusively on $\mathsf{p}$.
\end{remark} 

\color{black} 
\begin{remark}[Material covariance for anisotropic elasticity]\label{MT_aniso}\rm As we have mentioned above, material covariance \eqref{isot_hom} means that the elastic material is isotropic homogeneous. In particular, $ \mathcal{W} $ admits as a symmetry group, the group of $G_0(\mathsf{X})$-orthogonal transformations, \cite{MaHu1983}. For anisotropic materials, such as transverse isotropic or orthotropic materials, only a subgroup of those transformations preserves $ \mathcal{W} $, see \cite{Li1982} and \cite{Bo1987} for instance. Such subgroups arise as isotropy groups of the so called structural (or anistotropic) tensors. We refer to \cite{ZhSp1993} and references therein for the determination of the structural tensors associated to a given symmetry group. It is clear that material covariance, as written in \eqref{isot_hom} cannot hold for anisotropic materials. However, the abstract notion of material covariance can still be achieved for anisotropic elasticity by making explicit the dependence of the stored energy function on the structural tensors, and letting the material diffeomorphisms act on them, see \cite{SiMaKr1988} and \cite{LuPa2000}. As an illustration, let us consider the case of structural tensors given by $n$ linearly independent one-forms $\{ \alpha ^K\}_{K=1}^n$ on $ \mathcal{B}$. By writing the stored energy function as $\mathcal{W}( G_0, \{\alpha ^K\} , C)$, one notes that spacetime covariance trivially holds while material covariance now reads
\begin{equation}\label{isot_hom_anisot} 
\mathcal{W}( \psi ^* G_0,  \{\psi ^* \alpha ^K\}, \psi ^* C)= \mathcal{W}( G_0, \{\alpha ^K\}, C) \circ \psi , \quad  \text{for all $ \psi \in \operatorname{Diff}( \mathcal{B} )$}.
\end{equation}
In this case, the symmetric group of $ \mathcal{W} $ does indeed correspond to the subgroup of the $G(\mathsf{X})$-orthogonal transformations which consists of those transformations that preserve the structural tensors $ \alpha ^K(\mathsf{X})$, $K=1,...,n$.

Thanks to the geometric setting developed in the present paper, these considerations directly allow our setting to cover relativistic anisotropic media, such as transverse isotropic or orthotropic. For instance, an Eulerian version of the stored energy function above can be defined in the anisotropic case, similarly to \eqref{def_varpi}, which now reads $\varpi(\mathsf{c}, \{\mathsf{f}^K\}, \mathsf{p})$, with the additional dependence on the one-forms $\mathsf{f}^K= \Phi _* \alpha ^K$. From $ \pounds _{ \partial _ \lambda } \alpha ^K=0$ and $ \mathbf{i} _{ \partial _ \lambda } \alpha ^K=0$ one gets $ \pounds _w \mathsf{f}^K=0$ and $  \mathbf{i} _w \mathsf{f}^K=0$, from which the continuity equation $ \pounds _u \mathsf{f}^K=0$ for the Eulerian structural tensors follows by using the formula $ \pounds _{fu} \mathsf{f}^K= f\pounds _u \mathsf{f}^K+ df \wedge \mathbf{i} _u \mathsf{f}^K$. We refer to Remark \ref{spacetime_anis} for further comments.

The concept of material covariance, both in the isotropic and anisotropic case, plays a crucial role for several aspects of current interest in Newtonian nonlinear elasticity, such as in the study of inverse problems in anisotropic media, \cite{MaRa2006}, or for the geometric treatment of anelasticity as developed in \cite{SoYa2020}.
\end{remark} 


\color{black} 

\paragraph{Spacetime reduced Euler-Lagrange equations.} By exploiting the setting presented above, the results stated in Theorem \ref{spacetime_reduced_EL} and \ref{convective_reduced_EL} are directly applicable to isotropic relativistic elasticity. In particular, we get the following Eulerian variational principle and spacetime reduced Euler-Lagrange equations by direct application of \eqref{Eulerian_VP} and \eqref{spacetime_EL_g}.

\begin{proposition}\label{prop_Eulerian_VP_elasticity} The Eulerian variational formulation for relativistic isotropic elasticity takes the form
\color{black} 
\begin{equation}\label{Eulerian_VP_elasticity}
\begin{aligned} 
&\!\!\left. \frac{d}{d\varepsilon}\right|_{\varepsilon=0}\int_{ \mathcal{N}_ \varepsilon } \ell\big( w_ \varepsilon , \varrho  _ \varepsilon , \mathsf{c} _ \varepsilon , \mathsf{g} \big)=0 \quad \text{for variations}\\
& \delta \mathcal{N} = \zeta |_{ \partial \mathcal{N} } \big/T \partial \mathcal{N} , \quad \delta w = -\pounds _ \zeta w, \qquad  \delta \varrho  = - \pounds _ \zeta \varrho , \qquad  \delta \mathsf{c}   = - \pounds _ \zeta \mathsf{c},\phantom{\int_A^B}
\end{aligned} 
\end{equation}\color{black} 
\textcolor{black}{where $ \zeta$} is an arbitrary vector field on $ \mathcal{N} $ \textcolor{black}{such that $ \zeta |_{ \Phi (a, \mathcal{B} )}= \zeta |_{ \Phi (b, \mathcal{B} )}=0$}.

The critical conditions associated to \eqref{Eulerian_VP_elasticity} are
\begin{equation}\label{spacetime_EL_elasticity} 
\left\{
\begin{array}{l}
\displaystyle\vspace{0.2cm}\operatorname{div}^\nabla\!\left( \left(  \ell - \varrho \frac{\partial \ell}{\partial \varrho } \right)  \delta  + w \otimes \frac{\partial \ell}{\partial w}- 2 \frac{\partial \ell}{\partial \mathsf{c}} \cdot \mathsf{c} \right)=0\\
\displaystyle \left(  \left(  \ell - \varrho \frac{\partial \ell}{\partial \varrho }  \right)  \delta  + w \otimes \frac{\partial \ell}{\partial w} - 2 \frac{\partial \ell}{\partial \mathsf{c}} \cdot \mathsf{c}\right) ( \cdot ,n ^\flat )=0\;\;\text{on $\textcolor{black}{\partial_{\rm cont} \mathcal{N}}$}
\end{array}
\right.
\end{equation}
and the variables $w$, $ \varrho $, $ \mathsf{c} $ satisfy
\begin{equation}\label{advections_elasticity} 
\pounds _w \varrho =0 \quad\text{and}\quad \pounds _w \mathsf{c}   =0.
\end{equation} 
\end{proposition}

We note that equations \eqref{advections_elasticity} are equivalently written in terms of the world-velocity as
\[
\pounds _u (\rho  \mu(\mathsf{g}))  =0   \quad\text{and}\quad \pounds _u \mathsf{c}=0.
\]
Indeed, from the formula $\pounds _{fu} \mathsf{c}= f \pounds _u \mathsf{c}+ \nabla f \otimes \mathbf{i} _u \mathsf{c}+ \mathbf{i} _u \mathsf{c} \otimes \nabla f$
and the property $\mathbf{i}_u \mathsf{c}=0$, see Lemma \ref{tensors_elasticity}, we obtain the equivalence $ \pounds _u\mathsf{c}=0 \Leftrightarrow \pounds _w \mathsf{c}=0$.

\begin{remark}[Anisotropic elasticity]\label{spacetime_anis}\rm In the anisotropic case, see Remark \ref{MT_aniso}, the Lagrangian $\ell$ also depends on the one-forms $\mathsf{f}^K$. From the general expression \eqref{spacetime_EL_g} of the reduced Euler-Lagrange equations, this results in the inclusion of the term
\[
-\sum_{K=1}^n \frac{\partial \ell}{\partial \mathsf{f}^K} \otimes \mathsf{f}^K
\]
in both equations in \eqref{spacetime_EL_elasticity}. Also, the additional the continuity equations $ \pounds _w \mathsf{f}^K=0$ \textcolor{black}{for the structural tensors}, or its equivalent formulation $\pounds _u\mathsf{f}^K=0$,  appear in \eqref{advections_elasticity}. 
\end{remark}

\paragraph{Relativistic elasticity equations.}
From the expression \eqref{ell_elasticity} of the Lagrangian density for relativistic elasticity, one directly computes the partial derivatives
\begin{align*}
\frac{\partial \ell}{\partial w }&=\frac{1}{c^2}u ^\flat ( c^2+ \varpi) \varrho  - \rho  \frac{\partial \varpi}{\partial \mathsf{p} _{ \mu \nu } }\frac{\partial \mathsf{p} _{ \mu \nu }}{\partial w} \mu (\mathsf{g}), \qquad \frac{\partial \ell}{\partial \mathsf{c}}= - \rho  \frac{\partial \varpi}{\partial \mathsf{c}} \mu (\mathsf{g}),\\
\frac{\partial \ell}{\partial \varrho }&=  - \frac{\sqrt{-\textcolor{black}{ \mathsf{g}}(w,w)}}{c}   ( c^2 +  \varpi   ).
\end{align*}  
We note that
\[
\frac{\partial \mathsf{p} _{ \mu \nu } }{\partial w ^ \lambda } = \frac{2w_ \mu w_ \nu }{\textcolor{black}{ \mathsf{g}}(w,w) ^2 }   \textcolor{black}{ \mathsf{g}}_{ \alpha \lambda  }   w ^ \alpha - \frac{1}{\textcolor{black}{ \mathsf{g}}_{ \alpha \beta } w ^ \alpha w^ \beta } (\textcolor{black}{ \mathsf{g}}_{ \mu \lambda } w_ \nu   +  \textcolor{black}{ \mathsf{g}}_{ \nu \lambda }w_ \mu  )
\]
which then gives
\[
\frac{\partial \varpi}{\partial \mathsf{p} _{ \mu \nu }} \frac{\partial \mathsf{p} _{ \mu \nu } }{\partial w ^ \lambda } =-\frac{2}{\textcolor{black}{ \mathsf{g}}(w,w)} \frac{\partial \varpi}{\partial \mathsf{p} _{ \mu \nu }} \mathsf{p} _{ \mu \lambda } w_ \nu=\frac{2}{\textcolor{black}{ \mathsf{g}}(w,w)} \frac{\partial \varpi}{\partial \mathsf{c} _{ \mu \nu }} \mathsf{c} _{ \mu \lambda } w_ \nu ,
\]
where we used the spacetime covariance of $ \varpi$, i.e., $\varpi( \varphi ^* \mathsf{c}, \varphi ^* \mathsf{p} ) = \varpi( \mathsf{c}, \mathsf{p} ) \circ \varphi$ , for all $\varphi \in \operatorname{Diff}( \mathcal{M} )$, which yields
\[
\frac{\partial \varpi}{\partial \mathsf{c}} \cdot \mathsf{c} +  \frac{\partial \varpi}{\partial \mathsf{p} } \cdot \mathsf{p}  =0.
\]
From this, one gets the \textcolor{black}{stress-energy-momentum tensor density} for elasticity as
\[
\mathfrak{T}_{\rm el}=\left( \ell- \frac{\partial \ell}{\partial \varrho } \varrho \right) \delta  + w \otimes \frac{\partial \ell}{\partial w} - 2 \frac{\partial \ell}{\partial \mathsf{c}} \cdot \mathsf{c}= \big(  \epsilon _{\rm tot}\frac{1}{c ^2 } u \otimes    u ^\flat - \mathfrak{t}_{\rm el} \big) \mu (\mathsf{g}),
\]
where $ \epsilon _{\rm tot}=  \rho  ( c ^2 + \varpi( \mathsf{c}, \mathsf{p})) $ is the total energy density and with the relativistic stress tensor
\[
\mathfrak{t}_{\rm el}= - 2 \rho\,\mathsf{P} \!\cdot \!\frac{\partial \varpi}{\partial \mathsf{c}} \!\cdot\! \mathsf{c}, \qquad (\mathfrak{t}_{\rm el})^ \mu _ \nu  = - 2 \rho\,\mathsf{P}^ \mu _ \alpha    \frac{\partial \varpi}{\partial \mathsf{c}_{\lambda \alpha }}  \mathsf{c}_{ \lambda \nu } .
\]

The first equation in \eqref{spacetime_EL_elasticity} yields the relativistic Euler-Cauchy equation and energy equation
\begin{equation}\label{relativistic_Euler_Cauchy} 
\frac{1}{c^2} (  \epsilon _{\rm tot}   \nabla _ u u + u \;\mathfrak{t}_{\rm el}\!:\!\nabla u ) = \operatorname{div}\mathfrak{t} _{\rm el} \qquad\text{and}\qquad \operatorname{div}(\epsilon _{\rm tot}u) = \mathfrak{t}_{\rm el}\!:\!\nabla u
\end{equation} 
while the second one gives the zero traction boundary conditions
\begin{equation}\label{vaccum_Euler_Cauchy}
\mathfrak{t}_{\rm el}( \cdot , n^\flat) =0\quad\text{on}\quad \textcolor{black}{ \partial_{\rm cont} \mathcal{N} =0},
\end{equation}
which follows from $\mathfrak{T}_{\rm el}( \cdot , n^\flat)= \mathfrak{t}_{\rm el}( \cdot , n^\flat)\textcolor{black}{\mu (\mathsf{g})}$ on $\textcolor{black}{ \partial_{\rm cont} \mathcal{N} }$. We have $\mathsf{g}(u, n)=0$ on $\textcolor{black}{\partial _{\rm cont}\mathcal{N}}$ as before. 

\color{black}We refer to \cite{GrEr1966} for the derivation of the stress-energy-momentum tensor and equations \eqref{relativistic_Euler_Cauchy} for special relativistic elasticity from the point of view of relativistic balance laws and constitutive equations. We are not aware of such a standard derivation for the general setting we are treating here based on a specific stored energy $\varpi(\mathsf{c},\mathsf{p})$ involving both $\mathsf{c}$ and $\mathsf{p}$, similarly for the anisotropic case based on $\varpi(\mathsf{c}, \{\mathsf{f}^K\}, \mathsf{p})$, see Remarks \ref{MT_aniso} and \ref{spacetime_anis}. 

\textcolor{black}{Regarding the variational derivation, the closest approach seems to be that developed in \cite{Br2021} in which the equations for relativistic elasticity in the Eulerian description were obtained from the Lagrangian density \eqref{Lagrangian_density_elasticity} in the material description, for the case $ \mathcal{W} (G_0, C)= \overline{ \mathcal{W} }\big((C-G_0)/2\big)$.}
\color{black}

\paragraph{Coupling with the Einstein equations and junction conditions.} In a similar way with the fluid case, the extension of the variational formulation to the coupling with general relativity can be obtained by particularizing Theorem \ref{main} to the Lagrangian density \eqref{ell_elasticity} of elasticity. One obtains the Einstein equations on $ \mathcal{N} ^\pm$, the equations \eqref{relativistic_Euler} on $ \mathcal{N} ^-$, as well as the Israel-Darmois junction conditions on $ \partial \mathcal{N} $. On the boundary one gets
\[
[h]=[K]=0 \;\Longrightarrow \; \mathfrak{t}_{\rm el}( \cdot , n^\flat) =0\quad\text{on}\quad \textcolor{black}{ \partial_{\rm cont} \mathcal{N} }.
\]

\subsection{General relativistic continua}

The variational approach presented above for fluid and elasticity can be easily extended to more general continua with internal energy function $ \mathcal{W}( \rho  , \eta , G_0, C)$, which thus covers both the fluid and elasticity cases. We very briefly describe this situation below. For such continua the material Lagrangian takes the form
\begin{equation}\label{Lagrangian_density_general}
\begin{aligned} 
&\mathscr{L}(j^1 \Phi , \partial _ \lambda  , R, S, G, \mathsf{g} \circ \Phi )\\
&= - \frac{1}{c}  \sqrt{-  \mathsf{g}( \dot \Phi   ,  \dot \Phi  )}\Bigg( c ^2  +  \mathcal{W}  \bigg( \frac{1}{c} \sqrt{-  \mathsf{g}( \dot \Phi   ,  \dot \Phi  )}\frac{R }{\Phi ^* [\mu (\mathsf{g})]} ,\frac{ S }{R} , G_0, C\bigg) \Bigg) R.
\end{aligned}
\end{equation} 
This Lagrangian density is spacetime covariant. It is \textcolor{black}{materially covariant} if and only if the continua is isotropic homogeneous,
\begin{equation}\label{isot_hom_general} 
\mathcal{W}(\rho  \circ \psi , \eta \circ \psi ,  \psi ^* G_0, \psi ^* C)= \mathcal{W}(\rho  , \eta , G_0, C) \circ \psi , \quad  \text{for all $ \psi \in \operatorname{Diff}( \mathcal{B} )$}.
\end{equation} 
Material covariance can be also achieved for the anisotropic case by following Remarks \ref{MT_aniso} and \ref{spacetime_anis}, and expressing the internal energy as a function of the form $ \mathcal{W}( \rho  , \eta , G_0, \{ \alpha ^K\}, C)$. A detailed study of this situation will be pursued elsewhere .

Assuming isotropy for simplicity, the spacetime Lagrangian density reads
\[
\ell(w, \varrho ,\varsigma , \mathsf{c}, \mathsf{g})= - \rho  \left( c ^2 + \varpi (\rho  , \eta ,\mathsf{c}, \mathsf{p} ) \right) \mu (\mathsf{g})
\]
where we defined $\varpi ( \rho  , \eta , \mathsf{c}, \mathsf{p} )= \mathcal{W}( \rho  \circ \Phi , \eta \circ \Phi , \Phi ^* \mathsf{c}, \Phi ^* \mathsf{p} ) \circ \Phi ^{-1}$, as in \eqref{def_varpi}.
Propositions \ref{prop_Eulerian_VP_fluid} and \ref{prop_Eulerian_VP_elasticity} generalize to this case, yielding the spacetime reduced Euler-Lagrange equations in the general form
\begin{equation}\label{spacetime_EL_continua} 
\left\{
\begin{array}{l}
\displaystyle\vspace{0.2cm}\operatorname{div}^ \nabla \!\left( \left(  \ell - \varrho \frac{\partial \ell}{\partial \varrho }- \varsigma  \frac{\partial \ell}{\partial \varsigma  } \right)  \delta  + w \otimes \frac{\partial \ell}{\partial w}- 2 \frac{\partial \ell}{\partial \mathsf{c}} \cdot \mathsf{c} \right)=0\\
\displaystyle \left(  \left(  \ell - \varrho \frac{\partial \ell}{\partial \varrho }- \varsigma  \frac{\partial \ell}{\partial \varsigma  } \right)  \delta  + w \otimes \frac{\partial \ell}{\partial w} - 2 \frac{\partial \ell}{\partial \mathsf{c}} \cdot \mathsf{c}\right) ( \cdot ,n ^\flat )=0\;\;\text{on $\textcolor{black}{ \partial_{\rm cont} \mathcal{N} }$}
\end{array}
\right.
\end{equation}
and $\pounds _w \varrho =0$, $\pounds _w \varsigma  =0$,  and $\pounds _w \mathsf{c}=0$.

From this, one computes the \textcolor{black}{stress-energy-momentum tensor density} as
\begin{align*} 
\mathfrak{T}_{\rm cont}&=\left( \ell- \frac{\partial \ell}{\partial \varrho } \varrho - \varsigma  \frac{\partial \ell}{\partial \varsigma  }\right) \delta  + w \otimes \frac{\partial \ell}{\partial w} - 2 \frac{\partial \ell}{\partial \mathsf{c}} \cdot \mathsf{c}= \big(  \epsilon _{\rm tot}\frac{1}{c ^2 } u \otimes    u ^\flat - \mathfrak{t} \big) \mu (\mathsf{g}),
\end{align*} 
where $ \epsilon _{\rm tot}=  \rho  ( c ^2 + \varpi( \rho, \eta ,\mathsf{c}, \mathsf{p})) $ is the total energy density and with the relativistic stress tensor
\[
\mathfrak{t}= - \mathsf{P} \!\cdot \! \left( \!\rho  ^2 \frac{\partial \varpi}{\partial \rho  } \, \delta + 2 \rho\,\frac{\partial \varpi}{\partial \mathsf{c}} \!\cdot\! \mathsf{c}\! \right) =  - \mathsf{P} p+ \mathfrak{t}_{\rm el} , \quad \mathfrak{t}^ \mu _ \nu  = - \mathsf{P}^ \mu _ \alpha    \left( \! \rho  ^2 \frac{\partial \varpi}{\partial \rho  } \,  \delta ^ \alpha _ \nu + 2 \rho \frac{\partial \varpi}{\partial \mathsf{c}_{\lambda \alpha }}  \mathsf{c}_{ \lambda \nu }\! \right),
\]
written in terms of the pressure $p$ and relativistic elastic stress tensor $\mathfrak{t}_{\rm el}$.

The first equation in \eqref{spacetime_EL_continua} yields the relativistic Euler-Cauchy equation and energy equation for the relativistic continuum as
\[
\frac{1}{c^2} (  (\epsilon _{\rm tot} +p)  \nabla _ u u + u \;\mathfrak{t}_{\rm el}\!:\!\nabla u ) = -\mathsf{P} \nabla p + \operatorname{div}\mathfrak{t} _{\rm el} \;\;\text{and}\;\; \operatorname{div}(\epsilon _{\rm tot}u) + p\operatorname{div}u = \mathfrak{t}_{\rm el}\!:\!\nabla u
\]
while the second one gives the zero traction boundary conditions
\begin{equation}\label{vaccum_Euler_Cauchy_cont}
\mathfrak{t}_{\rm el}( \cdot , n^\flat) =p\, n^ \flat \quad\text{on}\quad \textcolor{black}{ \partial_{\rm cont} \mathcal{N} }.
\end{equation}

The coupling with general relativity follows again from Theorem \ref{main} and on the boundary one gets
\[
[h]=[K]=0 \;\Longrightarrow \; \mathfrak{t}_{\rm el}( \cdot , n^\flat) =p\, n^ \flat\quad\text{on}\quad \textcolor{black}{ \partial_{\rm cont} \mathcal{N} }.
\]

We give below two tables that summarize the analogies between the variational framework for nonrelativistic and relativistic Lagrangian continuum theories. We focus on the spacetime reduced description. The convective reduced description can be described similarly. 

\begin{table}[h!]
\centering
\begin{tabular}{|c | c |} 
\hline
\textsf{Nonrelativistic} & \textsf{Relativistic}\\ [0.25ex] 
\hline
\hline
Material tensor fields & Material tensor fields \\[0.25ex]
\& associated spatial versions & \& associated spacetime versions \\[0.25ex]
\hline 
$\begin{array}{c}
R  \\
S\\
G
\end{array}\xrightarrow{ \;\;\varphi _*\;\;}
\begin{array}{c}
\varrho \\
\varsigma  \\
\mathsf{c}
\end{array}$
&
$\begin{array}{l}
\partial _ \lambda \\
R\\
S\\
G
\end{array}
\xrightarrow{ \;\;\Phi _*\;\;}
\begin{array}{c}
w\\
\varrho \\
\varsigma  \\
\mathsf{c}
\end{array}$
\\
[0.25ex] 
\hline
\hline
Spatial tensor field & Spacetime tensor field \\[0.25ex]
\& associated material version & \& associated material version \\[0.25ex]
\hline 
$ C\;\; \xleftarrow{\;\; \varphi ^*\;\;} \;\; g$ &
$ C\;\; \xleftarrow{ \;\;\Phi  ^*\;\;} \;\; \mathsf{p}=\mathsf{p}(\mathsf{g}, w)$
\\[0.25ex] 
\hline\hline
Spatial covariance & Spacetime covariance \\[0.25ex]
\hline 
$
\begin{array}{c}
\phantom{\frac{\int_A}{\int} \!\!\!\!}\leadsto \mathscr{L}(\dot \varphi , T \varphi , R, S, G,g\circ \varphi )\phantom{\int^A\!\!\!\! }\\
\phantom{\frac{\int^A}{\int} \!\!\!\!}=\frac{1}{2} g( \dot \varphi , \dot \varphi ) R - \mathcal{W}\big(\textcolor{black}{  \frac{R}{\mu (C)} , \frac{S}{R}  , G}, C\big)R\vspace{0.2cm}
\end{array}
$
&
$
\begin{array}{c}
\phantom{\frac{\int_A}{\int} \!\!\!\!}\leadsto \mathscr{L}( j^1 \Phi , \partial _ \lambda ,R, S, G, \mathsf{g} \circ \Phi  )\phantom{\frac{\int_A}{\int} \!\!\!\!}\\
\phantom{\frac{\int_A}{\int}\!\!\!\!\!\!\!}= -  \frac{\sqrt{-  \mathsf{g}( \dot \Phi   ,  \dot \Phi  )}}{c} \big( c ^2  +  \mathcal{W}  \big(\textcolor{black}{  \frac{\sqrt{-  \mathsf{g}( \dot \Phi   ,  \dot \Phi  )}R }{ c \Phi ^* [\mu (\mathsf{g})]}   , \frac{S}{R} }, G_0, C\big) \big) R\vspace{0.2cm}
\end{array}$
\\[0.25ex] 
\hline\hline
Material covariance & Material covariance\\ [0.5ex] 
\hline
$\begin{array}{cc}
&\phantom{\int^A \!\!\!\!}\leadsto  \mathscr{L}( \dot \varphi , T \varphi , R, S, G, g \circ \varphi )\phantom{\int^A \!\!\!\!}\\
&=\varphi ^* [\ell(u, \varrho , \varsigma , \mathsf{c}, g)]\vspace{0.2cm}
\end{array}$
&
$\begin{array}{cc}
&\phantom{\int^A \!\!\!\!}\leadsto  \mathscr{L}( j^1 \Phi , \partial _ \lambda  , R, S, G, \mathsf{g} \circ \varphi ) \phantom{\int^A \!\!\!\!}\\
&= \Phi  ^* [\ell( w, \varrho , \varsigma , \mathsf{c}, \mathsf{g})]\vspace{0.2cm}
\end{array}$
\\
with & with \\
$
\begin{array}{c}
\phantom{\int^A \!\!\!\!}\ell(u, \varrho   , \varsigma , \mathsf{c},g)\phantom{\int^A \!\!\!\!}\\
=   \frac{1}{2} g(u,u) \varrho - \varpi( \rho  , \eta , \mathsf{c}, g) \varrho\vspace{0.2cm}
\end{array}$
&
$
\begin{array}{c}
\phantom{\int^A \!\!\!\!}\ell(w, \varrho ,\varsigma , \mathsf{c}, \mathsf{g})\phantom{\int^A \!\!\!\!}\\
= -  \frac{\sqrt{-\mathsf{g}(w,w)}}{c }   \left( c ^2 + \varpi (\rho  , \eta ,\mathsf{c}, \mathsf{p} ) \right) \varrho\vspace{0.2cm}
\end{array}$
\\ [0.25ex] 
\hline
\end{tabular}
\caption{Schematic correspondence between relativistic and nonrelativistic Lagrangian continuum theories.}
\label{comparison}
\end{table}

\begin{table}[h!]
\centering
\begin{tabular}{|c | c |} 
\hline
& \textsf{Equations of motion for the continuum} \\ [0.25ex] 
\hline
\textsf{Nonrelativistic}  & $
\begin{array}{c}
\partial _t \frac{\partial \ell}{\partial u} + \operatorname{div}^ \nabla \left(  \left(  \ell - \varrho \frac{\partial \ell}{\partial \varrho }- \varsigma  \frac{\partial \ell}{\partial \varsigma  } \right)  \delta  + u \otimes \frac{\partial \ell}{\partial u} - 2 \frac{\partial \ell}{\partial \mathsf{c}} \cdot \mathsf{c} \right) =0\\
\vspace{0.1cm}\partial _t \varrho + \pounds  _u \varrho =0,\quad  \partial _t \varsigma  + \pounds  _u \varsigma  =0, \quad\partial _t \mathsf{c} + \pounds  _u  \mathsf{c}  =0
\end{array}
$ \\
\hline 
\textsf{Relativistic} & $
\begin{array}{c}
\operatorname{div}^ \nabla \!\left( \left(  \ell - \varrho \frac{\partial \ell}{\partial \varrho }- \varsigma  \frac{\partial \ell}{\partial \varsigma  } \right)  \delta  + w \otimes \frac{\partial \ell}{\partial w}- 2 \frac{\partial \ell}{\partial \mathsf{c}} \cdot \mathsf{c} \right)=0\\
\vspace{0.1cm}\pounds  _w \varrho =0,\qquad   \pounds  _w \varsigma  =0, \qquad \pounds  _w  \mathsf{c}  =0
\end{array}
$ \\
\hline
\end{tabular} \qquad  
\caption{Relativistic and nonrelativistic spacetime reduced Euler-Lagrange equations for continuum theories.}
\label{comparison_equations}
\end{table}

Despite the apparent analogies in these tables, one should not forget crucial differences between the two situations. For instance (choosing $n=3$) the push-forward and pull-back operations $ \varphi _*$ and $ \varphi ^* $ act on tensor fields defined on the three dimensional manifolds $ \mathcal{B} $ and $ \mathcal{S} $, at each fixed time $t$, while $ \Phi _*$ and $ \Phi ^* $ act on tensor fields defined on the four dimensional manifolds $ \mathcal{D} =[a,b] \times \mathcal{B} $ and $ \mathcal{M} $. Similarly, the divergence and Lie derivative operators act on objects defined on three dimensional manifolds, at each fixed time $t$ for the nonrelativistic case, while they act on objects defined on four dimensional manifolds in the relativistic case. Note that we have kept the first jet notation $j^1 \Phi $ to encode the first derivatives of the world-tube $ \Phi : \mathcal{D} \rightarrow \mathcal{M} $ for the relativistic case. While for $ \mathcal{D} =[a,b] \times \mathcal{B}\ni X=( \lambda , \mathsf{X}) $ we could also write it as the couple $( \partial _ \lambda\Phi  , T \Phi )$ with $T \Phi $ the tangent map to $\Phi ( \lambda , \cdot ):  \mathcal{B} \rightarrow \mathcal{M} $ in analogy with the non-relativistic notation $(\dot \varphi , T \varphi )$, we prefer to use the first jet notation since $ \lambda $ is not treated differently from $\mathsf{X}$ when passing from a picture to another, as opposed to the time $t$ of the Newtonian case.

\section{Conclusion}

We have established a Lagrangian variational framework for general relativistic continuum theories that systematically parallels the Lagrangian reduction approach to nonrelativistic fluids and elasticity. Using spacetime and material covariance properties, our approach allows to rigorously deduce variational formulations in the convective and Eulerian descriptions by starting from the most natural and physically justified principle, namely, a continuum version of the Hamilton principle for the relativistic particle. In particular, our formulation does not need the inclusion of constraints or unphysical variables in the action functional.

We showed how this setting can be extended to the coupling of the continuum dynamics with gravitation theory by including the Gibbons-Hawking-York term associated to the boundary of the relativistic continuum. In particular, the critical condition of the resulting variational principle gives the Einstein equations for the gravitational field created by the relativistic continuum, both at the interior and outside the continuum; the equations of motion of the continuum in this gravitational field; the junction conditions between the solution at the interior of the relativistic continuum and the solution describing the gravity field produced outside from it.
The variation of the Gibbons-Hawking-York term term with respect to the boundary has necessitated some extensions of previous results on the first variation of mean curvature integrals.

The setting was then applied to relativistic fluids and relativistic elasticity by focusing on appropriate choices of reference tensor fields and Lagrangian densities. For elasticity, the variational setting also allowed to clarify the relation between formulations based on the relativistic right Cauchy-Green tensor and the metric $G$ on one hand, or based on the relativistic Cauchy deformation tensor and the projection tensor, \textcolor{black}{or radar metric}, on the other hand. The resulting Israel-Darmois junction conditions were shown to imply the zero traction type boundary conditions for the continuum, via the O'Brien-Synge conditions.  The setting developed here offers new tools for the modelling of general relativistic continuum theories by variational principles, which will be further exploited in subsequent parts of the paper.

{\footnotesize

\bibliographystyle{new}

}

\end{document}